\documentclass[11pt,letterpaper]{article}
\usepackage[lmargin=1.0in,rmargin=1.0in,bottom=1.0in,top=1.0in,twoside=False]{geometry}

\usepackage{fullpage,amssymb,amsmath}
\usepackage{graphicx}%,algorithmic,algorithm}%, verbatim}
\usepackage{enumerate}
\usepackage{tikz}
\usetikzlibrary{shapes}
\usetikzlibrary{calc} % For tikzhypergraph
\usetikzlibrary{graphs}
\usetikzlibrary{positioning}
\usepackage{ifthen} % For tikzhypergraph

\usepackage[T1]{fontenc}
\usepackage{marginnote}

 \usepackage{xcolor}
 \usepackage{mathtools}
\usepackage{microtype}
\usepackage{amsfonts}
\usepackage{comment}
\usepackage[english]{babel}
\usepackage{mathrsfs}
\usepackage[linesnumbered,ruled,vlined]{algorithm2e}

\usepackage{trimspaces}
\usepackage{nccfoots}
\usepackage{setspace}
\usepackage{inconsolata}
\usepackage{libertine}
\usepackage[absolute]{textpos}

\usepackage{enumitem}
\usepackage{todonotes}

\usepackage{longtable}

\usepackage{theoremref}

\definecolor{blue}{rgb}{0.1,0.2,0.5}
\definecolor{brown}{rgb}{0.6,0.6,0.2}
\usepackage[ocgcolorlinks, linkcolor={blue}, citecolor={brown}]{hyperref}
\usepackage{comment}

\usepackage[amsmath,thmmarks,hyperref]{ntheorem}
\usepackage{enumerate}

\usepackage{latexsym}

\theoremnumbering{arabic}
\theoremstyle{plain}
\theoremsymbol{}
\theorembodyfont{\itshape}
\theoremheaderfont{\normalfont\bfseries}
\theoremseparator{.}

\newtheorem{theorem}{Theorem}[section]
\newcommand{\newtheoremwithcrefformat}[2]{%
  \newtheorem{#1}[theorem]{#2}%
}
\newcommand{\newseptheoremwithcrefformat}[2]{%
  \newtheorem{#1}{#2}[section]%
}

\newseptheoremwithcrefformat{lemma}{Lemma}
\newtheoremwithcrefformat{proposition}{Proposition}
\newtheoremwithcrefformat{observation}{Observation}
\newtheoremwithcrefformat{conjecture}{Conjecture}
\newtheoremwithcrefformat{corollary}{Corollary}
\newseptheoremwithcrefformat{claim}{Claim}
\theorembodyfont{\upshape}
\newtheoremwithcrefformat{example}{Example}
\newtheoremwithcrefformat{remark}{Remark}
\newseptheoremwithcrefformat{definition}{Definition}
\newseptheoremwithcrefformat{question}{Question}

\theoremstyle{nonumberplain}
\theoremheaderfont{\scshape}
\theorembodyfont{\normalfont}
\theoremsymbol{\ensuremath{\square}}
\newtheorem{proof}{Proof}

\theoremsymbol{\ensuremath{\lrcorner}}
\newtheorem{clproof}{Proof}

\def\cqedsymbol{\ifmmode$\lrcorner$\else{\unskip\nobreak\hfil
\penalty50\hskip1em\null\nobreak\hfil$\lrcorner$
\parfillskip=0pt\finalhyphendemerits=0\endgraf}\fi}

\tikzset{
    position/.style args={#1:#2 from #3}{
        at=(#3.#1), anchor=#1+180, shift=(#1:#2)
    }
}

\renewenvironment{quote}{%
   \list{}{%
     \leftmargin0.0cm   % this is the adjusting screw
     \rightmargin\leftmargin
   }
   \item\relax
}
{\endlist}

\newcommand{\N}{\mathbb{N}}

\newcommand{\Cc}{\mathcal{C}}

\newcommand{\FO}{\mathsf{FO}}
\newcommand{\MSO}{\mathsf{MSO}}

\newcommand{\Oh}{\mathcal{O}}

\newcommand{\tdict}{\delta_{\mathrm{t}}}
\newcommand{\mdict}{\delta_{\mathrm{s}}}

\newcommand{\Univ}{U}

\newcommand{\glob}{\mathsf{glo}}
\newcommand{\lowe}{\mathsf{lower}}
\let\originalleft\left
\let\originalright\right
\renewcommand{\left}{\mathopen{}\mathclose\bgroup\originalleft}
\renewcommand{\right}{\aftergroup\egroup\originalright}

\renewcommand{\leq}{\leqslant}
\renewcommand{\geq}{\geqslant}

\renewcommand{\setminus}{-}

\newcommand{\Ff}{\mathcal{F}}

\newcommand{\Tt}{\mathcal{T}}

\newcommand{\Dt}{\mathbb{D}}
\newcommand{\Topt}{\mathbb{T}}
\newcommand{\Dict}{\mathbb{L}}
\newcommand{\Tdep}{\mathbb{R}}

\newcommand{\nil}{\mathbf{nil}}
\newcommand{\true}{\mathbf{true}}
\newcommand{\false}{\mathbf{false}}

\newcommand{\verts}{\mathbb{V}}
\newcommand{\edges}{\mathbb{E}}

\newcommand{\ins}{\mathtt{insert}}
\newcommand{\rem}{\mathtt{remove}}

\newcommand{\update}{\mathtt{update}}
\newcommand{\elimination}{\mathtt{elimination}}
\newcommand{\bicom}{\mathtt{bicomponents}}

\newcommand{\pathf}{\mathtt{path}}
\newcommand{\pathlf}{\mathtt{pathlen}}
\newcommand{\pathlb}{\mathtt{pathlb}}
\newcommand{\pathub}{\mathtt{pathub}}
\newcommand{\edge}{\mathtt{edge}}
\newcommand{\same}{\mathtt{same}}
\newcommand{\cpy}{\mathtt{copy}}
\newcommand{\retrieve}{\mathtt{retrieve}}

\newcommand{\link}{\mathtt{link}}
\newcommand{\cut}{\mathtt{cut}}
\newcommand{\expose}{\mathtt{expose}}
\newcommand{\merge}{\mathtt{merge}}
\newcommand{\splitt}{\mathtt{split}}
\newcommand{\find}{\mathtt{find}}

\newcommand{\art}{\mathtt{articul}}
\newcommand{\new}{\mathtt{new}}
\newcommand{\bridge}{\mathtt{bridge}}
\newcommand{\destroy}{\mathtt{destroy}}

\newcommand{\core}{\mathtt{core}}
\newcommand{\trim}{\mathtt{trim}}
\newcommand{\extend}{\mathtt{extend}}
\newcommand{\member}{\mathtt{member}}
\newcommand{\lookup}{\mathtt{lookup}}
\newcommand{\flush}{\mathtt{flush}}
\newcommand{\Lt}{\mathbb{L}}

\newcommand{\gl}{\mathrm{glo}}
\newcommand{\loc}{\mathrm{loc}}
\newcommand{\Prt}{\mathcal{N}}

\newcommand{\LPath}{{\sc{Long Path}}\xspace}
\newcommand{\LCycle}{{\sc{Long Cycle}}\xspace}

\newcommand{\pnt}{\mathsf{parent}}
\newcommand{\chld}{\mathsf{children}}
\newcommand{\anc}{\mathsf{anc}}
\newcommand{\desc}{\mathsf{desc}}
\newcommand{\depth}{\mathsf{depth}}
\newcommand{\height}{\mathsf{height}}
\newcommand{\roots}{\mathsf{roots}}
\newcommand{\SReach}{\mathsf{SReach}}
\newcommand{\App}{\mathsf{App}}
\newcommand{\td}{\mathsf{td}}
\newcommand{\wh}[1]{\widehat{#1}}

\newcommand{\Up}{\mathsf{NeiUp}}
\newcommand{\bucket}{\mathsf{B}}
\newcommand{\parent}{\mathsf{p}}
\newcommand{\tp}{\mathsf{toParent}}

\newcommand{\dw}[2]{#1_{\downarrow \tau}}
\newcommand{\type}{\mathrm{tp}}
\newcommand{\Types}{\mathrm{St}}
\newcommand{\Conf}{\mathfrak{C}}
\newcommand{\conf}{\mathsf{conf}}
\newcommand{\forget}{\mathsf{forget}}
\newcommand{\Mm}{\mathcal{M}}
\newcommand{\join}{\mathsf{join}}
\newcommand{\degree}{\mathsf{deg}}

\reversemarginpar
\newcommand{\CycleNote}{%
  \marginpar{\vspace{-0.85cm}\hspace{1.5cm}$\bigcirc$}
}
\newcommand{\CycleNoteSub}{%
  \marginpar{\vspace{-0.65cm}\hspace{1.5cm}$\bigcirc$}
}

\definecolor{orange}{RGB}{255,127,0}
\newcommand{\boldall}[1]{\ifmmode\mathbf{#1}\else\textbf{\boldmath{#1}}\fi}

\begin{document}

\title{Efficient fully dynamic elimination forests with applications to detecting long paths and cycles\thanks{This work is a part of projects that have received funding from the European Research Council (ERC) under the European Union's Horizon 2020 research and innovation programme:
Grant Agreements no.~714704 (W.~Nadara, Ma.~Pilipczuk, M.~Sorge) and no.~677651 (Mi.~Pilipczuk). This work is also the result of research conducted within research project number 2017/26/D/ST6/00264 financed by National Science Centre (Anna Zych-Pawlewicz). Andreas Emil Feldmann was
supported by the Czech Science Foundation GA\v{C}R (grant \#17-10090Y), and by 
the Center for Foundations of Modern Computer Science (Charles Univ. project 
UNCE/SCI/004).}}

\author{
    Jiehua Chen\thanks{TU Wien, Austria, \texttt{jiehua.chen@tuwien.ac.at}.}
    \and
    Wojciech Czerwi\'nski\thanks{University of Warsaw, Poland, \texttt{\{wczerwin,w.nadara,\{marcin,michal\}.pilipczuk,manuel.sorge,} \texttt{anka\}@mimuw.edu.pl}, \texttt{bw371883@students.mimuw.edu.pl}}
    \and
    Yann Disser\thanks{TU Darmstadt, Germany, \texttt{disser@mathematik.tu-darmstadt.de}.}
    \and
    Andreas Emil Feldmann\thanks{Charles University in Prague, Czechia, \texttt{feldmann.a.e@gmail.com}.}
    \and
    Danny Hermelin\thanks{Ben-Gurion University of the Negev, Israel, \texttt{hermelin@bgu.ac.il}.}
    \and
    Wojciech Nadara$^\ddagger$
    \and
    Marcin Pilipczuk$^\ddagger$
    \and
    Micha\l{} Pilipczuk$^\ddagger$
    \and
    Manuel Sorge$^\ddagger$
    \and
    Bart\l{}omiej Wr\'oblewski$^\ddagger$
    \and
    Anna Zych-Pawlewicz$^\ddagger$
}

\begin{titlepage}
\def\thepage{}
\thispagestyle{empty}
\maketitle

\begin{textblock}{20}(0, 11.9)
\includegraphics[width=40px]{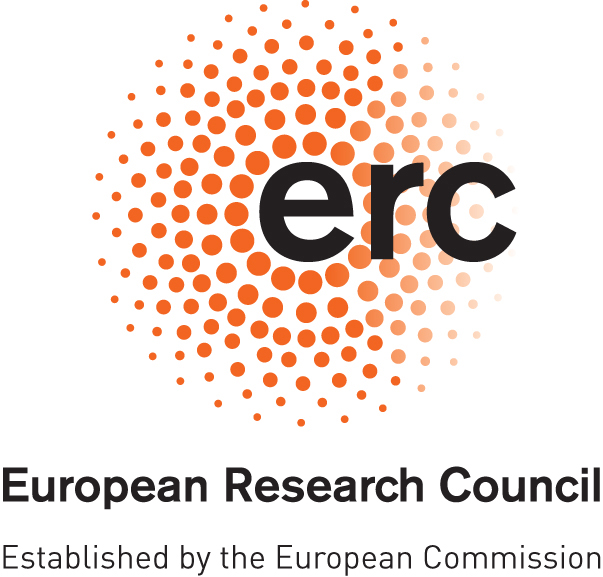}%
\end{textblock}
\begin{textblock}{20}(-0.25, 12.3)
\includegraphics[width=60px]{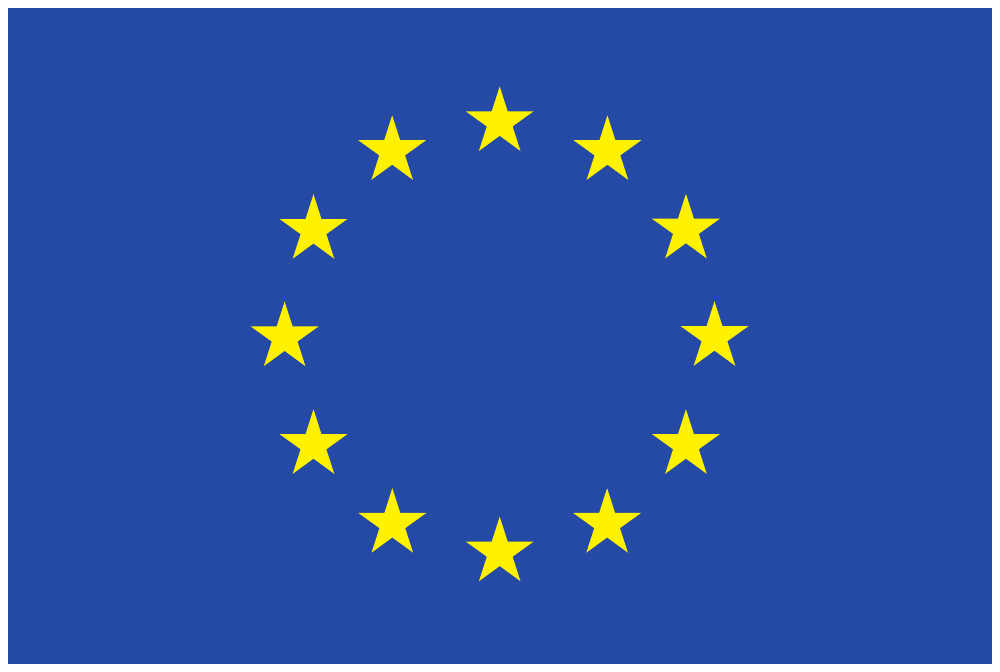}%
\end{textblock}

\begin{abstract}
 We present a data structure that in a dynamic graph of treedepth at most $d$, which is modified over time by edge insertions and deletions, maintains an optimum-height elimination forest. The data structure achieves worst-case update time $2^{\Oh(d^2)}$, which matches the best known parameter dependency in the running time of a static fpt algorithm for computing the treedepth of a graph.
 This improves a result of Dvo\v{r}\'ak et al.~[ESA 2014], who for the same problem achieved update time $f(d)$ for some non-elementary (i.e. tower-exponential) function~$f$. As a by-product, we improve known upper bounds on the sizes of minimal obstructions for having treedepth $d$ from doubly-exponential in $d$ to~$d^{\Oh(d)}$.
 
 As applications, we design new fully dynamic parameterized data structures for detecting long paths and cycles in general graphs. More precisely, for a fixed parameter $k$ and a dynamic graph~$G$, modified over time by edge insertions and deletions, our data structures maintain answers to the following~queries:
 \begin{itemize}[nosep]
  \item Does $G$ contain a simple path on $k$ vertices?
  \item Does $G$ contain a simple cycle on at least $k$ vertices?
 \end{itemize}
 In the first case, the data structure achieves amortized update time $2^{\Oh(k^2)}$. In the second case, the amortized update time is $2^{\Oh(k^4)} + \Oh(k \log n)$. In both cases we assume access to a dictionary on the edges of $G$.
\end{abstract}

\end{titlepage}

\section{Introduction}\label{sec:introduction}

In this paper we work with dynamic data structures for graph problems. The usual setting is as follows. We are given a graph $G$ that has an invariant vertex set, but is modified over time by edge insertions and edge deletions. The goal is to design a data structure that efficiently maintains $G$ under such modifications, while supporting queries about some properties of interest in $G$. We would like to optimize the worst-case or amortized guarantees on both the update time and the query time offered by the data structure.

Classically, the research on dynamic data structures concentrates on problems that, in the static setting, are polynomial-time solvable, such as testing connectivity and maintaining minimum weight spanning trees~\cite{RauchKing,HLT01,10.5555/2627817.2627943,Huang,Fre1,10.5555/2627817.2627898, Rasmussen16,NanongkaiS17,CWN17,8104124,EppsteinGIS96}, testing higher connectivity~\cite{RauchKing,HLT01,10.5555/3174304.3175269,Fre2,EppsteinGIS96,jin2020fully}, maintaining maximum matchings~\cite{DBLP:conf/soda/Sankowski07,GP13,10.1145/2897518.2897568,10.5555/3039686.3039716}, testing planarity~\cite{EppsteinGIS96,plan_test} or maintaining the distance matrix of the graph~\cite{10.1145/1039488.1039492,10.1145/1060590.1060607,ChechikKrinninger,doi:10.1137/1.9781611975994.156}. In this work we study problems that in the classic sense are ${\mathsf{NP}}$-hard, but are considered tractable from the point of view of {\em{parameterized complexity}}. In this paradigm, the usual goal is to design a {\em{fixed-parameter tractable}} (fpt) algorithm with running time of the form $f(k)\cdot n^{\Oh(1)}$, where $n$ is the total input size and $k$ is a {\em{parameter}} --- an auxiliary quantitative measure of the hardness of an instance. As function $f$ is allowed to be super-polynomial, this enables us to confine the combinatorial explosion, (seemingly) inherent in all ${\mathsf{NP}}$-hard problems, to the specific parameter under study. 

The idea of parameterized measurement of complexity can be also applied to dynamic data structures: just use auxiliary parameters in upper bounds on the update and query times, allowing exponential dependence in case the considered problem is ${\mathsf{NP}}$-hard in the classic sense. Despite the naturalness of this concept, so far there has been little systematic work on parameterized dynamic data structures. Here, the main point of reference is the work of Alman et al.~\cite{AlmanMW17}, who considered a large set of problems fundamental for parameterized complexity, such as {\sc{Vertex Cover}}, {\sc{Hitting Set}}, {\sc{Feedback Vertex Set}}, or \LPath, and delivered a number of parameterized dynamic data structures for them. Following the earlier work of Iwata and Oka~\cite{IwataO14}, many results of Alman et al. are based on dynamization of standard techniques of parameterized algorithms, such as branching, kernelization, or color-coding.

%One of the problems studied by Alman et al.~\cite{AlmanMW17}, which will be also considered here, is \LPath, often also called {\sc{$k$-Path}}: decide whether a given undirected graph $G$ contains a simple path on $k$ vertices. The problem is fixed-parameter tractable when parameterized by $k$, and designing efficient fpt algorithm for it has played a pivotal role throughout the history of parameterized complexity. Indeed, \LPath was the main protagonist in the development of fundamental techniques such as representative sets~\cite{Monien85}, treewidth-based win-win approaches~\cite{Bodlaender93}, color-coding~\cite{AlonYZ95}, and algebraic coding or monomial testing~\cite{BjorklundHKK17,Koutis08,Williams09}. For the dynamic variant of \LPath, Alman et al.~\cite{AlmanMW17} designed a data structure that uses $k!\cdot 2^{\Oh(k)}\cdot \mathsf{DC}(n)$ time per update and $2^{\Oh(k)}\cdot \log n$ time per query, where $\mathsf{DC}(n)$ denotes the query/update time for a data structure maintaining dynamic connectivity. There are several implementations of dynamic connectivity, 
%yet they all achieve an update time that is polylogarithmic in $n$; see the discussion in~\cite{AlmanMW17}. The data structure of Alman et al.~\cite{AlmanMW17} is based on dynamization of the standard color-coding approach.
%Thus, the polylogarithmic factors --- originating both from dynamic connectivity and from color-coding --- seem hard to avoid. 

Apart from~\cite{AlmanMW17,IwataO14}, we are aware of several scattered works on parameterized dynamic data structures; some of them will be mentioned later on, while others, due to not being directly related to our motivation, are reviewed in Appendix~\ref{app:intro}. However, roughly speaking, parameterized dynamic data structures present in the literature achieve update and query times of the forms (here, $k$ is always the solution size):
\begin{itemize}[nosep]
 \item $f(k)$, i.e., independent of the input size $n$ (examples: {\sc{Vertex Cover}}, {\sc{$d$-Hitting Set}} for fixed $d$~\cite{AlmanMW17,BannachHRT19});
 \item $f(k)\cdot \log^c n$, where $c$ is a universal constant (examples: {\sc{Feedback Vertex Set}}, \LPath~\cite{AlmanMW17});
 \item $\log^{f(k)} n$ (example: counting $k$-vertex patterns in sparse graphs~\cite{DvorakT13}).
\end{itemize}
Moreover, there are problems for which already for some constant value of $k$, every dynamic data structure requires $\Omega(n^\delta)$ update time or $\Omega(n^\delta)$ query time, for some $\delta>0$. As shown by Alman et al.~\cite{AlmanMW17}, under certain assumptions from fine-grained complexity, the directed variant of \LPath falls into this category.

This charts an intriguing and still largely unexplored complexity landscape. In order to make this area well-founded, we are particularly interested in developing new techniques for designing parameterized data~structures, specific to the dynamic setting. We contribute to this direction by developing decomposition-based techniques and applying them to the \LPath and \LCycle problems.

\paragraph*{Dynamic treedepth.} In this work we explore a new approach to the design of parameterized data structures, which is based on {\em{elimination forests}} and the parameter {\em{treedepth}}. Here, an {\em{elimination forest}} of a graph $G$ is a rooted forest $F$ on the same vertex set as $G$, where for every edge $uv$ of $G$, either $u$ is an ancestor of $v$ or vice versa. The {\em{treedepth}} of $G$ is the minimum possible height of an elimination forest of $G$. 

Treedepth can be regarded as a variant of treewidth where instead of the width of a decomposition, we measure its height; in fact, the treedepth of a graph is never larger than its treewidth. The importance of treedepth in the hierarchy of parameters has been gradually realized throughout the recent years. It is a central parameter in the theory of sparse graphs (see~\cite{NesetrilM12} for an overview), it has interesting combinatorial properties related to obstructions~\cite{CzerwinskiNP19,DvorakGT12,KawarabayashiR18}, and it was rendered an important dividing line from the points of view of logic~\cite{ElberfeldGT16} and of space/time complexity trade-offs~\cite{PilipczukW18}. In this work we are interested in the dynamic aspects of treedepth, and our first contribution is the following result.

\begin{theorem}\label{thm:td-data}
 Suppose $G$ is a dynamic graph on $n$ vertices that is updated by edge insertions and edge removals, subject to a promise that the treedepth of $G$ never exceeds $d$. Then there is a data structure that, under such updates, maintains a minimum-height elimination forest of $G$ using $2^{\Oh(d^2)}$ time per update in the worst case. Upon receiving an edge insertion that would break the promise, the data structure does not carry out the insertion and reports this fact. The data structure uses $\Oh(d \cdot n)$ memory.
\end{theorem}

In fact, we are not the first to consider dynamic data structures for graphs of bounded treedepth. This problem was considered by Dvo\v{r}\'ak et al.~\cite{DvorakKT14}, who gave a data structure with the same functionality as that provided by Theorem~\ref{thm:td-data}, but achieving update time $f(d)$ for some non-elementary function $f$. Recall that this means that $f(d)$ is tower-exponential: it is not bounded by the $t$-fold exponential function, for any constant $t$. 
The starting idea for our design of the data structure of Theorem~\ref{thm:td-data}, which will be further called the {\em{dynamic treedepth data structure}}, lies in the general strategy proposed by Dvo\v{r}\'ak et al.~\cite{DvorakKT14}. However, we rely on a new, deeper understanding of the combinatorics of treedepth and implement updates in a completely different way, which results in the improved update time of $2^{\Oh(d^2)}$. We include a comprehensive comparison of the approaches in Appendix~\ref{app:intro}.

The $2^{\Oh(d^2)}$ update time offered by Theorem~\ref{thm:td-data} reaches a certain limit. Namely, the fastest known static fpt algorithm for computing the treedepth of a graph, due to Reidl et al.~\cite{ReidlRVS14}, runs in time $2^{\Oh(d^2)}\cdot n$, where $d$ is the value of the treedepth. Thus, achieving $2^{o(d^2)}$ update time in Theorem~\ref{thm:td-data} would automatically improve the result of Reidl et al. to a $2^{o(d^2)}\cdot n$-time static algorithm, by introducing edges one by one. Interestingly, in the proof of Theorem~\ref{thm:td-data}, the $2^{\Oh(d^2)}$ update time in fact originates from applying the algorithm of Reidl et al.~\cite{ReidlRVS14} as a black-box to a graph of size $d^{\Oh(d)}$. This is the only bottleneck preventing the improvement of the $2^{\Oh(d^2)}$ update time. So we can actually conclude that improving this factor in the dynamic setting is {\em{equivalent}} to improving it in the static setting (up to the next bottleneck of $d^{\Oh(d)}$).

As a by-product of the combinatorial analysis leading to Theorem~\ref{thm:td-data}, we also give improved bounds on the sizes of minimal obstructions for having treedepth $d$. More precisely, we say that a graph $G$ is a {\em{minimal obstruction for treedepth $d$}} if its treedepth is larger than $d$, but every proper induced subgraph of $G$ has treedepth at most $d$.
Note that every graph of treedepth larger than $d$ contains some minimal obstruction for treedepth $d$ as an induced subgraph, hence such obstructions are minimal ``witnesses'' for having large treedepth.
Dvo\v{r}\'ak et al.~\cite{DvorakGT12} proved that every minimal obstruction for treedepth $d$ has at most $2^{2^{d-1}}$ vertices, and they gave a construction of an obstruction with $2^d$ vertices.
They also hypothesized that, in fact, every minimal obstruction for treedepth $d$ has at most $2^d$ vertices.
We get closer to this conjecture by showing an improved upper bound of $d^{\Oh(d)}$.

%Moreover, the data structure may maintain the answer to any fixed problem definable in {\em{Monadic Second Order logic}} $\MSO_2$ within the same update time.

\paragraph*{Detecting paths and cycles.} We showcase the potential of the dynamic treedepth  data structure by using it to design fully dynamic data structures for the \LPath and \LCycle problems. In these problems, for a given undirected graph $G$ and parameter $k$, the task is to decide whether $G$ contains a path on $k$ vertices or a cycle on at least $k$ vertices, respectively. The following theorem summarizes our results.

\begin{theorem}\label{thm:main-intro}
 Let $k$ be a fixed parameter.
 Suppose $G$ is a dynamic graph on $n$ vertices, updated by edge insertions and edge deletions, and we are given access to a dictionary on the edges of $G$ using $\mdict$ memory with operations taking amortized time bounded by $\tdict$. Then there are  data structures that, upon such updates, maintain the answers to the queries:
 \begin{itemize}[nosep]
  \item Does $G$ contain a simple path on $k$ vertices?
  \item Does $G$ contain a simple cycle on at least $k$ vertices?
 \end{itemize}
 In the first case, the data structure achieves amortized update time $2^{\Oh(k^2)}+\Oh(\tdict)$ and uses $(n \cdot 2^{\Oh(k \log k)} + \mdict)$ memory. In the second case, the amortized update time is $2^{\Oh(k^4)}+k^{\Oh(k^2)}\cdot \tdict+\Oh(k\log n)$, and the memory usage is $(n \cdot 2^{\Oh(k^2 \log k)} + \mdict)$.
\end{theorem}

Note that Theorem~\ref{thm:main-intro} concerns general graphs, not just graphs of bounded treedepth.
In Theorem~\ref{thm:main-intro} we do not specify the query time, because we consider decision problems: the answer to the query is recomputed upon every update and can be later be provided in constant time. Also, we assume access to a dictionary on the edges of the graph. There are several ways of implementing such a dictionary that differ in trade-offs between time/space complexity and allowing amortization or randomization. For instance, the simplest solution --- an adjacency matrix --- achieves worst-case constant operation time at the cost of quadratic space complexity, while dynamic perfect hashing~\cite{dynamicHashing} gives linear space complexity, but guarantees only {\em{expected amortized}} constant  time per operation. We review different options in Section~\ref{sec:prelims}.

\LPath occupies a central position in parameterized complexity theory due to serving as the main protagonist in the development of several fundamental techniques: representative sets~\cite{Monien85}, treewidth-based win-win approaches~\cite{Bodlaender93}, color-coding~\cite{AlonYZ95}, algebraic coding or monomial testing~\cite{BjorklundHKK17,Koutis08,Williams09}, and kernelization lower bounds~\cite{BodlaenderDFH09}. \LCycle is less prominent in comparison, but is known to be fpt even in the directed variant~\cite{Zehavi16}.
As mentioned above, \LPath in the dynamic setting was already considered by Alman et al.~\cite{AlmanMW17}. By dynamizing the standard color-coding approach~\cite{AlonYZ95}, they designed a data structure that uses $k!\cdot 2^{\Oh(k)}\cdot \mathsf{DC}(n)$ time per update, where $\mathsf{DC}(n)$ denotes the query/update time for a data structure maintaining dynamic connectivity. There are several implementations of dynamic connectivity, yet they all achieve an update time that is polylogarithmic in~$n$, and actually there is an $\Omega(\log n)$ lower bound in the cell-probe model~\cite{10.1145/1007352.1007435}; see the discussion in~\cite{AlmanMW17}. Thus, while Theorem~\ref{thm:main-intro} offers worse parametric factor of the update time compared to the data structure of Alman et al.~\cite{AlmanMW17}, it completely removes the dependence on the size of the graph, which seems difficult in the approach used in~\cite{AlmanMW17}.
We are not aware of any previous work on dynamic data structures for the \LCycle problem.

%The data structure of Alman et al.~\cite{AlmanMW17} is based . The polylogarithmic factors --- originating both from dynamic connectivity and from color-coding --- seem hard to avoid using this approach. 

\paragraph*{Techniques behind Theorem~\ref{thm:main-intro}.}
At first glance, it may seem surprising how the dynamic treedepth data structure can be helpful in designing data structures for \LPath and \LCycle, because these data structures should work on an arbitrary dynamic graph, without any promise about the treedepth. Here, we use the following well-known connection (see e.g.~\cite[Proposition~6.1]{NesetrilM12}): a graph of treedepth at least $k$ always contains a path on $k$ vertices. Hence, the answer to \LPath is non-trivial {\em{only}} if the treedepth is smaller than $k$; otherwise it is trivially positive. To capitalize on this observation, we use a technique of postponing invariant-breaking insertions, introduced by Eppstein et al.~\cite{EppsteinGIS96} in the context of planarity testing. 
Effectively, this enables us to focus on the case when the treedepth of the maintained dynamic graph is at all times bounded by $k$, at the cost of allowing amortization in the update time guarantees.

%In essence, we maintain the dynamic treedepth data structure for $d=k$. Whenever an edge insertion would violate the invariant that the maintained graph has treedepth at most $k$, instead of carrying out the insertion we put it in an auxiliary queue. The postponed insertions may be carried out later, when they become safe after some edge deletions. 

%To conclude the first point of Theorem~\ref{thm:main-intro} --- the data structure for \LPath --- it now remains to augment the dynamic treedepth data structure so that it also supports queries about $k$-vertex paths. 

Next, we show that the dynamic treedepth data structure can be conveniently enriched with all sorts of dynamic programming procedures on elimination forests, so that in the dynamic setting we may maintain their tables upon edge insertions and deletions. By doing this for the standard dynamic programming procedure for \LPath, we complete the proof of the first point of Theorem~\ref{thm:main-intro}.

When working out this part of the argument, we make effort to introduce a convenient language for formulating dynamic programming on elimination forests that combines well with the dynamic treedepth data structure. This is because we expect these parts of our work to be of a wider applicability. In fact, we consider this to be one of the most important conceptual messages of this paper: dynamic programming on graphs of bounded treedepth can be efficiently maintained in the dynamic setting, and this is a {\em{technique}} for the design of parameterized data structures.

We remark that the data structure of Dvo\v{r}\'ak et al.~\cite{DvorakKT14} can be similarly combined with dynamic programming. In fact, they show that for every fixed problem definable in {\em{Monadic Second Order logic}} $\MSO_2$ (see~\cite[Section~7.4.1]{platypus} for introduction), the answer to this problem can be maintained together with the dynamic treedepth data structure within the same complexity; that is, with update time $f(d)$ for a non-elementary function~$f$. This is also the case for our data structure. Since the \LPath problem is expressible in $\MSO_2$, it is possible to derive a data structure for dynamic \LPath with amortized update time $f(k)$, for a non-elementary function~$f$, by combining the result of Dvo\v{r}\'ak et al.~\cite{DvorakKT14} with the technique of Eppstein et al.~\cite{EppsteinGIS96}. Here, $k$ is the requested vertex count of the path.

We move on to the second point of Theorem~\ref{thm:main-intro} ---
the data structure for \LCycle. This requires further ideas. The main issue is that the connection with treedepth a priori fails: as witnessed by paths, there are graphs of arbitrary large treedepth and no cycles at all. However, to some extent the approach can be salvaged: it can be shown (see~\cite[Proposition~6.2]{NesetrilM12}) that every {\em{biconnected}} graph of treedepth at least $k^2$ contains a simple cycle on at least $k$ vertices. We use this combinatorial observation as follows.

Due to the technique of postponing insertions, we may assume that the maintained graph $G$ at all times does not contain a simple cycle on at least $k$ vertices.
Then the abovementioned combinatorial fact implies that every biconnected component of $G$ has treedepth at most $k^2$. Therefore, our data structure maintains a partition of $G$ into biconnected components, and for each biconnected component $H$ of $G$ we maintain an elimination forest $F_H$ of $H$ of height at most $k^2$. Roughly speaking, for maintaining the partition into biconnected components, we use the top trees data structure of Alstrup et al.~\cite{AlstrupIP,AlstrupJ}, which introduces the $\Oh(k \log n)$ factor to the update time. The forests $F_H$ for biconnected components $H$ are maintained using the dynamic treedepth data structures for $d=k^2$. Observe that, upon edge insertions and removals, the biconnected components of the graph may merge or split. For this, we need to design appropriate merge and split procedures for the dynamic treedepth data structures. Fortunately, our  understanding of the combinatorics of treedepth allows this, at the cost of significant technical effort.

\paragraph*{Lower bounds.}
Observe that the update time offered by Theorem~\ref{thm:main-intro} for the \LCycle problem contains an $\Oh(\log n)$ factor. This is in fact unavoidable: a data structure for detecting simple cycles of length at least $3$ (aka just cycles) can be used for maintaining dynamic connectivity in forests, for which there is an $\Omega(\log n)$ lower bound in the cell-probe model~\cite{10.1145/1007352.1007435}. See Corollary~\ref{cor:3-cyc-lb} in Section~\ref{sec:lower-bounds} for a formal derivation of this result. Thus, there is a qualitative difference between \LPath and \LCycle in the dynamic setting: the first problem admits a data structure with amortized update time independent of $n$, while in the second factors linear in $\log n$ are necessary in the update time guarantees.

Here, let us point out another curious application of the data structure offered by Theorem~\ref{thm:td-data}. Using it, it is very easy to implement connectivity queries (whether given vertices $u$ and $v$ are in the same connected component) in time $\Oh(d)$: it will be always the case that the maintained forest $F$ has one tree per each connected component of $G$, so it suffices to check whether $u$ and $v$ are in the same tree of $F$, which can be done by following parent pointers to respective roots. This gives a data structure for dynamic connectivity in graphs of treedepth at most $d$ with update time $2^{\Oh(d^2)}$ and query time $\Oh(d)$. On the other hand, the $\Omega(\log n)$ lower bound for dynamic connectivity of Demaine and P\v{a}tra\c{s}cu~\cite{10.1145/1007352.1007435} applies even to forests of paths, which can be thought of the simplest classes that do {\em{not}} have bounded treedepth. This means that in some sense, the possibility of maintaining dynamic connectivity with update and query time independent of $n$ is tightly linked with assuming a bound on the treedepth of the considered dynamic graph.

A different lower bound methodology was proposed by Alman et al.~\cite{AlmanMW17}. Among other results, they proved that any data structure for the directed variant of \LPath for $k=5$ has to assume $\Omega(n^{\delta})$ query time, or $\Omega(n^{\delta})$ update time, or $\Omega(n^{1+\delta})$ initialization time on an edgeless graph, for some $\delta>0$. This lower bound is conditional, subject to a hypothesis called {\em{$\ell$-layered reachability oracle ($\ell$LRO) Conjecture}}, which in turn is implied by the Triangle Conjecture and by the 3SUM Conjecture --- assumptions commonly adopted in fine-grained complexity. Using this technique, we give analogous lower bounds for the following variations of the considered problems:
\begin{itemize}[nosep]
 \item dynamic undirected {\sc{$5$-Path}}, where instead of asking for any $5$-path, we look for a $5$-path with a specified pair of endpoints; and
 \item dynamic undirected {\sc{Exact-$5$-Cycle}}, where we ask for the existence of a cycle on {\em{exactly}} $5$ vertices, instead of {\em{at least}} $5$ vertices.
\end{itemize}
See Theorem~\ref{thm:cond-lb} in Section~\ref{sec:lower-bounds} for a formal statement. Note that in the static setting, all the variations mentioned above can be solved in fpt time using color coding or algebraic coding.
Thus, the tractability domain for \LPath and related problems is much narrower in the dynamic setting. It seems that the combinatorial links with treedepth, heavily exploited in our approach, are a necessary ingredient without which not only the technique breaks, but the problems actually become provably hard.

\paragraph*{Organization.} In Section~\ref{sec:overview} we present a concise overview of the reasoning leading to our main results, focusing on explaining the key ideas rather than technical details. In the subsequent sections we gradually build up all the necessary tools for the proofs of Theorems~\ref{thm:td-data} and~\ref{thm:main-intro}, which are concluded in Sections~\ref{sec:data-structure} and~\ref{sec:kpath}, respectively. Along the way, in Section~\ref{sec:obstructions} we prove the results on obstructions for bounded treedepth. As the arguments needed for the data structure for the \LCycle problem extend those needed for the \LPath problem, we mark sections concerning \LCycle with a small cycle placed next to the section title; these sections can be omitted by a reader interested only in \LPath.
In Section~\ref{sec:lower-bounds} we present the lower bounds, while in Section~\ref{sec:conclusions} we gather concluding remarks and outline further research directions.

%%% Local Variables:
%%% mode: latex
%%% TeX-master: "main"
%%% End:

%\input{intro}

\section{Overview}\label{sec:overview}

\subsection{Treedepth, elimination forests, and cores}\label{secov:cores}
In this subsection we give an overview of the material presented in Sections~\ref{sec:cores} and~\ref{sec:obstructions}.
Our first goal is to obtain a fine combinatorial understanding of elimination forests of optimum height, so that we will be able to efficiently recompute them upon edge insertions and deletions. 

Recall that an {\em{elimination forest}} (or equivalently treedepth decomposition) of a graph $G$ is a rooted forest on the same vertex set as $G$ satisfying the following property: for every edge $uv$ of $G$, either $u$ is an ancestor of $v$ in $F$, or vice versa. Note that edges of $F$ do not need to be present in $G$. Elimination forests are graph decompositions underlying the parameter {\em{treedepth}}, defined as the minimum possible height of an elimination forest. 
Note that an elimination forest of a connected graph is necessarily a tree.

 We will work with elimination forests that are in some sense also ``locally optimal'', as explained in the following definition. Here, for a forest $F$ and vertex $u\in V(F)$, by $F_u$ we denote the subtree of $F$ induced by $u$ and all its descendants. Also, $\desc_F(u)$ denotes the set of descendants of $u$ in $F$, including $u$ itself.

\begin{definition}
 An elimination forest $F$ of a graph $G$ is {\em{recursively optimal}} if for every vertex $u$, the graph $G[\desc_F(u)]$ is connected and has treedepth equal to the height of $F_u$. 
\end{definition}

In other words, in a recursively optimal elimination forest $F$, each subtree $F_u$ is an optimum-height elimination tree of $G[\desc_F(u)]$. Thus, the height of a recursively optimal elimination forest $F$ always matches the treedepth of the graph, as $F$ needs to optimally decompose each connected component in a separate tree. It is easy to see that every graph has a recursively optimal elimination forest, and, up to technical details, such a forest can be computed in time $2^{\Oh(d^2)}\cdot n^{\Oh(1)}$ using the algorithm of Reidl et al.~\cite{ReidlRVS14}.

\medskip

From now on, let us fix a graph $G$ and assume, for simplicity, that $G$ is connected. Let $T$ be a recursively optimal elimination tree of $G$, say of height $d$. For a vertex $u$, we define the {\em{strong reachability set}} $\SReach(u)$ as the set $N_G(\desc_T(u))$, that is, $\SReach(u)$ consists of all (strict) ancestors of $u$ that have neighbors among $\desc_T(u)$. The intuition is that $\SReach(u)$ is the set to which the subtree $T_u$ is ``attached'' in $T$. In other words, if we temporarily removed $T_u$ from $T$ and wanted to attach it back, then the optimal way would be to find the deepest vertex $m$ of $\SReach(u)$ and attach $T_u$ by making $u$ a child of $m$. In this way, the conditions in the definition of an elimination forest are satisfied, while $T_u$ is attached as high as~possible.

A {\em{prefix}} of $T$ is an ancestor-closed subset of vertices of $T$. For a non-empty prefix $K$ of $T$, an {\em{appendix}} of $K$ is a vertex that does not belong to $K$, but whose parent already belongs to $K$. The set of appendices of $K$ will be denoted by $\App(K)$. The following definition is the cornerstone of our analysis.

\begin{definition}
 Let $q\in \N$. A non-empty prefix $K$ of $T$ is a {\em{$q$-core}} (of $(G,T)$) if the following property holds: for every appendix $a\in \App(K)$ and subset $X\subseteq \SReach(a)$ of size at most $2$, $a$ has at least $q$ distinct siblings $w$ in $T$ such that $w\in K$, $X\subseteq \SReach(w)$, and $\height(T_w)\geq \height(T_a)$.
\end{definition}

Before we continue, let us verify that we can always find very small cores.

\begin{lemma}\label{lemov:small-core}
 For each $q\geq 2$, there is a $q$-core of size at most $(qd)^{\Oh(d)}$.
\end{lemma}
\begin{proof}[Sketch]
 Consider the following recursive marking procedure that can be applied to a vertex $u\in V(G)$. For each set $X$ consisting of at most $2$ ancestors of $u$ (including $u$), consider all the children $w$ of $u$ satisfying $X\subseteq \SReach(u)$, and mark $q$ of them with the largest values of $\height(T_w)$, or all of them if their number is smaller than $q$. Then apply the procedure recursively on each marked child of $u$. It is straightforward to see that if $K$ comprises of all vertices that got marked after applying the procedure to the root of $T$, then $K$ is a $q$-core. Since for a marked vertex $u$ we also mark at most $q\cdot (1+d+\binom{d}{2})=\Oh(qd^2)$ children of $u$, and $T$ has height at most $d$, it follows that $|K|\leq (qd)^{\Oh(d)}$.
\end{proof}

Let us now explain the idea behind the definition of a core. Suppose $K$ is a $q$-core, $a\in \App(K)$, and $w$ is a sibling of $a$ satisfying the property from the definition for some $X=\{x,y\}\subseteq \SReach(a)$. Since $T$ is recursively optimal, $G[\desc_T(w)]$ is connected. As $x,y\in \SReach(w)=N_G(\desc_T(w))$, we conclude that in $G$ there is a path $P^{x,y}_w$ such that $P^{x,y}_w$ has endpoints $x$ and $y$, and all the internal vertices of $P^{x,y}_w$ belong to $\desc_T(w)$. As we have $q$ such siblings $w$, we can find $q$ such paths $P^{x,y}_w$, and they will be pairwise internally vertex-disjoint. It can now easily be seen that, provided $q\geq d$, such a set of $q$ paths forces that in {\em{every}} elimination tree of $G$ of depth at most $d$, $x$ and $y$ have to be in the ancestor-descendant relation. Since none of the paths $P^{x,y}_w$ intersects $\desc_T(a)$, this conclusion can be drawn even if we removed all the vertices of $\desc_T(a)$ from $G$. This reasoning can be applied for every pair $\{x,y\}\subseteq \SReach(a)$. After fixing technical details, this amounts to the following statement.

\begin{lemma}\label{lemov:straight}
 Suppose $K$ is a $d$-core and $T^K$ is any elimination tree of the graph $G[K]$ of height at most $d$. Then for each $a\in \App(K)$, the set $\SReach(a)$ is {\em{straight}} in $T^K$, that is, all the vertices of $\SReach(a)$ lie on one root-to-leaf path in $T^K$. 
\end{lemma}

\begin{figure}
  \centering
  \includegraphics[width=0.6\textwidth]{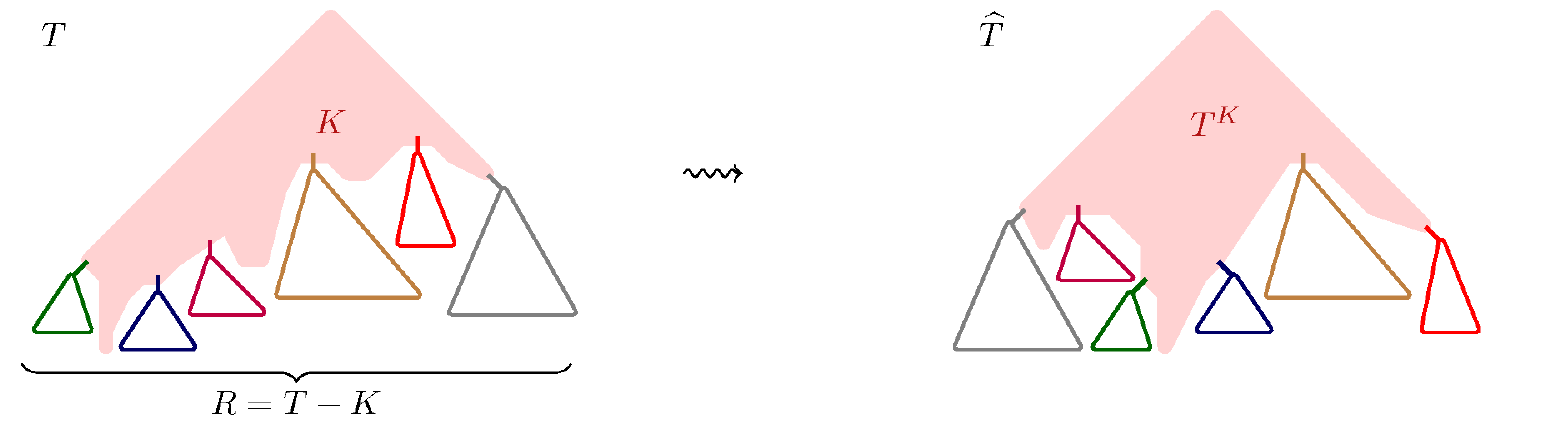}%
  \caption{Treedepth cores and their replaceable elimination trees.}\label{fig:choinka}
\end{figure}

Supposing that $K$ is a $d$-core, let $R=T-K$ be the rooted forest obtained by removing all the vertices of $K$ from $T$; see Figure~\ref{fig:choinka}. Note that $R$ is an elimination forest of the graph $G-K$. Then the conclusion of Lemma~\ref{lemov:straight} means for {\em{every}} elimination tree $T^K$ of $G[K]$ of height at most $d$ (possibly very different from $T[K]$),  $R$ is {\em{attachable}} to $T^K$ in the following sense: for each tree $S$ of $R$, the set $N_G(V(S))$ is straight in $T^K$. Recalling our previous intuition, this suggests that we can compute a new elimination forest $\wh{T}$ of $G$ by attaching $R$ ``below'' $T^K$ as follows (note that $\App(K)$ coincides with the set of roots of $R$). For each $a\in \App(K)$, we let $m$ be the deepest vertex of $\SReach(a)$ in $T^K$, and we attach the tree $R_a=T_a$ by making $a$ a child of $m$. It is straightforward to see that if $R$ is attachable to $T^K$, then the tree $\wh{T}$ obtained in this manner is indeed an elimination tree of $T$; we shall call it the {\em{extension}} of $T^K$ via $R$.

We see that $\wh{T}$ constructed as above is indeed an elimination forest of $G$, but so far we cannot say much about its height. Fortunately, from the definition of a core we can also derive useful properties in this respect. In particular, this is why in this definition we insisted that for each of the distinguished siblings $w$ of $a$, the height of $F_w$ is not smaller than the height of $F_a$. 

The intuition is as follows. Suppose $K$ is a $(d+1)$-core and $T^K$ is a recursively optimal elimination forest of $G[K]$ of height at most $d$. Consider any $a\in \App(K)$. By Lemma~\ref{lemov:straight}, the set $\SReach(a)$ is straight in $T^K$. Let then $m$ be the vertex of $\SReach(a)$ that is the deepest in $T^K$. Then, by the definition of the core we can find a set $W\subseteq K$ consisting of $d+1$ siblings $w$ of $a$ such that 
$$m\in \SReach(w)\qquad\textrm{and}\qquad\td(G[\desc_T(w)])=\height(T_w)\geq \height(T_a)=\td(G[\desc_T(a)]).$$
Here, the first and the last equality follows from the recursive optimality of $T$. Through a fairly complicated inductive scheme, we can show for one of the trees $T_w$ for $w\in W$, the graph $G[K\cap \desc_T(w)]$ has same treedepth as $G[\desc_T(w)]$ and its vertex set is fully contained in a subtree $T^K_x$ for some $x$ that is a child of $m$ in $T^K$. This witnesses that $\height(T^K_x)\geq \height(T_w)\geq \height(T_a)$. We infer that attaching $T_a$ below $m$ in the construction of $\wh{T}$ --- the extension of $T^K$ via $R=T-K$ --- {\em{cannot}} increase the height of $\wh{T}$ above $\height(T^K)$. This is because after the attachment, there anyway is a subtree rooted at a sibling of $a$ whose height is not smaller than the height of $T_a$. We may conclude the following.

\begin{lemma}\label{lemov:extend}
 Let $K$ be a $(d+1)$-core and let $T^K$ be any elimination tree of $G[K]$ of height at most $d$. Let $R=T-K$ and let $\wh{T}$ be the extension of $T^K$ via $R$ (which is well-defined by Lemma~\ref{lemov:straight}). Then $\wh{T}$ is a recursively optimal elimination tree of $G$ and the height of $\wh{T}$ is equal to the height of $T^K$.
\end{lemma}

Recall that in the first place, we were interested in recomputing a recursively optimal elimination forest of a graph under edge insertions and edge deletions. Let then $H$ be a graph obtained from $G$ by either inserting an edge $uv$, or deleting an edge $uv$. By following the same reasoning that led us to Lemmas~\ref{lemov:straight} and~\ref{lemov:extend}, but additionally keeping track of the modified edge $uv$, we can prove the following.

\begin{lemma}\label{lemov:extend-update}
 Let $K$ be a $(d+2)$-core of $(G,T)$ that includes both $u$ and $v$, and let $T^K$ be any elimination tree of $H[K]$ of height at most $d$. Let $R=T-K$. Then $R$ is attachable to $T^K$. Moreover, if $\wh{T}$ is the extension of $T^K$ via $R$, then $\wh{T}$ is a recursively optimal elimination forest of $H$ whose height is equal to the height of $T^K$.
\end{lemma}

Note that in Lemma~\ref{lemov:extend-update} we require that $K$ is a $(d+2)$-core, while Lemma~\ref{lemov:extend} only assumed that $K$ is a $(d+1)$-core. This is because some of the witnessing structures, for instance paths $P^{x,y}_w$ considered in the reasoning leading to Lemma~\ref{lemov:straight}, might be affected by the removal of the edge $uv$. However, as this is just a single edge, adding $1$ to the requirement on the core suffices for the argument to go through.

Lemma~\ref{lemov:extend-update} suggests the following strategy for recomputing a recursively optimal elimination tree upon inserting or deleting an edge $uv$. Here, $H$ is the updated graph.
\begin{itemize}[nosep]
 \item Using the procedure described in the proof of Lemma~\ref{lemov:small-core}, compute a $(d+2)$-core $K$ of $T$ of size $d^{\Oh(d)}$. By modifying this procedure slightly, we may make sure that $u,v\in K$.
 \item Using the static algorithm of Reidl et al.~\cite{ReidlRVS14}, compute a recursively optimal elimination forest $T^K$ of $H[K]$. Since the treedepth of $H[K]$ does not exceed the treedepth of $H$, which in turn does not exceed $d+1$, this takes time $2^{\Oh(d^2)}$. Observe that if the treedepth of $H[K]$ turns out to be larger than $d$, then the same can be concluded about $H$.
 \item Letting $R=T-K$, compute $\wh{T}$: the extension of $T^K$ via $R$.
\end{itemize}
Then Lemma~\ref{lemov:extend-update} asserts that $\wh{T}$ is a recursively optimal elimination forests of $H$. 

Note that this procedure readily can be implemented as a static algorithm, but the idea seems useful in the dynamic setting as well. This is because the forest $R$ --- which constitutes a vast majority of the graph, provided $d\ll n$ --- first gets detached from $T$ and then gets attached below $T^K$ while keeping its structure intact. In the next section, our goal will be to implement this detachment and attachment so that the internal data about $R$ does not need to be updated at all.

\medskip

Observe that from Lemma~\ref{lemov:extend} we may immediately derive the following conclusion: the subgraph induced by a $d$-core inherits the treedepth of the original graph.

\begin{lemma}\label{lemov:td-opt}
 Let $K$ be a $(d+1)$-core. Then the treedepth of $G[K]$ is equal to the treedepth of $G$.
\end{lemma}
\begin{proof}
 Let $T^K$ be an elimination forest of $G[K]$ of minimum height. By Lemma~\ref{lemov:extend}, there exists an elimination forest of $G$ of height equal to the height of $T^K$. This means that the treedepth of $G$ is not larger than the treedepth of $G[K]$. As the reverse inequality is obvious, the lemma follows. 
\end{proof}

We note that in reality, we prove (the formal analogs of) Lemmas~\ref{lemov:extend} and Lemmas~\ref{lemov:td-opt} in the reverse order, as (a generalization of) Lemma~\ref{lemov:td-opt} is needed in the inductive scheme used in the proof of Lemma~\ref{lemov:extend}.

From Lemmas~\ref{lemov:small-core} and~\ref{lemov:td-opt} we may now immediately derive the improved bounds on minimal obstructions for bounded treedepth. Recall here that a graph $G$ is a minimal obstruction for treedepth $d$ if the treedepth of $G$ is larger than $d$, but every proper induced subgraph of $G$ has treedepth at most $d$.

\begin{theorem}\label{thmov:obstructions}
 Every minimal obstruction for treedepth $d$ has at most $d^{\Oh(d)}$ vertices.
\end{theorem}
\begin{proof}
 Let $G$ be a minimal obstruction for treedepth $d$. Clearly, $G$ is connected. As removing one vertex decreases the treedepth by at most $1$, the treedepth of $G$ is equal to $d+1$. Let $T$ be a recursively optimal elimination tree of $G$; then $T$ has height $d+1$. By Lemma~\ref{lemov:small-core}, we may find a $(d+1)$-core $K$ of $(G,T)$ of size~$d^{\Oh(d)}$. By Lemma~\ref{lemov:td-opt}, the treedepth of $G[K]$ is $d+1$. Since $G$ is a minimal obstruction, we necessarily have $K=V(G)$. So $G$ has $d^{\Oh(d)}$ vertices.
\end{proof}

\subsection{Dynamic treedepth data structure}\label{secov:data}

In this subsection we give an overview of the material presented in Sections~\ref{sec:data-structure} and~\ref{sec:dp}, in particular we explain the proof of Theorem~\ref{thm:td-data}. The idea is to design the dynamic treedepth data structure so that the detachment/attachment strategy described in the previous subsection can be implemented efficiently. We note that the internal organization of information in our data structure roughly follows the general strategy of Dvo\v{r}\'ak et al.~\cite{DvorakKT14}, but the approach we use to implement the update methods, which is based on the analysis of cores that we developed in the previous section, is completely new and different from~\cite{DvorakKT14}. This is the part of the reasoning that leads to the improvement.

Suppose that $G$ is the considered graph, say connected for simplicity, and $T$ is an elimination tree of $G$ of depth at most $d$. Suppose further that $G$ is updated by inserting or deleting an edge $uv$, and $H$ is the updated graph. As outlined in the previous subsection, we should compute a $(d+2)$-core $K$ of $(G,T)$ of size $d^{\Oh(d)}$, recompute a recursively optimal elimination forest $T^K$ of $H[K]$ using a static algorithm, and reattach $R\coloneqq T-K$ below $T^K$. Consider any appendix $a\in \App(K)$ and recall that reattaching $R_a=T_a$ boils down to making $a$ a child of the vertex of $\SReach(a)$ that is the deepest in $T^K$. Now comes the main observation: for any other appendix $a'\in \App(K)$ satisfying $\SReach(a')=\SReach(a)$, the tree $R_{a'}=T_{a'}$ will be attached at exactly the same place as $R_a$. Therefore, we can treat all trees $R_a$ with the same $\SReach(a)$ as one ``batch'', which will be detached from $T$ and reattached to $\wh{T}$ concurrently. Here, it is not hard to see that all the trees of this batch have the same parent in $T$, which obviously belongs to $K$.

We now implement this idea algorithmically. The tree $T$ is stored as follows. For every vertex $u$, we remember the parent of $u$ in $T$, $\SReach(u)$, and $\Up(u)\coloneqq \SReach(u)\cap N_G(u)$; the last set is used to represent the edge set of $G$. Further, $u$ remembers all its children, but these children are partitioned into {\em{buckets}} as follows: for each $X\subseteq \SReach(u)\cup \{u\}$ and $i\leq d$, we store the bucket 
$$\bucket[u,X,i]\coloneqq \{\,v\in \mathsf{children}(u)\ \colon\ \SReach(v)=X\textrm{ and }\height(T_v)=i\,\}.$$
Thus, buckets $\bucket[u,\cdot,\cdot]$ form a partition of the children of $u$ and there are at most $2^d\cdot d$ buckets associated with $u$. Buckets are represented using doubly-linked lists.

Inserting or deleting the edge $uv$ can now be implemented as follows:
\begin{itemize}[nosep]
 \item Construct a $(d+2)$-core $K$ of size $d^{\Oh(d)}$ that includes $u$ and $v$. Having access to buckets, this can be done by simulating the marking procedure presented in Lemma~\ref{lemov:small-core} in time $d^{\Oh(d)}$.
 \item Apply the static algorithm of Reidl et al.~\cite{ReidlRVS14} to compute a recursively optimal elimination tree $T^K$ of~$H[K]$. This step takes time $2^{\Oh(d^2)}\cdot |K|^{\Oh(1)}=2^{\Oh(d^2)}$ and is the only bottleneck: all the other steps take time $d^{\Oh(d)}$. Also, if it turns out that $\height(T^K)>d$, then the treedepth of $H$ exceeds $d$ and the update should be rejected.
 \item Remove all vertices of $K$ from all the buckets. This can be done in time $\Oh(|K|)$ by remembering, for each vertex, a pointer to a list element representing it in the bucket to which it belongs.
 \item For each $u\in K$, $X\subseteq \SReach(u)\cup \{u\}$, and $i\leq d$, rename the bucket $\bucket[u,X,i]$ to $\bucket[m,X,i]$, where $m$ is the vertex of $\SReach(u)$ that is the deepest in $T^K$. Note that there are at most $|K|\cdot 2^d\cdot d=d^{\Oh(d)}$ buckets that need to be renamed in this way.
 \item Recompute the information for the vertices of $K$ and place them in appropriate buckets. This can be done in time $2^{\Oh(d)}\cdot |K|^{\Oh(1)}=d^{\Oh(d)}$  by a bottom-up traversal of $T^K$.
\end{itemize}
The key observation is that in such an implementation, the following assertion holds: for each renamed bucket $\bucket[u,X,i]$, all the values stored for vertices of trees $R_a$ for $a\in \bucket[u,X,i]$ do not need to be updated at all. The only exception are the parent pointers for vertices $a\in \bucket[u,X,i]$, which after renaming should all point to the new parent $m$. This can be easily remedied by storing one parent pointer per bucket, and changing it in a single operation. This concludes the proof of Theorem~\ref{thm:td-data}.

\paragraph*{Configuration schemes.} As mentioned in Section~\ref{sec:introduction}, the dynamic treedepth data structure can be conveniently augmented to maintain a run of a dynamic programming procedure on the stored elimination forest. This applies to a wide range of dynamic programming procedures, in particular those obtained for $\MSO_2$-expressible problems through the classic connection with tree automata. In Section~\ref{sec:dp} we present a general formalism of {\em{configuration schemes}} that can be used to formulate such dynamic programming procedures. To keep the overview simple, we now explain how this idea applies to \LPath.

Suppose $T$ is the maintained elimination tree and consider any $u\in V(G)$. Let $X=\SReach(u)$ and let $G_u$ be the subgraph of $G$ such that the vertex set of $G_u$ is $X\cup \desc_T(u)$, while the edge set of $G_u$ comprises of all the edges of $G$ that have an endpoint in $\desc_T(u)$. Note that, thus, $X$ forms an independent set in $G_u$. We now define a set $\Conf(X)$ of {\em{configurations}} on $X$: a configuration $c$ is a pair $(F,j)$, where $F$ is a {\em{linear forest}} (i.e. a forest of paths) on vertex set $X\cup \{s,t\}$, where $s,t$ are special vertices, and $j$ is an integer not larger than $k$ (the requested vertex count of the path). Note that $|\Conf(X)|\leq |X|^{\Oh(|X|)}\cdot (k+1)\leq d^{\Oh(d)}\cdot k$. Configuration $c=(F,j)$ is {\em{realizable}} in $G_u$ if in $G_u$ there is a family of paths $\{P_e\colon e\in E(F)\}$ of total length $j$ such that the paths $P_e$ are disjoint apart from endpoints in $X$, and the endpoints of $P_e$ match the endpoints of $e$. Here, the special vertices $s,t$ can be replaced with any vertices in $G_u$. Then with each $u\in V(G)$, we can associate the set $\conf(G_u)\subseteq \Conf(X)$ comprising configurations realizable in $G_u$. Whether $G$ contains a $k$-path can be determined by checking whether $\conf(G_r)$, where $r$ is the root of $T$, contains configuration $(F_{st},k-1)$, where $F_{st}$ has only one edge: $st$.

This configuration scheme has two important properties:
\begin{itemize}[nosep]
 \item {\em{Compositionality}}: $\conf(G_u)$ can be computed from the multiset $\{\{\conf(G_v)\colon v\in \chld(u)\}\}$.
 \item {\em{Idempotence}}: There is a threshold $\tau$ (equal to $d+2$) such that for the computation above, it is immaterial whether a configuration is realized in $\tau$ or more children of $u$ (see Section~\ref{sec:dp} for a formal~definition).
\end{itemize}
We show that these two basic properties alone are sufficient for augmenting the dynamic treedepth data structure so that with each $u\in V(G)$, we also implicitly store $\conf(G_u)$. This of course introduces factors depending on the configuration scheme to the update time, but in case \LPath and assuming $k\leq d$, these factors are dominated by the $2^{\Oh(d^2)}$ update time of the data structure.

The augmentation is done as follows. For every bucket $\bucket[u,X,i]$ stored in the data structure, we additionally store a {\em{mug}} $\bucket[u,X,i,c]$ for each configuration $c\in \Conf(X)$. This mug comprises all $v\in \bucket[u,X,i]$ for which $c\in \conf(G_v)$, and is organized as a doubly-linked list (a sublist of $\bucket[u,X,i]$). Note that each $v\in \bucket[u,X,i]$ can appear in multiple mugs, but the number of mugs is bounded by $|\Conf(X)|$, which depends only on $d$ and the configuration scheme in question. The set $\conf(G_v)$ corresponds to the set of mugs to which $v$ belongs. It is now not hard to maintain the mugs during updates using the assumptions of compositionality and idempotence, similarly as we do for buckets.

\subsection{Postponing insertions}

Using all the tools prepared so far for $d=k-1$,
we can implement a fully dynamic data structure that for a graph $G$, promised to be always of treedepth smaller than $k$, maintains an elimination forest of $G$ of height smaller than $k$ together with the answer to the query whether $G$ contains a $k$-path. Upon receiving an invariant-breaking edge insertion, the data structure rejects the update and reports this. We now use the technique of Eppstein et al.~\cite{EppsteinGIS96} to turn this into a data structure working without the promise.

We maintain the dynamic treedepth data structure $\Dt$ described above, which stores a subgraph $G'$ of $G$. Additionally, we have a queue $Q$ of edges whose insertions are postponed. We maintain the invariants: $G$ consists of $G'$ and all the edges stored in $Q$; $G'$ has treedepth smaller than $k$; and if $Q$ is non-empty, then $G'+e$ has treedepth at least $k$, where $e$ is the edge at the front of $Q$. Thus, if $Q$ is empty then $G'=G$. The query about a $k$-path can be implemented as follows: if $Q$ is empty then we can simply query $\Dt$, and otherwise the answer is $\true$, because $G$ has treedepth at least $k$. When inserting an edge, we either try to insert it into $\Dt$ in case $Q$ is empty, or we push it at the back of $Q$ otherwise. The former case may result in rejecting the insertion and pushing the edge into $Q$. Finally, when deleting an edge we either delete it from $\Dt$ or from $Q$, depending where it is currently stored. We may now need to perform a clean-up: iteratively pop an edge from the front of $Q$ and insert it into $\Dt$. This can take large worst-case time, but it is not hard to see that the amortized time remains constant. We use the dictionary to quickly locate edges within $Q$.

\subsection{Detecting cycles}

\CycleNoteSub

Finally, we show how the whole machinery can be put into motion to handle also the \LCycle problem. That is, we give an overview of the material presented in Sections~\ref{sec:splitting-merging},~\ref{sec:scheme-ms}, and~\ref{sec:cyc_det_ds}.

Recall that a biconnected graph of treedepth at least $k^2$ necessarily contains a simple cycle on at least $k$ vertices~\cite[Proposition~6.2]{NesetrilM12}. Hence, if $G$ does not contain a simple cycle on at least $k$ vertices --- and using the technique of Eppstein et al.~\cite{EppsteinGIS96} we can focus on this case --- then every biconnected component of $G$ has treedepth smaller than $k^2$. Therefore, in our data structure we maintain the partition $\Prt$ of $G$ into biconnected components and, for each $H\in \Prt$, the dynamic treedepth data structure $\Dt[H]$ for $d=k^2$, which stores $H$ together with some recursively optimal elimination tree $T_H$ of height at most $d$.

To efficiently handle the partition $\Prt$ upon updates, we maintain a spanning forest $\Upsilon$ of $G$ using the top tree data structure of Alstrup et al.~\cite{AlstrupIP,AlstrupJ}. This data structure supports adding and removing edges from $\Upsilon$ in time $\Oh(\log n)$, as well as the following two queries about a pair of vertices $u,v$:
\begin{itemize}[nosep]
 \item What is the distance between $u$ and $v$ in $\Upsilon$?
 \item If $u,v$ are in the same connected component of $\Upsilon$, return the path in $\Upsilon$ between $u$ and $v$. 
\end{itemize}
These queries take time $\Oh(\log n)$ and $\Oh(\ell \log n)$, respectively, where $\ell$ is the length of the reported path.

Consider now the operation of inserting an edge $uv$. First, we check what is the distance between $u$ and $v$ in $\Upsilon$. If $u$ and $v$ are in different connected components of $\Upsilon$, and therefore also of $G$, then we add $uv$ to $\Upsilon$ and we add a new biconnected component consisting only of $uv$. If $u$ and $v$ are in the same connected component of $\Upsilon$, but the distance between them in $\Upsilon$ is at least $k-1$, then the insertion should be rejected: $uv$ would close a simple cycle of length at least $k$. Otherwise, in time $\Oh(k \log n)$ we can retrieve a path $P\subseteq \Upsilon$ with endpoints $u$ and $v$, and this path has length smaller than $k-1$.

It may happen that the edges of $P$ belong to different biconnected components of $G$. It is easy to see that then, the following should happen to the partition $\Prt$ after the insertion of $uv$: all the biconnected components containing edges of $P$ get merged into one biconnected component. To carry this out, we iterate through the (at most $k-2$) edges on $P$ and if two consecutive edges $xy$ and $yz$ belong to different biconnected components $H_1$ and $H_2$, then we merge $H_1$ and $H_2$. This requires merging the data structures $\Dt[H_1]$ and $\Dt[H_2]$, and in particular computing a recursively optimal elimination forest of the union of $H_1$ and $H_2$. We show that this can be done using the toolbox of cores as follows. Let $T_1$ and $T_2$ be the elimination trees of $H_1$ and $H_2$ stored in $\Dt[H_1]$ and $\Dt[H_2]$, respectively.
\begin{itemize}[nosep]
 \item First, we compute $(d+1)$-cores $K_1$ and $K_2$ of $(H_1,T_1)$ and $(H_2,T_2)$, respectively, each of size $d^{\Oh(d)}=k^{\Oh(k^2)}$. We make sure that $\{x,y\}\subseteq K_1$ and $\{y,z\}\subseteq K_2$.
 \item Letting $K\coloneqq K_1\cup K_2$, we compute a recursively optimal elimination forest $T^K$ of $G[K]$ using the static algorithm of Reidl et al.~\cite{ReidlRVS14}. This takes time $2^{\Oh(d^2)}=2^{\Oh(k^4)}$.
 \item We construct an elimination forest $\wh{T}$ of $H_1\cup H_2$ by attaching both the forests $T_1-K_1$ and $T_2-K_2$ below $T^K$, as in the extension operation.
\end{itemize}
It can be argued that $\wh{T}$ constructed in this manner is a recursively optimal elimination forest of $H_1\cup H_2$. Moreover, all of this can be done in time $2^{\Oh(d^2)}=2^{\Oh(k^4)}$ using our representation of the dynamic treedepth data structure. Note that since the length of $P$ is smaller than $k-1$, we perform at most $k-3$ such merges.

The operation of edge deletion is essentially symmetric: we need to split biconnected components instead of merging them, which can be done analogously. However, there is one issue: the deleted edge $uv$ may belong to $\Upsilon$, in which case it is pointless to query $\Upsilon$ for a $u$-to-$v$ path $P$ along which the splits should be performed. Fortunately, we show that in this case, we can retrieve a suitable path $P$ from the data structure $\Dt[H]$, where $H$ is the biconnected component that contains $uv$. Note that $P$ has length smaller than $k$, for otherwise together with $uv$ it would constitute a simple cycle of length at least $k$. Another caveat is that in order to maintain the invariant that $\Upsilon$ is spanning, after deleting $uv$ from $\Upsilon$ we may need to find a replacement edge and insert it into $\Upsilon$. Again fortunately, the retrieved path $P$ provides at most $k-2$ candidates for such a replacement edge.

%%% Local Variables:
%%% mode: latex
%%% TeX-master: "main"
%%% End:

\section{Preliminaries}\label{sec:prelims}

\paragraph*{Graphs.} We use standard graph notation. For a graph $G$, by $V(G)$ and $E(G)$ we denote the vertex and the edge set of $G$, respectively.
The {\em{open}} and {\em{closed neighborhoods}} of a vertex $u$ are respectively defined as $N_G(u)\coloneqq \{v\colon uv\in E(G)\}$ and $N_G[u]\coloneqq N_G(u)\cup \{u\}$. 
We extend this notation to sets of vertices as follows: for $X\subseteq V(G)$, we write $N_G[X]\coloneqq \bigcup_{x\in X} N_G[x]$ and $N_G(X)\coloneqq N_G[X]\setminus X$.
For a subset of vertices $A$ of a graph $G$, the subgraph {\em{induced}} by $A$, denoted $G[A]$, consists of $A$ and all the edges of $G$ with both endpoints in $A$.
For a vertex $u\in V(G)$, by $G-u$ we denote the graph obtained from $G$ by removing vertex $u$ and all its incident edges.

\paragraph*{Forests.} A {\em{rooted forest}} $F$ is a directed acyclic graph in 
which each vertex $u$ has at most one out-neighbor, called the {\em{parent}} of 
$u$ and denoted by $\pnt_F(u)$. 
A vertex $u$ is a {\em{root}} of $F$ if it has no parent, which we denote by $\pnt_F(u)=\bot$. The set of roots of a forest $F$ is denoted by $\roots_F$.
The in-neighbors of a vertex $u$ are the {\em{children}} of $u$ and their set is 
denoted by $\chld_F(u)$. 
Two vertices of $F$ that either are both roots or have the same parent are called {\em{siblings}}.

If a vertex $v$ is reachable from $u$ by a directed path in $F$ then $v$ is an {\em{ancestor}} of $u$ and $u$ is a {\em{descendant}} of $v$.
Note that each vertex is its own ancestor and descendant.
For a vertex $u$, by $\anc_F(u)$ and $\desc_F(u)$ we denote the sets of ancestors and descendants of $u$, respectively.
By $F_u$ we denote the subtree of $F$ induced by the descendants of $u$. 
The {\em{depth}} of a vertex $u$ in $F$ is $\depth_F(u)\coloneqq |\anc_F(u)|$ and the {\em{height}} of $F$ is $\height(F)\coloneqq \max_{u\in V(F)} \depth_F(u)$.

A subset of vertices $X\subseteq V(F)$ is {\em{straight}} in $F$ if for all $u,v\in X$, either $u$ is an ancestor of $v$ in $F$ or $v$ is an ancestor of $u$ in $F$.
Equivalently, vertices of a straight set lie on one leaf-to-root path in $F$. Here, by a root-to-leaf path in $F$ we mean a path connecting a leaf with the root of some tree in $F$. 

A {\em{prefix}} of a forest $F$ is an ancestor-closed subset of vertices, that is, $A\subseteq V(F)$ is a prefix if $u\in A$ implies $\anc_F(u)\subseteq A$.
The set of {\em{appendices}} of a prefix $A$, denoted $\App_F(A)$, comprises all ancestor-minimal elements of $V(F)\setminus A$, that is, vertices $u\notin A$ such that either $u\in \roots_F$ or $\pnt_F(u)\in A$.
Note that for all $u\in \App_F(A)$, we have $\anc_F(u)\setminus \{u\}\subseteq A$.
If $A$ is a prefix of $F$, then by $F-A$ we denote the forest obtained from $F$ by removing all the vertices of $A$ and keeping the parent/child relation on the remaining vertices intact.

%In all the definitions above, we may omit the subscript if the forest $F$ is clear from the context.

\paragraph*{Elimination forests and treedepth.} An {\em{elimination forest}} of a graph $G$ is a rooted forest $F$ with $V(F)=V(G)$ such that for every edge $uv\in E(G)$, the set $\{u,v\}$ is straight in $F$.
Note that if $G$ is connected then every elimination forest of $G$ has to be a tree.
Hence in such a case we may also speak about {\em{elimination trees}}.

If $F$ is an elimination forest of a graph $G$ and $u\in V(G)$, then we define the {\em{strong reachability set}} of $u$:
$$\SReach_{F,G}(u)\coloneqq N_G(\desc_F(u)).$$
We remark that the name {\em{strong reachability set}} comes from the theory of structural sparsity, where this concept is present and is an analogue of the definition above; 
see e.g.~\cite{KiersteadY03,GroheKRSS18,Zhu09}.
Note that for every vertex~$u$ we have $\SReach_{F,G}(u)\subseteq \anc_F(u)\setminus \{u\}$.

The {\em{treedepth}} of a graph $G$, denoted $\td(G)$, is the minimum height of an elimination forest of $G$.
An elimination forest of $G$ is {\em{optimal}} if its height is equal to the treedepth of $G$. We will need a more refined notion of ``local'' optimality, as expressed next.

\begin{definition}\label{def:recursively-optimal}
 An elimination forest $F$ of $G$ is {\em{recursively optimal}} if for every $u\in V(G)$, we have that:
 \begin{itemize}[nosep]
  \item   the graph $G[\desc_F(u)]$ is connected; and
  \item   the tree $F_u$ is an optimal elimination forest of $G[\desc_F(u)]$.
 \end{itemize}
\end{definition}
 
\noindent We remark that Dvo\v{r}\'ak et al.~\cite{DvorakKT14} also use this definition; they call such elimination forests just ``optimal''.
We will also widely use a weakened version of this definition given below.
\begin{definition}\label{def:recursively-connected}
 An elimination forest $F$ of $G$ is {\em{recursively connected}} if for every $u\in V(G)$, we have that the graph $G[\desc_F(u)]$ is connected.
\end{definition}

Let us note here that if $F$ is an elimination forest of a graph $G$ and $A$ is a prefix of $F$, then $F-A$ is an elimination forest of $G-A$. If $F$ is moreover recursively connected or recursively optimal, then the same can be also said about $F-A$.

We now point out a simple, yet important property of recursively connected elimination forests.

\begin{lemma}\label{lem:parent-SReach}
 Suppose $F$ is a recursively connected elimination forest of a graph $G$. Let $u$ be a vertex of $G$ and $v$ be a child of $u$ in $F$.
 Then $u\in \SReach_{F,G}(v)$.
\end{lemma}
\begin{proof}
 Otherwise, no vertex of $\desc_F(v)$ would have a neighbor in $\desc_F(u)\setminus \desc_F(v)$, so the graph $G[\desc_F(u)]$ would not be connected.
\end{proof}

Clearly, every recursively optimal elimination forest is recursively connected, as we require that explicitly in the definition.
Note also that every recursively optimal elimination forest is in particular optimal, as it optimally decomposes each connected component of the graph.
As shown by Reidl et al.~\cite{ReidlRVS14}, an optimal elimination forest of an $n$-vertex graph of treedepth $d$ can be computed in time $2^{\Oh(d^2)}\cdot n$.
We can use this algorithm as a black-box to show the following.

\begin{lemma}\label{lem:static}
 Given an $n$-vertex graph $G$ of treedepth $d$, a recursively optimal elimination forest of $G$ can be computed in time $2^{\Oh(d^2)}\cdot n^{\Oh(1)}$.
\end{lemma}
\begin{proof}
 Clearly, if $G$ is disconnected, then it suffices to run the algorithm on each connected component of $G$ separately and output the union of the obtained trees.
 
 Suppose then that $G$ is connected.
 Call a vertex $u\in V(G)$ {\em{admissible}} if $\td(G-u)<\td(G)$.
 Observe that such an admissible vertex always exists: if $F$ is any optimal elimination tree of $G$, then the root of $F$ is admissible.
 Moreover, such an admissible vertex can be found in time $2^{\Oh(d^2)}\cdot n^2$ by testing for each $u\in V(G)$ whether $\td(G-u)<\td(G)$, where each such test can be done in time $2^{\Oh(d^2)}\cdot n$ using the algorithm of Reidl et al.~\cite{ReidlRVS14}.
 Finally, if $u$ is admissible, then it is easy to see that a recursively optimal elimination tree of $G$ can be obtained by taking a recursively optimal elimination forest $F'$ of $G-u$ and adding $u$ as the new root, that is, making all the roots of $F'$ into children of $u$.
 Thus, after finding an admissible vertex $u$, one can recurse into $G-u$ and adjust the obtained elimination forest as above.
 
 Since at every step we use $2^{\Oh(d^2)}\cdot n^2$ time to find an admissible vertex, it is straightforward to see that the algorithm runs in time $2^{\Oh(d^2)}\cdot n^3$.
\end{proof}

Observe that every DFS forest of a graph $G$ --- a forest of calls of depth-first search started at an arbitrary vertex of each connected component of $G$ --- is actually an elimination forest of $G$.
Note also that if $G$ does not contain a simple path on $k$ vertices, then every DFS forest of $G$ has depth smaller than $k$, because a root-to-leaf path in a DFS forest is a simple path in the graph.
Hence, we have the following.

\begin{lemma}[\cite{NesetrilM12}, Proposition 6.1]\label{lem:path-td}
  If $G$ is a graph such that $\td(G)\geq k$, then $G$ contains a simple path on $k$ vertices.
\end{lemma}

We remark that in a weak sense, also the converse implication holds: if $G$ contains a simple path on $2^k$ vertices, then $\td(G)>k$; see the discussion before Proposition~6.1 in~\cite{NesetrilM12}. We will not use this fact though.

\medskip

Finally, we will use some basic properties of elimination forests related to the connectivity in graphs. First, the following fact is well-known.

\begin{lemma}\label{lem:top}
 Suppose that $F$ is an elimination forest of a graph $G$ and $A\subseteq V(G)$ is such that $G[A]$ is connected. Then there exists $x\in A$ such that $A\subseteq \desc_F(x)$.
\end{lemma}
\begin{proof}
 Let $x$ be any vertex of $A$ that minimizes $\depth_F(x)$. Let $B\coloneqq A\cap \desc_F(x)$ and suppose, for contradiction, that $A\setminus B$ is non-empty.
 Observe that every vertex $y\in A\setminus B$ can be neither an ancestor of $x$ --- by the minimality of $x$ --- nor a descendant of $x$ --- for it would be included in $B$.
 Hence, for each $y\in A\setminus B$ and $z\in B$, the set $\{y,z\}$ is not straight in $F$.
 As $F$ is an elimination forest of $G$, this implies that there is no edge between $B$ and $A\setminus B$.
 Since $B$ is non-empty due to containing $x$, this is a contradiction to the assumption that $G[A]$ is connected.
\end{proof}

Next, vertices $x$ and $y$ of a graph $G$ are {\em{$k$-vertex-connected}} in $G$ if there exists $k$ internally vertex-disjoint paths with endpoints $x$ and $y$. We will use the following simple claim.

\begin{lemma}\label{lem:connectivity}
 Suppose that $x$ and $y$ are $d$-vertex-connected in a graph $G$. Then for every elimination forest $F$ of $G$ of height at most $d$, the set $\{x,y\}$ is straight in $F$.
\end{lemma}
\begin{proof}
 If $\{x,y\}$ was not straight in $F$, then every path connecting $x$ and $y$ would have to contain an internal vertex that belongs to $\anc_F(x)\cap \anc_F(y)$.
 Since $|\anc_F(x)\cap \anc_F(y)|<\height(F)\leq d$, this implies that there cannot be $d$ internally vertex-disjoint paths connecting $x$ and $y$ in $G$, a contradiction.
\end{proof}

\paragraph*{Dictionaries.} In this paper we will often assume that we are given access to a dictionary data structure on edges of a graph.
This is in principle not an unusual assumption when dealing with dynamic graphs. In our case, however, it becomes an issue, as we would like to carefully keep track of factors dependent on $|V|$ of $|E|$ in our running times. Let us formalize the matter and briefly describe solutions available in the literature.

A \emph{dictionary} data structure $\Lt$ is built on top of some universe of keys $\Univ$. It stores a dynamically changing set of keys together with records (of constant size) associated with them. 
To be more precise, $\Lt$ offers operations $\ins(e,r)$, $\rem(e)$ and $\lookup(e)$, which allow inserting a new key $e$ with its record $r$, removing key $e$ and its record, and looking up the record of $e$. 
The performance of $\Lt$ depends on two parameters: the size of universe $|\Univ|$ 
and the number $|\Lt|$ of keys currently stored. In our applications, we have $\Univ=\binom{V}{2}$ where $V$ is the invariant vertex set of a dynamic graph.
That is, $|U|=\Oh(|V|^2)$, while $|\Lt|=|E|$ is the number of edges in the current graph.

The literature offers, among other, the following implementations that can be applied to our setting:
\begin{itemize}[nosep]
 \item adjacency matrix: $\Oh(|V|^2)$ space, $\Oh(1)$ worst case time per operation;
 \item perfect hashing (FKS-hashing)~\cite{FKShashing}: $\Oh(M)$ space, $\Oh(1)$ expected worst case time per operation, where $M$ is an upper bound on the number of distinct edges that may appear, known a priori;
 \item dynamic perfect hashing~\cite{dynamicHashing}: $\Oh(|E|)$ space, $\Oh(1)$ expected amortized time per operation;
 \item $Y$-fast tries~\cite{WILLARD198381}: $\Oh(|E|)$ space, $\Oh(\log \log |V|)$ amortized time per operation.
\end{itemize}
In our algorithms we do not fix any of the solutions above. Instead, for simplicity we usually assume that we are given an abstract dictionary structure on the edges of the graph that uses $\mdict$ space and guarantees worst-case running time $\tdict$ per operation. Then the running time guarantees in our data structures depend on $\tdict$. We remark that if the assumed dictionary offers operations with {\em{amortized}} running time $\tdict$ per operation, instead of {\em{worst-case}}, then in all our data structures exactly the same complexity bounds can be derived, but for amortized time complexity instead of worst-case.

%%% Local Variables:
%%% mode: latex
%%% TeX-master: "main"
%%% End:

\section{Treedepth cores}\label{sec:cores}

We now introduce the most important definition in this work: the {\em{core}} of an elimination forest of a graph. 
Intuitively, this is a relatively small subset of vertices that retains all the relevant connectivity information that is essential for, say, treedepth computation.

\begin{definition}
 Suppose that $F$ is an elimination forest of a graph $G$. For $q\in \N$, a non-empty prefix $K$ of $F$ is called a {\em{$q$\nobreakdash-core}} of $(G,F)$ if the following condition holds
 for every vertex $u\in \App_F(K)$: 
 for each $X\subseteq \SReach_{F,G}(u)$ with $|X|\leq 2$, there exist at least $q$ siblings $w$ of $u$ such that $w\in K$, $X\subseteq \SReach_{F,G}(w)$, and $\height(F_{w})\geq \height(F_u)$.
\end{definition}
\noindent
See Figure~\ref{fig:core} for an illustration.

\begin{figure}[t]
  \centering
  \includegraphics{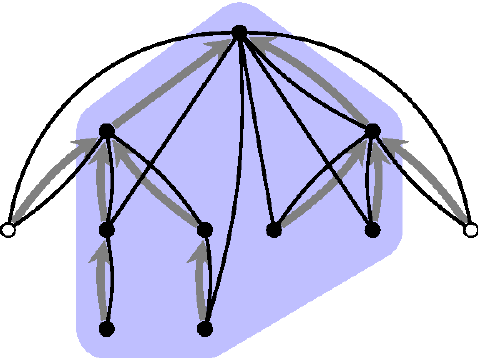}%
  \caption{A graph $G$ (black edges), an elimination tree~$F$ of $G$ (thick gray arcs, directed from children to parents), and a $2$-core of $F$ (vertices on blue background). White vertices are in $\App_F(K)$.}
  \label{fig:core}
\end{figure}

Before we proceed further, let us observe that in recursively optimal elimination forests we can always find $q$-cores of size bounded by a function of $q$ and the height.
In the following, for a set $A$, by $\binom{A}{\leq 2}$ we denote the set of all subsets of $X$ of size at most $2$.

\begin{lemma}\label{lem:find-core}
 Let $F$ be a recursively optimal elimination forest of a graph $G$ and let $d$ be the height of $F$. 
 Then for every $q\in \N$, there is a $q$-core $K$ of $(G,F)$ such that $$|K|\leq q\cdot \frac{\left(q(d^2+1)\right)^d-1}{q(d^2+1)-1}.$$
\end{lemma}
\begin{proof}
 Consider the following recursive marking procedure $\mathtt{recCore}(q,u)$, which can be applied to any vertex $u\in V(G)$.
 For each $X\in \binom{\anc_F(u)}{\leq 2}$, among vertices $w\in \chld_F(u)$ that satisfy $X \subseteq \SReach_F(w)$ mark $q$ with the highest value of $\height(F_w)$,
 or all of them if their number is smaller than $q$. Note that the total number of vertices marked in this way is bounded by $q\cdot |\binom{\anc_F(u)}{\leq 2}|\leq q(d^2+1)$.
 Finally, apply the procedure $\mathtt{recCore}(q,w)$ recursively for every marked child $w$ of $u$.

 Now, let $R\subseteq \roots_F$ comprise $q$ roots $r$ of $F$ with the highest value of $\height(F_r)$, or all of them if their number is smaller than $q$.
 We apply procedure $\mathtt{recCore}(q,r)$ to all $r\in R$, and let $K$ be the set comprising all the vertices marked this way.
 Clearly, for every $i\in \{1,\ldots,d\}$, $K$ contains at most $q\cdot \left(q(d^2+1)\right)^{i-1}$ vertices at depth $i$ in $F$, hence
 $$|K|\leq q\cdot \sum_{i=1}^d \left(q(d^2+1)\right)^{i-1} = q\cdot \frac{\left(q(d^2+1)\right)^d-1}{q(d^2+1)-1},$$
 as claimed. That $K$ is indeed a $q$-core of $(G,F)$ follows directly from the construction.
\end{proof}

In subsequent lemmas we will establish several important properties of cores.
The following notation will be convenient: if $K$ is a prefix in an elimination forest $F$ of $G$, then for each $u\in \App_F(K)$ and $X\in \binom{\SReach_{F,G}(u)}{\leq 2}$, we define $W_K(u,X)$ to be the set of all siblings~$w$ 
of $u$ in $F$ such that:
\begin{itemize}[nosep]
 \item $w\in K$;
 \item $X\subseteq \SReach_{F,G}(w)$; and
 \item $\height(F_{w})\geq \height(F_u)$.
\end{itemize}
Then, provided that $K$ is a $q$-core of $(G,F)$, we have $|W_K(u,X)|\geq q$ for all $u\in \App(K)$ and $X\in \binom{\SReach_{F,G}(u)}{\leq 2}$.
Note that the definition of $W_K(\cdot,\cdot)$ depends on $F$ and $G$; these will always be clear from the context, as $K$ will be a core with respect to some pair $(G,F)$.

As mentioned, the intuition is that a $q$-core for a sufficiently large $q$ stores all the essential information about the graph needed for the purpose of computing its treedepth.
We now formalize this intuition in a series of lemmas. First, we observe that cores retain the essential connectivity property from Definition~\ref{def:recursively-connected}.

\begin{lemma}\label{lem:core-connected}
 Suppose that $F$ is a recursively connected elimination forest of a graph $G$ and $K$ is a $1$-core of $(G,F)$. Then for every $u\in K$, we have that
 \begin{enumerate}[label=(\roman*),ref=(\roman*),nosep]
  \item\label{c:sreach} $\SReach_{F,G}(u)=N_{G[K]}(K\cap \desc_F(u))$ and
  \item\label{c:conn}   the graph $G[K\cap \desc_F(u)]$ is connected.
 \end{enumerate}
\end{lemma}
\begin{proof}
 We first prove~\ref{c:sreach} by induction on $\height(F_u)$.
 The base case when $\height(F_u)=1$ is trivial: then $K\cap\desc_F(u)=\desc_F(u)=\{u\}$ and $\anc_F(u)\subseteq K$, hence $\SReach_{F,G}(u)=N_{G}(u)=N_{G[K]}(u)=N_{G[K]}(K\cap \desc_F(u))$.
 
 Suppose now that $\height(F_u)>1$.
 The inclusion $\SReach_{F,G}(u)\supseteq N_{G[K]}(K\cap \desc_F(u))$ is obvious, so it remains to show the following: if some $v\in \anc_F(u)$, $v\neq u$, 
 has a neighbor in $\desc_F(u)$, then $v$ has also a neighbor in $K\cap \desc_F(u)$. For this, let $a$ be a neighbor of $v$ in $\desc_F(u)$.
 If $a=u$ then we are done, because $u\in K$.
 Otherwise $a\in \desc_F(w)$ for some $w\in \chld_F(u)$, which implies that $v\in \SReach_{F,G}(w)$. 
 As $K$ is a $1$-core in $F$, there exists $w'\in \chld_F(u)\cap K$ such that $v\in \SReach_{F,G}(w')$. 
 By applying the induction hypothesis to $w'$ we infer that $\SReach_{F,G}(w')=N_{G[K]}(K\cap \desc_F(w'))$, hence $v$ has a neighbor $a'$ in $K\cap \desc_F(w')$.
 Then also $a'\in K\cap \desc_F(u)$; this completes the proof of~\ref{c:sreach}.
 
 We now move to the proof of~\ref{c:conn}, which we again perform by induction on $\height(F_u)$.
 As before, the base case when $\height(F_u)=1$ is trivial, as then $G[K\cap \desc_F(u)]$ consists of one vertex $u$. So suppose that $\height(F_u)>1$.
 Since $F$ is recursively connected, by Lemma~\ref{lem:parent-SReach} $u$ has a neighbor in each of the sets $\desc_F(w)$ for $w\in \chld_F(u)$.
 Consider any $w\in \chld_F(u)$ such that $K\cap \desc_F(w)\neq \emptyset$.
 Since $K$ is a prefix of $F$, we have $w\in K$.
 By applying~\ref{c:sreach} and the induction hypothesis to $w$, we infer that $u$ has a neighbor in $K\cap \desc_F(w)$ and $G[K\cap \desc_F(w)]$ is connected.
 Thus, $G[K\cap \desc_F(u)]$ consists of a disjoint union of connected graphs, plus there is vertex $u$ which has neighbors in each of these connected graphs. 
 Hence $G[K\cap \desc_F(u)]$ is connected as well.
 This proves~\ref{c:conn}.
\end{proof}

In the subsequent lemmas we will often consider another graph that is obtained from $G$ by a minor modification within a given core $K$: we may add or remove some edges with both endpoints in $K$, 
but the total number of removals is bounded by some integer $\ell\in \N$.
For this, we introduce the following definition:

\begin{definition}
For a graph~$G$, elimination forest $F$ of $G$, and a prefix $K$ of $F$, a graph $H$ is a {\em{$(K,\ell)$-restricted augmentation}} of~$G$ if the following three conditions hold:
\begin{itemize}[nosep]
 \item $V(H)=V(G)$;
 \item $E(H)\setminus \binom{K}{2} = E(G)\setminus \binom{K}{2}$; and
 \item $|E(G)\setminus E(H)|\leq \ell$.
 \end{itemize}
\end{definition}
Note that the second condition above means that edges from the symmetric difference of $E(H)$ and $E(G)$ have both endpoints in $K$, hence for every $u\notin K$ we have $N_H(u)=N_G(u)$.

We now present the following lemma, which intuitively provides good ``re-attachment points'' for trees obtained by removing a core from an elimination forest.

\begin{lemma}\label{lem:straight}
  Let $F$ be a recursively connected elimination forest of a graph $G$ such that $F$ has height at most~$d$.
  Let $K$ be a $(d+\ell)$-core of $(G,F)$, for some $\ell\geq 0$.
  Suppose that $H$ is a $(K,\ell)$-restricted augmentation of $G$ and let $F^K$ be any elimination forest of $H[K]$ of height at most $d$.
  Then for every $u\in \App_F(K)$, the set $\SReach_{F,G}(u)$ is straight in $F^K$.
\end{lemma}
\begin{proof}
 Consider any pair of distinct vertices $x,y\in \SReach_{F,G}(u)$.
 Denote $W\coloneqq W_K(u,\{x,y\})$; then $|W|\geq d+\ell$. 
 Consider any $w\in W$; recall that $w \in K$ and $x,y\in \SReach_{F,G}(w)$.
 By Lemma~\ref{lem:core-connected}, the graph $G[K\cap \desc_F(w)]$ is connected and both $x$ and $y$ have neighbors in $K\cap \desc_F(w)$ in $G[K]$.
 This implies that in $G[K]$ there is a path $P^{xy}_w$ with endpoints $x$ and $y$ such that every internal vertex of $P^{xy}_w$ belongs to $\desc_F(w)$.
 Since the sets in $\{\desc_F(w)\colon w\in W\}$ are pairwise disjoint, this shows that $x$ and $y$ are $(d+\ell)$-vertex-connected in $G[K]$.
 Since $|E(G)\setminus E(H)|\leq \ell$, at most $\ell$ of the paths $P^{xy}_w$ may contain an edge that does not appear in $H$, hence $x$ and $y$ are $d$-vertex-connected in $H[K]$.
 By Lemma~\ref{lem:connectivity} we infer that $\{x,y\}$ is straight in the elimination forest $F^K$.
 Since $x$ and $y$ were chosen arbitrarily from $\SReach_{F,G}(u)$, we conclude that $\SReach_{F,G}(u)$ is straight in $F^K$, as claimed.
\end{proof}

From Lemma~\ref{lem:straight} we can derive the following claim: restricting an elimination forest to a $d$-core preserves the treedepth of each subgraph induced by a subtree.

\begin{figure}[t]
  \centering
  \includegraphics{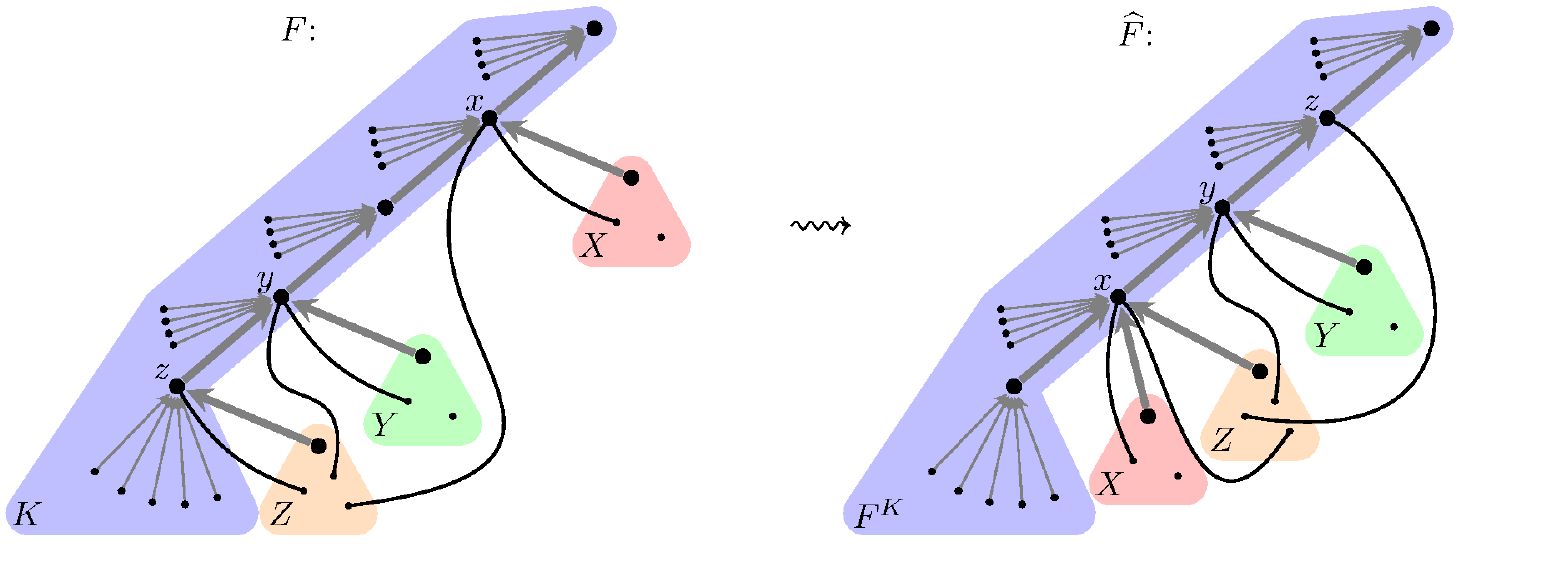}%
  \caption{The reattachment procedure carried out in Lemma~\ref{lem:core-opt-td}.
    Left: A recursively optimal elimination tree $F$ of the graph~$G$ and a core~$K$ of~$F$.
    Edges of $F$ are gray and edges of $G$ are black.
    $X$, $Y$ and $Z$ indicate the subtrees of $F$ that are disjoint from~$K$.
    The edges of $G$ drawn between $K$ and $X$, $Y$ and $Z$ indicate the sets $\SReach_{F, G}$ of the roots of $X$, $Y$, and $Z$.
    Right: An elimination tree $\wh{F}$ constructed from $G$, $F$ and $K$ according to Lemma~\ref{lem:core-opt-td}. We take a recursively optimal elimination tree $F^K$ for $G[K]$ and reattach each subtree $X$, $Y$, and $Z$ at the deepest vertex of the set $\SReach_{F, G}$ of the corresponding root of $X$, $Y$ or~$Z$.}
  \label{fig:core-adapt}
\end{figure}

\begin{lemma}\label{lem:core-opt-td}
 Let $F$ be a recursively optimal elimination forest of a graph $G$ such that $F$ is of height at most $d$, and let $K$ be a $d$-core of $(G,F)$.
 Then for every $v\in K$, we have $\td(G[K\cap \desc_F(v)])=\td(G[\desc_F(v)])$. 
\end{lemma}
\begin{proof}
  Observe that for every $v\in K$, we have that $F_v$ is a recursively optimal elimination tree of $G[\desc_F(v)]$ and $K\cap V(F_v)$ is a $d$-core for $(G[\desc_F(v)],F_v)$.
  Hence, it suffices to prove the lemma for the case when $F$ is a tree and $v$ is the root of $F$.
  Indeed, if we succeed in this, then applying the statement for this case to each $v \in K$ yields the general statement of the lemma.
  
  We proceed by induction on $\height(F)$.
  The base case where $\height(F) = 1$ is trivial.
  For the inductive step, we may assume that
 \begin{equation}\label{eq:squirrel}
\td(G[K\cap \desc_F(z)])=\td(G[\desc_F(z)])\qquad\textrm{for each }z\in K\setminus \{v\}.  
 \end{equation}
 We need to prove that $\td(G[K])=\td(G)$. Clearly $\td(G[K])\leq \td(G)$, so
 for contradiction suppose that $\td(G[K])<\td(G)$.
 Let $F^K$ be an elimination forest of $G[K]$ of height strictly smaller than $\td(G)$; in particular, $\height(F^K)<d$.
 Observe that, thus, $V(G) \setminus K \neq \emptyset$.
 Our goal is to construct an elimination forest $\wh{F}$ of $G$ of height equal to the height of $F^K$, which will be a contradiction.
 
 Consider any $u\in \App_F(K)$.
 By Lemma~\ref{lem:straight} applied to $H=G$, the set $M\coloneqq \SReach_{F,G}(u)$ is straight in~$F^K$.
 Note here that since $u\neq v$ (due to $u\notin K$ and $v\in K$) and $v$ is the only root of $F$, $u$ has a parent in~$F$, which by Lemma~\ref{lem:parent-SReach} belongs to $M$.
 In particular $M\neq \emptyset$. Let then $m$ be the vertex of $M$ that is the deepest in $F^K$; this is well-defined because $M$ is straight in $F^K$.
 Further, let $\wh{M}\coloneqq \anc_{F^K}(m)$; as $\height(F^K)<d$ and $M$ is straight in $F^K$, we have $M\subseteq \wh{M}$ and $|\wh{M}|<d$.

 Let $W\coloneqq W_K(u,\{m\})$; then $|W|\geq d$ as $K$ is a $d$-core.
 Since $|W|\geq d$, there exists $w\in W$ such that $\desc_F(w)\cap \wh{M}=\emptyset$.
 Recall that $m\in \SReach_{F,G}(w)$ and $\height(F_w)\geq \height(F_u)$ by the definition of $W$.
 
 By Lemma~\ref{lem:core-connected}, graph $G[K\cap \desc_F(w)]$ is connected and $N_{G[K]}(K\cap \desc_F(w))=\SReach_{F,G}(w)$.
 As $G[K\cap \desc_F(w)]$ is connected and $F^K$ is an elimination forest of $G[K]$, by Lemma~\ref{lem:top} there exists $x\in K\cap \desc_F(w)$ such that $K\cap \desc_F(w)\subseteq \desc_{F^K}(x)$.
 Note that $x\notin \wh{M}$ because $\desc_F(w)\cap \wh{M}=\emptyset$.
 On the other hand, since $m\in \SReach_{F,G}(w)=N_{G[K]}(K\cap \desc_F(w))$, $m$ has a neighbor in the set $K\cap \desc_F(w)\subseteq \desc_{F^K}(x)$. 
 As $F^K$ is an elimination forest of $G[K]$, this implies that $x\in \desc_{F^K}(m)\setminus \{m\}$.
 Since $K\cap \desc_F(w)\subseteq \desc_{F^K}(x)$ and $M\subseteq \anc_{F^K}(m)$, we conclude that 
 in $F^K$, all the vertices of $K\cap \desc_F(w)$ are descendants of all the vertices of $M$, hence also of all the vertices of $\wh{M}$.
 In particular,
 $$|\wh{M}|+\td(G[K\cap \desc_F(w)])\leq \height(F^K).$$
 
 Since $w\neq v$, applying~\eqref{eq:squirrel} with $z = w$ yields $\td(G[K\cap \desc_F(w)])=\td(G[\desc_F(w)])$. 
  Moreover, observe that we have $\td(G[\desc_F(w)])=\height(F_w)\geq \height(F_u)= \td(G[\desc_F(u)])$, where the equalities follow from the recursive optimality of $F$. By combining these, we find that
 \begin{equation}\label{eq:beaver}
   |\wh{M}|+\height(F_u)\leq \height(F^K).
 \end{equation}
 
 We now construct a new elimination forest $\wh{F}$ of $G$ as follows.
 See Figure~\ref{fig:core-adapt} for an illustration.
 Begin with $\wh{F}=F^K$ and, for every $u\in \App_F(K)$, insert the tree $F_u$ into $\wh{F}$ by making $u$ a child of $m$, defined as the vertex of $\SReach_{F,G}(u)$ that is deepest in $F^K$ (this vertex is well-defined since $\SReach_{F,G}(u)$ is straight in $F^K$). 
 It is straightforward to see that $\wh{F}$ constructed this way is an elimination forest of~$G$.
 Moreover, by~\eqref{eq:beaver} we conclude that $\height(\wh{F})=\height(F^K)$ (recall that $M \subseteq \wh{M}$). As $\height(F^K)<\td(G)$, this is a contradiction, and the inductive step is proved.
\end{proof}

Let us formalize and slightly generalize the re-attachment procedure that we have used in the proof of Lemma~\ref{lem:core-opt-td}.

\begin{definition}
  Suppose that $H$ is a graph, $K$ is a subset of vertices of $H$, and $F^K$ is an elimination forest of $H[K]$. Let $R$ be an elimination forest of the graph $H-K$. 
  We call $R$ {\em{attachable}} to $(H,F^K)$ if for each tree $S$ of the forest $R$, the set $N_H(V(S))$ is straight in~$F^K$. Then we can construct a forest $\wh{F}$ as follows: start with $\wh{F}=F^K$ and for each tree $S$ in the forest $R$, add $S$ to $\wh{F}$ by making the root of $S$ a child of the vertex of $N_H(V(S))$ that is the deepest in $F^K$, or making the root of $S$ a new root of $\wh{F}$ in case $N_H(V(S))=\emptyset$. Then $\wh{F}$ shall be called the {\em{extension}} of $F^K$ via $R$.
\end{definition}

It is straightforward to see the following.

\begin{lemma}\label{lem:attach-elim}
 Suppose $K$ is a subset of vertices of a graph $H$, $F^K$ is an elimination forest of $H[K]$, and $R$ is an elimination forest of $H-K$ that is attachable to $(H,F^K)$. Then the extension of $F^K$ via $R$ is an elimination forest of $H$.
\end{lemma}
\begin{proof}
  Let $\wh{F}$ be the extension in question.
  We need to prove that for every edge $ab\in E(H)$, $\{a,b\}$ is straight in $\wh{F}$.
  If both $a,b\in K$, then this follows from the fact that $F^K$ is an elimination forest of $H[K]$.
  If both $a,b\notin K$, then $ab\in E(H)$ implies that $a,b$ are both contained in the same tree $S$ in $R$, and $\{a,b\}$ is straight in this tree. As this tree is included in $\wh{F}$ without modifications, it follows that $\{a,b\}$ is straight in~$\wh{F}$ as well.
  
  Finally, if say $a\in K$ and $b\notin K$, then there is a tree $S$ of $R$ such that $b\in V(S)$, hence also $a\in N_H(V(S))$. By construction, all the vertices of $N_H(V(S))$ are ancestors of all the vertices of $V(S)$ in~$\wh{F}$. In particular, $a$ is an ancestor of $b$.
\end{proof}

Next, we formulate and discuss a key technical property.
\begin{definition}
Suppose that $H$ is a graph, $K$ is a subset of its vertices, $F^K$ is an elimination forest of $H[K]$, and $R$ is an elimination forest of $H-K$ that is attachable to $(H,F^K)$. Then, letting $\wh{F}$ be the extension of $F^K$ via $R$, we shall say that $R$ has the {\em{sibling substitution property}} in $\wh{F}$ if the following condition holds: 
\begin{quote}
For each $u\in \roots_R$, there exists $z\in K$ such that $u$ and $z$ are siblings in $\wh{F}$ and $\height(F^K_z)\geq \height(R_u)$.
\end{quote}
% \item for any given $p\in \SReach_{\wh{F},H}(u)$, $z$ can be chosen so that in addition $p\in \SReach_{F^K,H[K]}(z)$.
\end{definition}

We now show that the sibling substitution property holds assuming that $K$ is a $q$-core for a sufficiently large~$q$, as explained formally below. Note that in the following lemma we {\em{do not}} require $F^K$ to be recursively optimal, or even recursively connected.

\begin{lemma}\label{lem:ssp}
 Let $F$ be a recursively optimal elimination forest of a graph $G$ such that the height of $F$ is at most $d$, let $K$ be a $(d+\ell+1)$-core of $(G,F)$ for some $\ell\geq 0$, and let $R\coloneqq F-K$.
 Let $H$ be a $(K,\ell)$-restricted augmentation of $G$ and let $F^K$ be any elimination forest of $H[K]$ of height at most $d$.
 Then $R$ is attachable to $(H,F^K)$ and has the sibling substitution property in the extension of $F^K$ via $R$.
\end{lemma}
\begin{proof}
 Note that $R$ is an elimination forest of the graph $G-K$ and $G - K$ is equal to $H-K$. Moreover, from Lemma~\ref{lem:straight} it follows that $R$ is attachable to $(H,F^K)$, because for each tree $S$ in $R$ we have $N_H(V(S))=\SReach_{F,G}(u)$, where $u$ is the root of $S$. Hence, by Lemma~\ref{lem:attach-elim} we conclude that $\wh{F}$, defined as the extension of $F^K$ via $R$, is an elimination forest of $H$. It remains to argue that $R$ has the sibling substitution property. The reasoning will be similar to that used in the proof of Lemma~\ref{lem:core-opt-td}.

  Consider any $u\in \roots_R=\App_F(K)$.
  Let us first consider the corner case when $u\in \roots_F$. Note that then $\SReach_{F,G}(u)=\emptyset$. %, so there is no $p$ to be chosen and the second point of the definition of the sibling substitution property holds vacuously. Hence we are left with arguing the first point.
  Let $W\coloneqq W_K(u,\emptyset)$; then $|W|\geq d+\ell+1$.
  By Lemmas~\ref{lem:core-connected} and~\ref{lem:core-opt-td}, for every $w\in W$ we have that $G[K\cap \desc_F(w)]$ is connected and $\td(G[K\cap \desc_F(w)])=\height(F_w)$.
  Moreover, since $H$ is a $(K,\ell)$-restricted augmentation of $G$ and $|W|\geq d+\ell+1$, there exists a vertex $w\in W$ such that no edge of $E(G)\setminus E(H)$ is incident to any vertex of $\desc_F(w)$ (recall that the vertices in $W$ are mutual siblings in $F$ and thus no edge in $E(G)\setminus E(H)$ can connect descendants of two of them).
  This implies that $H[K\cap \desc_F(w)]$ is a supergraph of $G[K\cap \desc_F(w)]$. Hence, $H[K\cap \desc_F(w)]$ is connected as well and $\td(H[K\cap \desc_F(w)])\geq \td(G[K\cap \desc_F(w)])$.
  As $H[K\cap \desc_F(w)]$ is connected and $F^K$ is an elimination forest of $H[K]$, by Lemma~\ref{lem:top} there exists $x\in K\cap \desc_F(w)$ such that $K\cap \desc_F(w)$ is entirely contained in $F^K_x$. Let $z$ be the root of $F$ that is an ancestor of $x$. We conclude that
  $$\height(F^K_z)\geq \height(F^K_x)\geq \td(H[K\cap \desc_F(w)])\geq \td(G[K\cap \desc_F(w)])=\height(F_w)\geq \height(F_u).$$
  We conclude by observing that both $u$ and $z$ are roots of $\wh{F}$, hence they are siblings in $\wh{F}$. Thus, $z$ satisfies all the requested properties.
  
  \medskip
  
  We proceed to the proof of the main case when $u\notin \roots_F$.
  Let $M\coloneqq \SReach_{F,G}(u)$. By Lemma~\ref{lem:straight}, $M$~is straight in $F^K$.
  Moreover, as $u\notin \roots_F$, from Lemma~\ref{lem:parent-SReach} we infer that $\pnt_F(u)\in M$, hence in particular $M\neq \emptyset$.
  Let then $m$ be the vertex of $M$ that is the deepest in $F^K$.
  Further, let $\wh{M}=\anc_{F^K}(m)$. Then, as $M$ is straight in $F^K$, we have $M\subseteq \wh{M}$ and $|\wh{M}|\leq \height(F^K)\leq d$.
  %In the remainder of the proof we in addition fix any $p\in M$ and we look for a vertex $z$ that in addition to being a sibling of $u$ in $\wh{F}$ and satisfying $\height(F^K_z)\geq \height(F_u)$, also satisfies $p\in \SReach_{F^K,H[K]}(z)$.
  
  Let $W\coloneqq W_K(u,\{m\})$; then $|W|\geq d+\ell+1$.
  Observe that among the vertices $w\in W$, for at most $\ell$ of them there may exist an edge in $E(G)\setminus E(H)$ that is incident to a vertex of $\desc_F(w)$.
  This leaves us with a set $W'\subseteq W$ of size at least $d+1$ such that for each $w\in W'$, we have
  \begin{itemize}[nosep]
   \item $E(H[\desc_F(w)])\supseteq E(G[\desc_F(w)])$ and
   \item $N_H(\desc_F(w))\supseteq N_G(\desc_F(w))=\SReach_{F,G}(w)\ni m$.
  \end{itemize}
  As $|W'|>d$ and $|\wh{M}|\leq d$, there exists a vertex $w\in W'$ such that $\desc_F(w)\cap \wh{M}=\emptyset$.
 
  By Lemma~\ref{lem:core-connected} applied to $w$, we infer that the graph $G[K\cap \desc_F(w)]$ is connected and $N_{G[K]}(K\cap \desc_F(w))=\SReach_{F,G}(w)$.
  As vertices of $\desc_F(w)$ are not incident to the edges of $E(G)\setminus E(H)$, the graph $H[K\cap \desc_F(w)]$ is a supergraph of $G[K\cap \desc_F(w)]$, hence $H[K\cap \desc_F(w)]$ is connected as well.
  Moreover, we have $N_{H[K]}(K\cap \desc_F(w))\supseteq N_{G[K]}(K\cap \desc_F(w))=\SReach_{F,G}(w)$. As $m\in \SReach_{F,G}(w)$, this means that in the graph $H[K]$, $m$ has a neighbor in $K\cap \desc_F(w)$.
 
  Since $H[K\cap \desc_F(w)]$ is connected and $F^K$ is an elimination forest of $H[K]$, by Lemma~\ref{lem:top}
  there exists a vertex $x\in K\cap \desc_F(w)$ such that all vertices of $K\cap \desc_F(w)$ are descendants of $x$ in $F^K$.  
  Since $\desc_F(w)\cap \wh{M}=\emptyset$ by the choice of~$w$, we in particular have $x\notin \wh{M}=\anc_{F^K}(m)$.
  As $F^K$ is an elimination forest of $H[K]$ and, in this graph, $m$ has a neighbor in $K\cap \desc_F(w)\subseteq \desc_{F^K}(x)$, we conclude that $x$ is a strict descendant of $m$ in $F^K$.
  %In particular, as $K\cap \desc_F(w)\subseteq \desc_{F^K}(x)$ and $M\subseteq \anc_{F^K}(m)$, we have that in $F^K$, every vertex of $K\cap \desc_F(w)$ is a descendant of every vertex of $M$.
  By the construction of $\wh{F}$, we can find a vertex $z\in K$ that in $\wh{F}$ is a sibling of $u$ and an ancestor of $x$.
  Note that $K\cap \desc_F(w)\subseteq \desc_{F^K}(x)\subseteq \desc_{F^K}(z)$.
  %Moreover, we have
  %$$\SReach_{F^K,H[K]}(z)=N_{H[K]}(\desc_{F^K}(z))\supseteq N_{H[K]}(K\cap \desc_F(w))\cap \wh{M}\supseteq \SReach_{F,G}(w)\cap M\ni m.$$
  It now remains to observe that
  \begin{align}
    \height(F^K_z) & \geq     \td(H[\desc_{F^K}(z)]) \label{eq:1} \\
                   & \geq  \td(H[K\cap \desc_F(w)]) \nonumber \\
                   & \geq  \td(G[K\cap \desc_F(w)]) \nonumber \\
                   & =     \td(G[\desc_F(w)]) \label{eq:2} \\
                   & =     \height(F_w) \label{eq:3} \\
                   & \geq \height(F_u). \nonumber
  \end{align}
  Here, \eqref{eq:1} follows from the fact that $F^K_z$ is an elimination forest for $H[\desc_{F^K}(z)]$, \eqref{eq:2} follows from Lemma~\ref{lem:core-opt-td}, 
  and \eqref{eq:3} follows from the recursive optimality of $F$.
\end{proof}

We now strengthen Lemma~\ref{lem:attach-elim} by showing that, assuming recursive optimality of $F^K$ and presence of the sibling substitution property, the obtained extension is recursively optimal.

\begin{lemma}\label{lem:ssp-optimal}
 Let $K$ be a subset of vertices of a graph $H$ and $F^K$ be a recursively optimal elimination forest of $H[K]$. Suppose that $R$ is a recursively optimal elimination forest of $H-K$ that is attachable to $(H,F^K)$. Let $\wh{F}$ be the extension of $F^K$ via $R$ and suppose further that $R$ has the sibling substitution property in $\wh{F}$. Then $\wh{F}$ is a recursively optimal elimination forest of $H$ and $\height(\wh{F})=\height(F^K)$. %Moreover, for every $u\in K$, we have $\SReach_{\wh{F},H}(u)=\SReach_{F^K,H[K]}(u)$ and $\height(\wh{F}_u)=\height(F^K_u)$.
\end{lemma}
\begin{proof}
 That $\wh{F}$ is an elimination forest of $H$ follows directly from Lemma~\ref{lem:attach-elim}. We first check that $\wh{F}$ is recursively connected. For this, consider any vertex $u$; we need to prove that $H[\desc_{\wh{F}}(u)]$ is connected. If $u\notin K$ then $\desc_{\wh{F}}(u)=\desc_{R}(u)$, so the claim follows from the recursive optimality of $R$. 
 
  Suppose then that $u\in K$. By the recursive optimality of $F^K$, we have that the graph $H[\desc_{F^K}(u)]=H[K\cap \desc_{\wh{F}}(u)]$ is connected.
  Observe that the connected components of $H[\desc_{\wh{F}}(u)]-K$ are exactly the graphs $H[\desc_F(w)]$ for those $w\in \roots_R$ that satisfy $w\in \desc_{\wh{F}}(u)$ (notice that for such $w$ the graph $H[\desc_F(w)]$ is connected by recursive optimality of~$R$).
  By the construction of $\wh{F}$, and in particular the choice of the attachment point for $w$, for each such $w$ we have that $\SReach_{F,H}(w)$ intersects $K\cap \desc_{\wh{F}}(u)$.
  In particular, some vertex of $\desc_F(w)$ has a neighbor in $H$ that is contained in $K\cap \desc_{\wh{F}}(u)$.
  We conclude that every connected component of $H[\desc_{\wh{F}}(u)]-K$ contains 
  a vertex adjacent to $\desc_{\wh{F}}(u)\cap K$ and $H[K\cap \desc_{\wh{F}}(u)]$ itself is connected, so $H[\desc_{\wh{F}}(u)]$ is connected as well.
  
  Having argued that $\wh{F}$ is recursively connected, it remains to show that for each vertex $u$, we have $\height(\wh{F}_u)=\td(H[\desc_{\wh{F}}(u)])$. Again, if 
  $u\notin K$ then $\wh{F}_u=R_u$, so the claim follows immediately from the recursive optimality of~$R$.
 
  For $u\in K$, from the sibling substitution property we infer that $\height(\wh{F}_u)=\height(F^K_u)$:
  Indeed, adding trees $R_a$ for $a\in \roots_R$ when constructing $\wh{F}$ from $F^K$ cannot increase the height of any subtree of $F^K$, for each $a\in \roots_R$ always has a sibling $z$ in $\wh{F}$ such that already the subtree $F^K_z$ has height not smaller than the height of the added subtree $R_a$. Therefore, by recursive optimality of $F^K$ we conclude that 
  $$\height(\wh{F}_u)=\height(F^K_u)=\td(H[\desc_{F^K}(u)])=\td(H[K\cap \desc_{\wh{F}}(u)])\leq \td(H[\desc_{\wh{F}}(u)]).$$
  However $\height(\wh{F}_u)\geq \td(H[\desc_{\wh{F}}(u)])$ by the definition of treedepth, so indeed we have $\height(\wh{F}_u)=\td(H[\desc_{\wh{F}}(u)])$.
  Applying the above argument for each $u \in K \cap \roots_{\wh{F}}$ yields that $\height(\wh{F})=\height(F^K)$.
\end{proof} 

We may now combine Lemmas~\ref{lem:ssp} and~\ref{lem:ssp-optimal} to obtain the following statement, which summarizes the analysis provided in this section.

\begin{lemma}\label{lem:core-uberlemma}
  Let $F$ be a recursively optimal elimination forest of a graph $G$ such that $F$ is of height at most~$d$.
  Let $K$ be a $(d+\ell+1)$-core of $(G,F)$ for some $\ell\geq 0$, and let $R=F-K$.
 Let $H$ be a $(K,\ell)$-restricted augmentation of $G$ and let $F^K$ be any recursively optimal elimination forest of $H[K]$ of height at most $d$.
 Then:
 \begin{enumerate}[label=(U\arabic*),ref=(U\arabic*),nosep]
 \item\label{u:attachable} $R$ is attachable to $(H,F^K)$; and
 \item\label{u:optimal} the extension of $F^K$ via $R$ is a recursively optimal elimination forest of $H$ of height at most $d$.
 %\item for every $u\in K$, we have $\SReach_{\wh{F},H}(u)=\SReach_{F^K,H[K]}(u)$ and $\height(\wh{F}_u)=\height(F^K_u)$.
 \end{enumerate}
\end{lemma}
\begin{proof}
 By Lemma~\ref{lem:ssp}, $R$ is attachable to $(H,F^K)$ and has the sibling substitution property. Next, since $F$ is a recursively optimal elimination forest of $G$, it follows that $R$ is a recursively optimal elimination forest of the graph $G-K=H-K$. We may now apply Lemma~\ref{lem:ssp-optimal} to infer assertion~\ref{u:optimal}.
\end{proof}

%%% Local Variables:
%%% mode: latex
%%% TeX-master: "main"
%%% End:

\section{Obstructions for treedepth}\label{sec:obstructions}

In this section we consider bounds on the sizes of induced subgraphs that are obstructions for having low treedepth, as explained formally through the following notion.

\begin{definition}
 A graph $G$ is a {\em{minimal obstruction for treedepth $d$}} if $\td(G)>d$, but $\td(G-v)\leq d$ for each~$v\in V(G)$.
\end{definition}

As mentioned in Section~\ref{sec:introduction}, Dvo\v{r}\'ak et al.~\cite{DvorakGT12} proved the following result.

\begin{theorem}[\cite{DvorakGT12}]\label{thm:dvorak-obstructions}
 Let $d\in \N$. Then every minimal obstruction for treedepth $d$ has at most $2^{2^{d-1}}$ vertices.
 Furthermore, there exists a minimal obstruction for treedepth $d$ that has $2^d$ vertices.
\end{theorem}

In fact, the lower bound of $2^d$ provided by Theorem~\ref{thm:dvorak-obstructions} is obtained by showing that every acyclic minimal obstruction for treedepth $d$ has exactly $2^d$ vertices, 
and such obstructions can be precisely characterized by means of an inductive construction. 
This led Dvo\v{r}\'ak et al.~\cite{DvorakGT12} to conjecture that in fact every minimal obstruction for treedepth $d$ has at most $2^d$ vertices.
We now show that from the combinatorial analysis presented in the previous section we can derive an upper bound with asymptotic growth $d^{\Oh(d)}$. 
While this still leaves a gap to the conjectured value of $2^d$,
the new estimate is dramatically lower than the doubly-exponential upper bound provided in~\cite{DvorakGT12}.

\begin{theorem}\label{thm:our-obstructions}
 If $G$ is a minimal obstruction for treedepth $d$, then the vertex count of $G$ is at most 
 $$(d+1)\cdot \frac{\left((d+1)((d+1)^2+1)\right)^{d+1}-1}{(d+1)((d+1)^2+1)-1}\in d^{\Oh(d)}.$$
\end{theorem}
\begin{proof}
 Since $G$ is a minimal obstruction for treedepth $d$, $G$ is connected and $\td(G)=d+1$.
 Let $F$ be a recursively optimal elimination tree of $G$; then $\height(F)=d+1$.
 Let $r$ be the root of $F$.
 By Lemma~\ref{lem:find-core}, we can find a $(d+1)$-core $K$ of $(G,F)$ of size at most $M(d)$, where $M(d)$ is the bound provided in the theorem statement.
 Clearly, $r\in K$.
 Applying Lemma~\ref{lem:core-opt-td} to $v=r$, we find that $\td(G[K])=\td(G)$.
 Since $G$ is a minimal obstruction for treedepth $d$, this means that $K=V(G)$, implying $|V(G)|\leq M(d)$.
\end{proof}

We remark that a more careful analysis of the bounds used in Lemma~\ref{lem:find-core} yields a slightly better upper bound than the one claimed in Theorem~\ref{thm:our-obstructions},
however with the same asymptotic growth of $d^{\Oh(d)}$. It remains open whether this can be improved to an upper bound of the form $2^{\Oh(d)}$.

\section{Splitting and merging}\label{sec:splitting-merging}

\CycleNote

In this section we continue the reasoning from Section~\ref{sec:cores}. While Section~\ref{sec:cores} gave combinatorial foundations for handling recursively optimal elimination forests under edge insertions and deletions, here we will determine how to recompute recursively optimal elimination forests upon {\em{splitting}} or {\em{merging}} graphs along single cut-vertices. This understanding will be needed when designing the data structure for the dynamic \LCycle problem.

%, hence readers interested only in the $k${\sc{-Path}} problem may skip this section.

Throughout this section we will be mainly working with connected graphs, hence we will usually talk about elimination trees instead of elimination forests.

\paragraph*{Biconnectivity.} Recall that in a connected graph $G$, two vertices $u$ and $v$ are {\em{biconnected}} if they are either adjacent, or there exist two internally vertex-disjoint paths connecting them. By Menger's theorem, this is equivalent to requiring that for every vertex $z\notin \{u,v\}$, $u$ and $v$ are in the same connected component of $G-z$. A {\em{biconnected component}} of a graph $G$ is a subgraph of $G$ induced by an inclusion-wise maximal subset of pairwise biconnected vertices. Note that some vertices may be shared by two or more biconnected components, and that isolated vertices do not belong to any biconnected components. These are called {\em{articulation points}} or {\em{cut-vertices}}, and they are precisely the vertices of $G$ whose (singular) removal disconnects the graph.

A connected graph is called {\em{biconnected}} if it has at most one biconnected component. Equivalently, a connected graph is biconnected if it stays connected after removing any single vertex.

The motivating idea behind our data structure for the \LCycle problem is to use the following analogue of Lemma~\ref{lem:path-td}.

\begin{lemma}[\cite{NesetrilM12}, Proposition 6.2]\label{lem:cycle-td}
  If a biconnected graph $G$ does not contain a simple cycle on at least $k$ vertices, then $\td(G)<k^2$.
\end{lemma}

Before we continue, recall Lemma~\ref{lem:core-connected} which stated intuitively that, for any vertex~$u$ in a core~$K$ of a recursively optimal elimination forest, the descendants of $u$ induce a connected graph and $u$'s strong reachability set is the neighborhood in~$K$ of the descendants of~$u$.
We now prove an analogue of Lemma~\ref{lem:core-connected} for biconnectedness.

\begin{lemma}\label{lem:biconn}
 Suppose that $T$ is a recursively connected elimination tree of a connected graph $G$ and $K$ is a $2$-core of $(G,T)$.
 Let $u,v,z\in K$ be three distinct vertices of $K$. Then $u$ and $v$ are in the same connected component of $G-z$ if and only if they are in the same connected component of $G[K]-z$. In particular, for any $u,v\in K$, if $u$ and $v$ are biconnected in $G$ then they are also biconnected in $G[K]$.
\end{lemma}
\begin{proof}
 Obviously if $u$ and $v$ are in the same connected component of $G[K]-z$ then they are in the same connected component of $G-z$ as well, so we only need to prove the converse implication.
 So assume there exists a path $P$ in $G$ whose endpoints are $u$ and $v$ and such that $P$ avoids $z$. Let $Q$ be any inclusion-wise maximal subpath of $P$ consisting of vertices not belonging to $K$. Let $x$ and $y$ be the vertices directly preceding and directly succeeding $Q$ on $P$, respectively. Note that $x$ and $y$ exist due to the assumption that $u,v\in K$. 
 
 Since $Q$ is connected in $G$ and $V(Q)\cap K=\emptyset$, there exists $a\in \App_T(K)$ such that $V(Q)\subseteq \desc_T(a)$ and $x,y\in \SReach_{T,G}(a)$. Since $K$ is a $2$-core of $(G,T)$, there exist at least two siblings of $a$ in $T$ that belong to $K$ and such that both $x$ and $y$ belong to the strong reachability sets of these siblings. Vertex $z$ can be a descendant of at most one of these siblings, which leaves us with a sibling $a'$ of $a$ satisfying the following properties:
 \begin{itemize}[nosep]
  \item $a'\in K$;
  \item $\{x,y\}\subseteq \SReach_{T,G}(a')$; and
  \item $z\notin \desc_T(a')$.
 \end{itemize}
 By Lemma~\ref{lem:core-connected}, the graph $G[K\cap \desc_T(a')]$ is connected and $\{x,y\}\subseteq N_{G[K]}(K\cap \desc_F(a'))$. This implies that in $G[K\cap \desc_F(a')]$ there exists a path $Q'$ such that the first endpoint of $Q'$ is adjacent to $x$ and the second endpoint of $Q$ is adjacent to $y$. Note that in particular $V(Q')\subseteq K$ and  $Q'$ avoids $z$.
 
 Let $P'$ be the path obtained from $P$ by replacing every maximal subpath $Q$ consisting of vertices not belonging to $K$ with the path $Q'$ constructed as above. Then $P'$ is a path in $G[K]$ that connects $u$ with $v$ and avoids $z$, which means that $u$ and $v$ are in the same connected component of $G[K]-z$.
\end{proof}

\paragraph*{Separations.}
A {\em{separation}} in a graph $G$ is a pair of vertex subsets $(A,B)$ such that $A\cup B=V(G)$ and there is no edge with one endpoint in $A\setminus B$ and the second endpoint in $B\setminus A$. Then $A\cap B$ is the {\em{separator}} of $(A,B)$ and the quantity $|A\cap B|$ is the {\em{order}} of $(A,B)$. We will only deal with separations of order one, that is, where the separator consists of only one vertex.

We will use the following simple observation several times.

\begin{lemma}\label{lem:non-mixed}
 Let $T$ be a recursively connected elimination tree of a connected graph $H$ and $(A,B)$ be a separation in $H$.
 Then for every $u\in V(H)$ such that $\desc_T(u)\cap (A\cap B)=\emptyset$, we either have $\desc_T(u)\subseteq A\setminus B$ or $\desc_T(u)\subseteq B\setminus A$.
\end{lemma}
\begin{proof}
 Since $\desc_T(u)\cap (A\cap B)=\emptyset$, the sets $\desc_T(u)\cap (A\setminus B)$ and $\desc_T(u)\cap (B\setminus A)$ form a partition of $\desc_T(u)$.
  Clearly, there is no edge in $H$ with one endpoint in $\desc_T(u)\cap (A\setminus B)$ and second in $\desc_T(u)\cap (B\setminus A)$.
 Since $T$ is recursively connected, the graph $H[\desc_T(u)]$ is connected, hence one of these two sets must be empty. 
\end{proof}

\paragraph*{Inheriting tree structure.}
Suppose $F$ is an elimination forest of a graph $G$.
For a subset of vertices $A\subseteq V(G)$, we define $F|_A$ to be the forest on the vertex set $A$ where the ancestor relation is inherited from $F$: for $u,v\in A$, $u$ is an ancestor of $v$ in $F|_A$ iff $u$ is an ancestor of $v$ in $F$.
In other words, for $u\in A$, the parent of $u$ in $F|_A$ is the vertex of $(\anc_F(u)\setminus \{u\})\cap A$ that has the largest depth in $F$ and $u$ is a root of $F|_A$ if $(\anc_F(u)\setminus \{u\})\cap A = \emptyset$.
Note that $F|_A$ is an elimination forest of $G[A]$.

\paragraph*{Merging.} The following lemma shows how a recursively optimal elimination forest of a graph can be computed when the graph in question is obtained by gluing two subgraphs along a cut-vertex. Here, we assume that we are given suitable cores of elimination forests of these two subgraphs. See Figure~\ref{fig:merging} for an illustration.

\begin{figure}[t]
  \centering
  \includegraphics{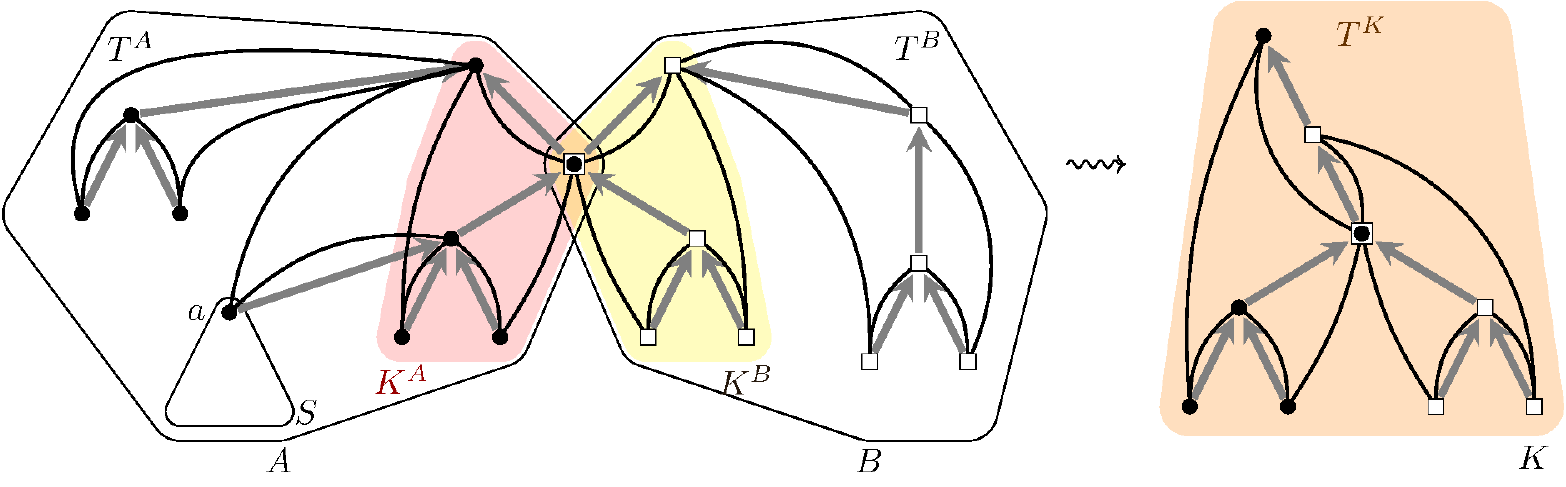}%
  \caption{The situation in Lemma~\ref{lem:merging}. Left: The graph $H$, drawn with black edges.
    Vertices drawn by a filled circle are in~$A$ and the vertices drawn by an empty rectangle are in~$B$.
    The vertices on red background are in~$K^A$ and the vertices on yellow background in~$K^B$.
    The trees~$T^A$ and $T^B$ are drawn with grey edges and encircled with lines labeled~$A$ and~$B$, respectively.
    The forest~$R$ is not labeled and consists of the vertices on white background.
    Right: A recursively optimal elimination tree~$T^K$ for $K = K^A \cup K^B$. Lemma~\ref{lem:merging} shows that $R$ is attachable to~$T^K$. In particular, for the subtree~$S$ of rooted at~$a$ on the left-hand side the set $N_H(S)$ is straight in~$T^K$.}
  \label{fig:merging}
\end{figure}

\begin{lemma}\label{lem:merging}
 Let $H$ be a connected graph of treedepth at most $d$ and $(A,B)$ be a separation of order $1$ in $H$. 
 Suppose that $T^A$ and $T^B$ are recursively optimal elimination forests of $H[A]$ and $H[B]$, respectively. Further, let $K^A$ be a $(d+1)$-core of $(H[A],T^A)$ and $K^B$ be a $(d+1)$-core of $(H[B],T^B)$ such that $A\cap B=K^A\cap K^B$. Let $K\coloneqq K^A\cup K^B$ and suppose $T^K$ is a recursively optimal elimination forest of $H[K]$. Let
 \begin{eqnarray*}
 R & \coloneqq & (T^A-K^A)\uplus (T^B-K^B),
 \end{eqnarray*}
 where $\uplus$ denotes the disjoint union of rooted forests.
 Then 
 \begin{enumerate}[label=(M\arabic*),ref=(M\arabic*),nosep]
  \item\label{m:attachable} $R$ is a recursively optimal elimination forest of $H-K$ that is attachable to $(H,T^K)$;
  \item\label{m:optimal} the extension of $T^K$ via $R$ is a recursively optimal elimination tree of $H$.
%  \item\label{m:info} for each $u\in K$, we have $\SReach_{\wh{T},H}(u)=\SReach_{T^K,H[K]}(u)$ and $\height(\wh{T}_u)=\height(T^K_u)$.
 \end{enumerate}
\end{lemma}
\begin{proof}
 Note that, since $T^A$ is recursively optimal, we have 
 $\height(T^A)= \td(H[A])\leq \td(H)\leq d$.
 Similarly, $\height(T^B)\leq d$.

 Observe that $T^A-K^A$ is a recursively optimal elimination forest of $H[A]-K^A$, and $T^B-K^B$ is a recursively optimal elimination forest of $H[B]-K^B$. Since $A\cap B=K^A\cap K^B$ and in $H$ there are no edges with one endpoint in $A\setminus B$ and second in $B\setminus A$, it follows that $R$ is a recursively optimal elimination forest of $H-K$.
 
 To prove~\ref{m:attachable}, it remains to show that for each tree $S$ in the forest $R$, the set $N_H(V(S))$ is straight in~$T^K$. Since $T^A$ and $T^B$ are recursively connected, it follows that $H[V(S)]$ is connected, hence either $V(S)\subseteq A\setminus K^A$ or $V(S)\subseteq B\setminus K^B$. By symmetry assume the former case. 
 
 Observe that $S=T^A_a$ for some $a\in \App_{T^A}(K^A)$ and $N_H(V(S))=N_{H[A]}(V(S))=\SReach_{T^A,H[A]}(a)$. 
 Recall that $T^A$ is recursively optimal for $H[A]$ and has height at most $d$, and that $K^A$ is a $(d+1)$\nobreakdash-core of $(H[A],T^A)$. In addition to that, $T^K|_{K^A}$ is an elimination forest of $H[K^A]$ of height at most $\height(T^K)=\td(H[K])\leq \td(H)\leq d$.
 Thus, from Lemma~\ref{lem:straight} (in the application of which we set $F$ to $T^A$; both $G$ and $H$ to $H[A]$; $K$ to $K^A$; and $F^K$ to $T^K|_{K^A}$) we infer that $N_{H[A]}(V(S))=N_H(V(S))$ is straight in $T^K|_{K^A}$. Therefore, this set is also straight in $T^K$. This proves~\ref{m:attachable}.
 
 Thus, $\wh{T}$ --- the extension of $T^K$ via $R$ --- is well defined and, as asserted by Lemma~\ref{lem:attach-elim}, is an elimination forest of $H$. Note here that as we assumed that $H$ is connected, in fact $\wh{T}$ must be an elimination tree of $H$. 
 
 Recall that $T^K|_{K^A}$ is an elimination forest of $H[K^A]$ of height at most $d$.
 We may now apply Lemma~\ref{lem:ssp}, wherein we set $F$ to $T^A$; $G$ and $H$ both to $H[A]$; $K$ to $K^A$; $R$ to $T^A - K^A$; and $F^K$ to $T^K|_{K^A}$.
 We thus obtain that $T^A-K^A$ is attachable to $(H[A],T^K|_{K^A})$ and has the sibling substitution property in the extension of $T^K|_{K^A}$ via $T^A-K^A$. Similarly, $T^B-K^B$ is attachable to $(H[B],T^K|_B)$ and has the sibling substitution property in the extension of $T^K|_{K^B}$ via $T^B-K^B$. From these two assertions it follows that in the extension of $T^K$ via $R$, $R$ has the sibling substitution property.
 Combining this with assertion~\ref{m:attachable}, we may now apply Lemma~\ref{lem:ssp-optimal} to infer assertion~\ref{m:optimal}.
\end{proof}

\paragraph*{Splitting.} We now prove the analogue of Lemma~\ref{lem:merging} for splitting instead of merging.
See Figure~\ref{fig:splitting} for an illustration.

\begin{figure}[t]
  \centering
  \includegraphics{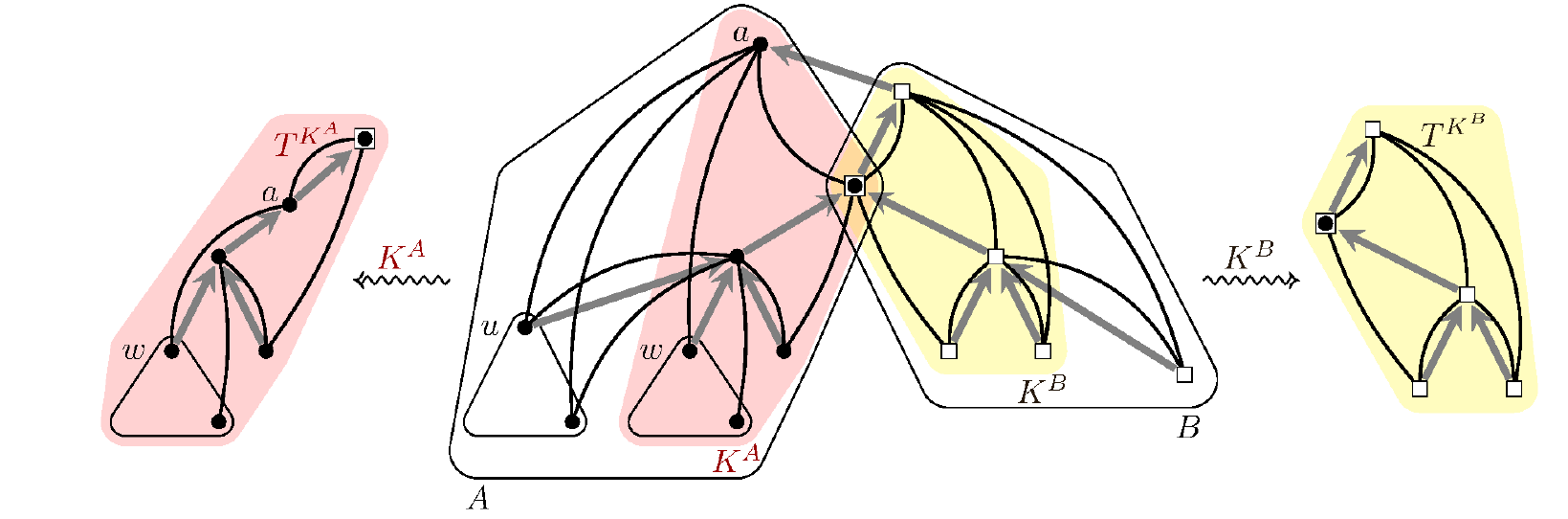}%
  \caption{The situation in Lemma~\ref{lem:splitting}. Center: The
    graph $H$, drawn with black edges, and a recursively optimal elimination tree~$T$, drawn with grey edges.
    For illustrative purposes, a $1$-core is given by the vertices on red and yellow background, though Lemma~\ref{lem:splitting} requires a $(d + 2)$-core.
    Vertices drawn by a filled circle are in~$A$ and the vertices drawn by an empty rectangle are in~$B$.
    The vertices on red background are in~$K^A$ and the vertices on yellow background in~$K^B$.
    The trees~$T^A$ and $T^B$ are drawn with grey edges and encircled with lines labeled~$A$ and~$B$, respectively.
    The forest~$R$ is not labeled and consists of the vertices on white background.
    Left and right: Recursively optimal elimination trees~$T^{K^A}$ and $T^{K^B}$ for $K^A$ and $K^B$, respectively.
    Lemma~\ref{lem:splitting} shows in particular that $H[K^A]$ and $H[K^B]$ are biconnected, and the subtrees of $R$ in~$A$ (resp.\ in $B$) are attachable to~$T^{K^A}$ (resp.\ to $T^{K^B}$).}
  \label{fig:splitting}
\end{figure}

\begin{lemma}\label{lem:splitting}
 % Let $H$ be a connected graph of treedepth at most $d$ and $(A,B)$ be a separation of order $1$ in $H$. Suppose $H[A]$ and $H[B]$ are biconnected. Let $T$ be a recursively optimal elimination tree of $H$ and let $K$ be a $(d+2)$-core of $(H,T)$ that contains $A\cap B$, at least one vertex of $A\setminus B$, and at least one vertex of $B\setminus A$. Further, let $K^A\coloneqq K\cap A$ and $K^B\coloneqq K\cap B$, and let $T^{K^A}$ and $T^{K^B}$ be recursively optimal elimination forests of $H[K^A]$ and $H[K^B]$, respectively. Finally, let $R\coloneqq T-K$, let $R^A$ be the subforest of $R$ consisting of those trees of $R$ whose vertex sets are entirely contained in $A$, and let $R^B$ be the subforest of $R$ consisting of those trees of $R$ whose vertex sets are entirely contained in $B$.
  Let $H$ be a connected graph of treedepth at most $d$ and $(A,B)$ be a separation of order $1$ in $H$. Suppose $H[A]$ and $H[B]$ are biconnected.
  Let:
  \begin{itemize}[nosep]
  \item $T$ be a recursively optimal elimination tree of $H$,
  \item $K$ be a $(d+2)$-core of $(H,T)$ that contains $A\cap B$, at least
    one vertex of $A\setminus B$, and at least one vertex of
    $B\setminus A$,
  \item $K^A\coloneqq K\cap A$ and $K^B\coloneqq K\cap B$,
  \item $T^{K^A}$ and $T^{K^B}$ be 
    recursively optimal elimination forests of $H[K^A]$ and $H[K^B]$,
    respectively,
  \item $R\coloneqq T-K$,
  \item $R^A$ be the subforest of $R$ consisting of those trees of $R$ whose vertex
    sets are entirely contained in $A$, and
  \item $R^B$ be the subforest of $R$ consisting of those trees of $R$ whose vertex sets are entirely contained in $B$.
  \end{itemize}
  \noindent Then for each $C\in \{A,B\}$, we have: 
 \begin{enumerate}[label=(S\arabic*),ref=(S\arabic*),nosep]
  \item\label{s:biconn} $H[K^C]$ is biconnected;
  \item\label{s:attachable} $R^C$ is a recursively optimal elimination forest of $H[C]-K^C$ that is attachable to  $(H[C],T^{K^C})$;
  \item\label{s:optimal} the extension of $T^{K^C}$ via $R^C$ is a recursively optimal elimination tree of $H[C]$.
%  \item\label{s:info} for each $u\in K^C$, we have $\SReach_{\wh{T}^C,H[C]}(u)=\SReach_{T^{K^C},H[K^C]}(u)$ and $\height(\wh{T}^C_u)=\height(T^{K^C}_u)$.
 \end{enumerate} 
\end{lemma}
\begin{proof}
 We give the proof for $C=A$, the proof for $C=B$ is symmetric.
 
 Let us first focus on assertion~\ref{s:biconn}. Take any two distinct vertices $u,v\in K^A$. Since $H[A]$ is biconnected and $K^A\subseteq A$, we have that $u$ and $v$ are biconnected in $H$. By Lemma~\ref{lem:biconn}, $u$ and $v$ are also biconnected in $H[K]$, hence there exist two internally vertex-disjoint paths $P_1,P_2$ that are entirely contained in $H[K]$ and connect $u$ and $v$. Since $(A\cap K,B\cap K)=(K^A,K^B)$ is a separation of order at most $1$ of $H[K]$, it follows that both $P_1$ and $P_2$ are entirely contained in $H[K^A]$. Hence $u$ and $v$ are biconnected in $H[K^A]$. As $u$ and $v$ were selected arbitrarily, this shows that $H[K^A]$ is indeed biconnected, thereby proving~\ref{s:biconn}.
 
 Before we proceed to the next assertions, let us note the following about the roots of $R$.
  
 \begin{claim}\label{cl:trees-unmix}
  For each $u\in \roots_R$, if $u\in A\setminus B$ then $u\in \roots_{R^A}$, $T_u=R^A_u$, and 
  $\desc_T(u)\subseteq A\setminus K\subseteq A\setminus B$.
  Otherwise $u\in \roots_{R^B}$, $T_u=R^B_u$, and $\desc_T(u)\subseteq B\setminus K\subseteq B\setminus A$.
 \end{claim}
 \begin{clproof}
 Since $V(R)\cap (A\cap B)=\emptyset$, either $u\in A\setminus B$ or $u\in B\setminus A$. Therefore, it suffices to consider the first case $u\in A\setminus B$, as the second is symmetric.
 
 Note that $u\in \roots_R$ entails $u\notin K$.
 Since $K$ is a prefix of $T$, we have $\desc_T(u)\cap K=\emptyset$, hence in particular $\desc_T(u)\cap (A\cap B)=\emptyset$. By Lemma~\ref{lem:non-mixed} we conclude that either $\desc_T(u)\subseteq A\setminus B$ or $\desc_T(u)\subseteq B\setminus A$.
 Since $u\in A$, the former case must hold: $\desc_T(u)\subseteq A\setminus B$.
 This implies that $\desc_T(u)$ is entirely contained in $R^A$, $u\in \roots_{R^A}$, and $T_u=R^A_u$. As $\desc_T(u)\cap K=\emptyset$, the inclusion $\desc_T(u)\subseteq A\setminus B$ implies also $\desc_T(u)\subseteq A\setminus K$.
 \end{clproof}

 We now move to assertion~\ref{s:attachable}.
 By Claim~\ref{cl:trees-unmix}, every tree of $R$ is either entirely contained in $R^A$ or entirely contained in $R^B$, which means that $R=R^A\uplus R^B$. Since $T$ is a recursively optimal elimination forest of~$H$ and $K$ is a prefix of~$T$, $R=T-K$ is a recursively optimal elimination forest of $H-K$. It follows that $R^A$ is a recursively optimal elimination forest of the graph $(H-K)[A-K]=H[A]-K^A$. 
 
 To complete the proof of assertion~\ref{s:attachable}, it remains to show the following claim, whose proof is based on the same idea as that of Lemma~\ref{lem:straight}.
 
 \begin{claim}\label{cl:split-attach}
  $R^A$ is attachable to $(H[K^A],T^{K^A})$.  
 \end{claim}
 \begin{clproof}
 Consider any $u\in \roots_{R^A}$. By Claim~\ref{cl:trees-unmix}, we have $T_u=R^A_u$ and $\desc_T(u)\subseteq A\setminus B$. This implies that $N_H(\desc_T(u))\subseteq A$.
 
 We need to prove that $N_H(\desc_T(u))$ is straight in $T^{K^A}$.
 For this, take any distinct $x,y\in N_H(\desc_T(u))$; note that $x,y\in K\cap A=K^A$.
 We now prove that $\{x, y\}$ is straight in~$T^{K^A}$.
 Observe that this implies that $N_H(\desc_T(u))$ is straight in~$T^{K^A}$.
 Since $K$ is a $(d+2)$-core for $(H,T)$ and $u\in \roots_{R^A}\subseteq \App_{K}(T)$, we can find a set $W\coloneqq W_K(u,\{x,y\})$ consisting of at least $(d+2)$ siblings of $u$ in $T$ and satisfying that $\{x,y\}\subseteq\SReach_{T,H}(w)$ for each $w\in W$. As $|A\cap B|=1$, for at most one vertex $w\in W$ we may have $\desc_T(w)\cap (A\cap B)\neq \emptyset$. This leaves us with a set $W'\subseteq W$ of size at least $d+1$ such that in addition, for each $w\in W'$ we have $\desc_T(w)\cap (A\cap B)=\emptyset$.
 
 By Lemma~\ref{lem:non-mixed}, for each $w\in W'$ we have either have $\desc_T(w)\subseteq A\setminus B$ or $\desc_T(w)\subseteq B\setminus A$. Since $x,y\in A$ and at most one of the vertices $x$ and $y$ may belong to $B$ (due to $|A\cap B|=1$), we have $x\in A\setminus B$ or $y\in A\setminus B$. As $x,y\in N_H(\desc_T(w))$, we conclude that the case $\desc_T(w)\subseteq B\setminus A$ is impossible, for we would have an edge with one endpoint in $B\setminus A$ and the other in $A\setminus B$. This means that $\desc_T(w)\subseteq A\setminus B$ for each $w\in W'$.
 
 Since $T$ is recursively connected and $K$ is a $(d+2)$-core for $(H,T)$, from Lemma~\ref{lem:core-connected} we conclude that for each $w\in W'$ we may find a path $P^{xy}_w$ that connects $x$ with $y$ and such that all the internal vertices of $P^{xy}_w$ belong to $\desc_T(w)\cap K$. Since $\desc_T(w)\subseteq A\setminus B$, each such path $P^{xy}_w$ is entirely contained in the graph $H[K\cap A]=H[K^A]$.
 Thus, as the vertices in $W'$ are pairwise siblings, the paths in $\{P^{xy}_w \mid w \in W'\}$ are pairwise internally vertex-disjoint.
 This means that $x$ and $y$ are $(d+1)$-vertex-connected in~$H[K^A]$.
 As $T^{K^A}$ is an elimination forest of $H[K^A]$ of height at most $d$ (due to $\height(T^{K^A})=\td(H[K^A])\leq \td(H)\leq d$), from Lemma~\ref{lem:connectivity} we infer that $\{x,y\}$ is straight in $T^A$. As $x$ and $y$ were chosen arbitrarily from $N_H(\desc_T(u))$, we conclude that $N_H(\desc_T(u))$ is straight in $T^A$.
 \end{clproof}
 
 Claim~\ref{cl:split-attach} completes the proof of assertion~\ref{s:attachable}. Thus, we may define $\wh{T}^A$ to be the extension of $T^{K^A}$ via~$R^A$. As asserted by Lemma~\ref{lem:attach-elim}, $\wh{T}^A$ is an elimination tree of $H[A]$.
 To show assertion~\ref{s:optimal} it remains to show that $\wh{T}^A$ is recursively optimal. For this we will use the following claim, whose proof follows the same idea as that of Lemma~\ref{lem:ssp}.
 
 \begin{claim}\label{cl:split-ssp}
  $R^A$ has the sibling substitution property in $\wh{T}^A$.
 \end{claim}
 \begin{clproof}
 Take any $u\in \roots_{R^A}$. 
 By Claim~\ref{cl:trees-unmix}, we have $T_u=R^A_u$ and $\desc_T(u)\subseteq A\setminus K\subseteq A\setminus B$.
 Note that $\SReach_{\wh{T}^A,H}(u)=N_H(\desc_{\wh{T}^A}(u))=N_H(\desc_T(u))$.
 Recall that $K$ contains $A\cap B$ and at least one vertex of $A\setminus B$, hence $|A\cap K|\geq 2$. Since $\desc_T(u)\subseteq A\setminus K$ it follows that the set $M\coloneqq N_H(\desc_T(u))$ is contained in $A$ and, because  $H[A]$ is biconnected by precondition of the lemma, $M$ has size at least $2$.
 
 By Claim~\ref{cl:split-attach}, $M$ is straight in $T^{K^A}$. Let $m$ be the vertex of $M$ that is the deepest in $T^{K^A}$, and let $\wh{M}\coloneqq \anc_{T^{K^A}}(m)$. Note that $M\subseteq \wh{M}$ and $|\wh{M}|\leq \height(T^{K^A})\leq d$. 
 Since $|M|\geq 2$, we may fix some $p\in M$ that is different from $m$.
 
% to prove the claim it suffices to argue the following statement: for any fixed $p\in M$, $p\neq m$, there exists a sibling $z$ of $u$ in $\wh{T}$ such that $z\in K^A$, $\height(T^{K^A}_z)\geq \height(T_u)$, and $\{p,m\}\subseteq \SReach_{F^{K^A},H[K^A]}(z)$.
 
 Recall that $K$ is a $(d+2)$-core of $(T,H)$. Therefore, there exists a set $W\coloneqq W_K(u,\{m,p\})$ of $d+2$ siblings of $u$ in $T$ such that for each $w\in W$, we have $w\in K$, $\{m,p\}\subseteq \SReach_{T,H}(w)$, and $\height(T_w)\geq \height(T_u)$. Observe that there are at most $|\wh{M}|\leq d$ vertices $w\in W$ that satisfy $\desc_T(w)\cap \wh{M}\neq \emptyset$, and at most one vertex $w\in W$ satisfying $\desc_T(w)\cap (A\cap B)\neq \emptyset$. Therefore, there exists at least one vertex $w\in W$ satisfying $\desc_T(w)\cap \wh{M}=\emptyset$ and $\desc_T(w)\cap (A\cap B)=\emptyset$.
 Fix such a vertex $w \in W$ from now~on.
 
 Since $\desc_T(w)\cap (A\cap B)=\emptyset$, by Lemma~\ref{lem:non-mixed} we have that either $\desc_T(w)\subseteq A\setminus B$ or $\desc_T(w)\subseteq B\setminus A$. On the other hand, since $m,p\in M\subseteq A$, $m\neq p$, and $|A\cap B|=1$, either $m\in A\setminus B$ or $p\in A\setminus B$. As $\{m,p\}\subseteq \SReach_{T,H}(w)=N_H(\desc_T(w))$, we conclude that the case $\desc_T(w)\subseteq B\setminus A$ is impossible, for it would imply the existence of an edge with one endpoint in $A\setminus B$ and the other in~$B\setminus A$.
 Therefore, $\desc_T(w)\subseteq A\setminus B$. Note that in particular $w\in K\cap A=K^A$.
 
 By Lemma~\ref{lem:core-connected}, the graph $H[K\cap \desc_T(w)]$ is connected and $\SReach_{T,H}(w)=N_{H[K]}(K\cap \desc_T(w))$. Since $m\in \SReach_{T,H}(w)$, this in particular means that $m$ has a neighbor in $K\cap \desc_T(w)$.
 
 Since $H[K\cap \desc_T(w)]$ is connected, $K\cap \desc_T(w)\subseteq K\cap A=K^A$, and $T^{K^A}$ is an elimination forest of $H[K^A]$, from Lemma~\ref{lem:top} we infer that there exists a vertex $x\in K\cap \desc_T(w)$ such that $K\cap \desc_T(w)\subseteq \desc_{T^{K^A}}(x)$. As $\desc_T(w)\cap \wh{M}=\emptyset$ by the choice of~$w$, we have $x\notin \wh{M}$.
 Since $m$ has a neighbor in $K \cap \desc_T(w)$ and $K \cap \desc_T(w) \subseteq \desc_{T^{K^A}}(x)$ we have that $\desc_{T^{K^A}}(x)$ contains a neighbor of~$m$.
 Since $x\notin \wh{M}=\anc_{T^{K^A}}(m)$, it follows that $x$ must be a strict descendant of $m$ in $T^{K^A}$. In particular, there exists a sibling $z$ of $u$ in $\wh{T}$ such that $x$ is a descendant of $z$ in $T^{K^A}$. Note that $K\cap \desc_T(w)\subseteq \desc_{T^{K^A}}(x)\subseteq \desc_{T^{K^A}}(z)$. 
 %Therefore, we have
 %$$\SReach_{T^{K^A},H[K^A]}(z)\supseteq N_{H[K^A]}(K\cap \desc_T(w))\cap \wh{M}=N_{H[K]}(K\cap \desc_T(w))\cap \wh{M}\supseteq \{m,p\}.$$
 It now remains to observe that
  \begin{align}
    \height(T^{K^A}_z) & \geq  \td(H[\desc_{T^{K^A}}(z)]) \label{eq:s1} \\
                       & \geq  \td(H[K\cap \desc_T(w)]) \nonumber \\
                       & =     \td(H[\desc_T(w)]) \label{eq:s2} \\
                       & =     \height(T_w) \label{eq:s3} \\
                       & \geq  \height(T_u) \nonumber \\
                       & =     \height(R_u). \nonumber
  \end{align}
  Here, \eqref{eq:s1} follows from the fact that $T^{K^A}_z$ is an elimination tree of the subgraph of $H$ induced by its vertices, \eqref{eq:s2} follows from Lemma~\ref{lem:core-opt-td}, and \eqref{eq:s3} follows from the recursive optimality of $T$.
 \end{clproof}
 
 Now, the remaining assertion~\ref{s:optimal} follows directly from applying Lemma~\ref{lem:ssp-optimal}. Here, the recursive optimality of $R^A$ has been argued in assertion~\ref{s:attachable}, the recursive optimality of $T^A$ is assumed in the lemma statement, and the sibling substitution property is provided by Claim~\ref{cl:split-ssp}.
\end{proof}

%%% Local Variables:
%%% mode: latex
%%% TeX-master: "main"
%%% End:

\section{Data structure}\label{sec:data-structure}

In this section we present our data structure for maintaining a recursively optimal elimination forest of a dynamic graph of bounded treedepth.
Before we proceed to the details, let us clarify the computation model.
We assume the standard word RAM model of computation with words of bitlength $\Oh(\log n)$, where $n$ is the vertex count of the input graph.
Further, we assume that the vertices' identifiers fit into single machine words, hence they take unit space and can be operated on in constant time.
Edges are represented as pairs of vertices.

In all our data structures we assume that the initialization method takes the number $n$ as part of the input and constructs the data structure for an edgeless graph on $n$ vertices.
Of course, if one wishes to initialize the structure for a graph given on input, it suffices to initialize the edgeless graph of appropriate order and add all the edges by a repeated use of the insertion method.

\paragraph*{Description of the data structure.}
We now present a data structure $\Dt[F,G]$ that stores an elimination forest $F$ of a graph $G$. We will always assume that $F$ is recursively connected and its height is bounded by a given parameter $d$. Under this assumption, the data structure uses $\Oh(nd)$ space.

In $\Dt[F,G]$, each vertex $u$ is associated with a record consisting of:
\begin{itemize}[nosep]
 \item a pointer $\tp(u)$ which points to a memory cell that stores $\pnt_F(u)$;
 \item a set $\SReach(u)$ equal to $\SReach_{F,G}(u)$;
 \item a set $\Up(u)$ equal to $\Up_{F,G}(u)\coloneqq N_G(u)\cap \SReach_{F,G}(u)$;
 \item a number $\height(u)$ equal to $\height(F_u)$;
 \item for each $X\subseteq \SReach_{F,G}(u) \cup \{ u \}$ and $i\in \{1,\ldots,d\}$, the {\em{bucket}}
       $$\bucket[u,X,i]\coloneqq \{\ w\ \colon\ w\in \chld_F(u),\ \SReach_{F,G}(w)=X,\ \height(F_w)=i\ \}.$$
\end{itemize}
Note that the buckets $\bucket[u,\cdot,\cdot]$ form a partition of the children of $u$. 

Sets $\SReach_{F,G}(u)$, $\Up_{F,G}(u)$, as well as all the buckets $\bucket[u,X,i]$ are stored as doubly linked lists, where a doubly linked list is represented as a pair of pointers: to its first and last element.
This representation is not essential for $\SReach_{F,G}(u)$ and $\Up_{F,G}(u)$, as these sets have sizes at most~$d$ anyway, but is important for the buckets, as their sizes are unbounded.
We only store non-empty buckets, so the storage space for the buckets sums up to $\Oh(n)$. Sets $\SReach()$ and $\Up()$ stored with each vertex increase the total space to $\Oh(nd)$.

In addition to the above, we assume that for every bucket $\bucket[u,X,i]$ there is a single memory cell $\parent[u,X,i]$ that stores $u$, and that $\tp(w)$ points to $\parent[u,X,i]$ for each $w\in \bucket[u,X,i]$.
That is, all elements of $\bucket[u,X,i]$ point to the same memory cell $\parent[u,X,i]$ for storing the information on their parent.
Thus, changing the parent of the whole bucket can be done in constant time.

Furthermore, we store also the roots of $F$ in buckets as follows. It is simpler to think of the buckets (and to implement our algorithms) on a tree rather than a forest. We hence introduce an artificial symbol $\bot$, which represents an artificial root connecting all trees in $F$, i.e., all roots in $F$ are treated as its children. Now for each $i\in \{1,\ldots,d\}$, we create a bucket
$$\bucket[\bot,\emptyset,i]\coloneqq \{\ r\ \colon\ r\in \roots_F,\ \height(F_r)=i\ \}.$$
Thus, the buckets $\bucket[\bot,\emptyset,\cdot]$ form a partition of $\roots_F$.
These buckets are not associated with any vertex of $G$: they are stored at the global level in $\Dt[F,G]$, again as doubly linked lists.
Also, with each of these buckets we associate a memory cell $\parent[\bot,\emptyset,i]$ which stores $\bot$ and is pointed to by $\tp(w)$ for all $w\in \bucket[\bot,\emptyset,i]$.

Thus, each vertex $w\in V(G)$ is stored in exactly one bucket, namely 
$$w\in \bucket\left[\pnt_F(w),\SReach_{F,G}(w),\height(F_w)\right].$$
In addition to the record described in the beginning, for every vertex $w$ we store a pointer to the list element corresponding to $w$ in the doubly linked list representing the bucket in which $w$ resides.
Note that this allows removing $w$ from this bucket in constant time.

\medskip

This concludes the description of the data structure $\Dt[F,G]$. 
It is clear that the initialization for an edgeless graph $G$ can be done in 
$\Oh(n)$ time, as one only needs to initialize $\Oh(n)$ empty buckets.

We note that the edges of the graph $G$ are stored in $\Dt[F,G]$ implicitly: 
given $u,v\in V(G)$, to verify whether $u$ and $v$ are adjacent in $G$ it 
suffices to check whether $u\in \Up_{F,G}(v)$ or $v\in \Up_{F,G}(u)$, which can 
be done in $\Oh(d)$ time.
Thus, one can think of $\Dt[F,G]$ as an implicit representation of $G$ as well.
Also, in the following we will repeatedly use the fact that, given $u\in V(G)$ 
and access to $\Dt[F,G]$, the set $\anc_F(u)$ can be computed in $\Oh(d)$ time 
by iteratively following parent pointers from $u$.

\paragraph*{Extracting cores.} We now show that, given the data structure $\Dt[F,G]$, one can efficiently extract small cores from it. 
The argument essentially boils down to implementing the procedure outlined in the proof of Lemma~\ref{lem:find-core} using access to $\Dt[F,G]$.

\begin{lemma}\label{lem:extract-core}
 Suppose we have access to a data structure $\Dt[F,G]$ that stores a recursively connected elimination forest $F$ of a graph $G$ of height at most $d$.
 Then one can implement a method $\mathtt{core}(L,q)$ which, given $L\subseteq V(G)$ and $q\in \N$,
 in $(d+q+|L|)^{\Oh(d)}$ time computes a $q$-core $K$ of 
$(G,F)$ satisfying $L\subseteq K$ and $|K|\leq (d+q+|L|)^{\Oh(d)}$.
 Moreover, within the same running time one can also construct the graph $G[K]$. 
\end{lemma}
\begin{proof}
 The algorithm is presented using pseudocode as Algorithm~\ref{alg:kernel}. 
 We implement it as method $\mathtt{core}(L,q)$ of $\Dt[F,G]$: this method, given $L$ and $q$, outputs a $q$-core $K$ with the desired properties.
 
 The first step of $\mathtt{core}(L,q)$ is to compute the ancestor closure 
$\wh{L}\coloneqq \anc_F(L)$; this can be done in $\Oh(d|L|)$ time.
 Next, we call a recursive method $\mathtt{recCore}(\wh{L},q,u)$, presented using pseudocode as Algorithm~\ref{alg:rkernel}.
 This method is a slight generalization of the procedure $\mathtt{recCore}(q,u)$ outlined in the proof of Lemma~\ref{lem:find-core}, where we are additionally given the set $\wh{L}$ that should be included in the computed core.
 Precisely, the method $\mathtt{recCore}(\wh{L},q,u)$ is given a vertex $u$ and should output a $q$-core of $(G[\desc_F(u)],F_u)$ that contains $\wh{L}\cap \desc_F(u)$; this output is represented as a doubly linked list.
 In the initial call, the vertex $u$ is substituted with the marked $\bot$, and we follow the convention that $\desc_F(\bot)=V(G)$ and $F_\bot=F$.

 The method $\mathtt{recCore}(\wh{L},q,u)$ is implemented as follows.
 The first step is to gather a list $R$ consisting of children of $u$ (or roots of $F$ in case $u=\bot$) into which the construction of the core should recurse; the implementation is in
 Lines~\ref{alg:mark:start}-\ref{alg:mark:end} of Algorithm~\ref{alg:rkernel}.
 Into this list we first include all vertices $w\in \wh{L}$ for which 
$\pnt_F(w)=u$; this can be easily done in $\Oh(|\wh{L}|)=\Oh(d|L|)$ time.
 Next, for each $X\subseteq \SReach_{F,G}(u) \cup \{ u \}$ such that $|X|\leq 2$, 
 we consider all vertices $w\in \chld_F(u)$ (or $w\in \roots_F$ if $u=\bot$) satisfying $\SReach_F(w)\supseteq X$.
 From those vertices we add to $R$ any $q$ with the highest value of $\height(F_w)$, or all of them if their total number is smaller than $q$.
 Note that, for each $X$, this can be done in time $\Oh(q(q+|L|)+d\cdot 2^{d})$ by inspecting all the 
 buckets $\bucket[u,Y,i]$ satisfying $Y\supseteq X$ in any order with non-increasing $i$, and iteratively including vertices from
 them until a total number of $q$ vertices has been included. In order to avoid repetitions on the list, whenever inserting a new vertex into $R$, we check whether it has not been included before; this takes time $\Oh(q+|L|)$.

 Once the list $R$ is constructed, method $\mathtt{recCore}(\wh{L}, q, w)$ is applied recursively to each $w\in R$.
 The return list is the concatenation of all the lists obtained from the recursion, with $u$ appended in addition.
 
 It is clear that the algorithm correctly constructs a $q$-core $K$ of $(G,F)$ which contains $L$.
 As for the bound on $|K|$, observe that in procedure $\mathtt{recCore}(\wh{L},q,u)$, the total number of vertices included in the list $R$ is bounded by $q\cdot \left|\binom{\anc_F(u)}{\leq 2}\right|+|L|\leq q(d^2+1)+|L|\eqqcolon i$.
 The recursion depth is bounded by the height of $F$, which is at most $d$, hence the total number of nodes in the recursion tree is bounded by
 $$i^0+i^1+\ldots+i^d\in (d+q+|L|)^{\Oh(d)}.$$
 Observe that for each node of the recursion tree exactly one vertex is added to~$K$.
 Thus we conclude that $|K|\leq (d+q+|L|)^{\Oh(d)}$, as claimed.
 Finally, the internal computation for each node takes time $\Oh(q(q+|L|)+d\cdot 2^d)$, which is asymptotically dominated by the bound on $|K|$. Hence, the total 
running time is $(d+q+|L|)^{\Oh(d)}$.
 
 In order to construct $G[K]$ from $K$ within the same asymptotic running time, it suffices to observe that the edge set of $G[K]$ is exactly $\bigcup_{u\in K} \{uv\colon v\in \Up_{F,G}(u)\}$, 
 so it can be constructed in $\Oh(d\cdot |K|)$ time given access to $\Dt[F,G]$.
\end{proof}

\begin{algorithm}\label{alg:kernel}
    
    \SetKwInOut{Input}{Input}
    \SetKwInOut{Output}{Output}

  %  \vskip 0.2cm
    
    \Input{A subset of vertices $L$ and a positive integer $q$}
    \Output{A $q$-core $K$ of $(G,F)$ such that $L\subseteq K$}
    
    \vskip 0.1cm
    
    $\wh{L}\gets \anc_F(L)$\\
    \Return{$\mathtt{recCore}(\wh{L},q,\bot)$}
    
    \caption{method $\mathtt{core}(L,q)$}
\end{algorithm}

\begin{algorithm}\label{alg:rkernel}
    
    \SetKwInOut{Input}{Input}
    \SetKwInOut{Output}{Output}
    \SetKw{Break}{break}

  %  \vskip 0.2cm
    
    \Input{A prefix $\wh{L}$ of $F$, a positive integer $q$, and a vertex $u\in V(G)\cup \{\bot\}$}
    \Output{A $q$-core $K_u$ of $(G[\desc_F(u)],F_u)$ such that $\wh{L}\cap \desc_F(u)\subseteq K_u$}
    
    \vskip 0.1cm
    
    $R\gets \mathtt{new\ List}()$\label{alg:mark:start}  \\   
    \ForEach{$w\in \wh{L}$}{
       \If{$^{\star}(\tp(w))=u$} { $R.\mathtt{append}(w)$ }
    }
    \ForEach{$X\subseteq \SReach_{F}(u) \cup \{ u \}$ \normalfont{\textbf{such that}} $|X|\leq 2$} {
      $c \gets q$\\
      \For{$i\gets d$ \normalfont{\textbf{downto}} $1$} {
         \ForEach{$Y\subseteq \SReach_{F}(u) \cup \{ u \}$ \normalfont{\textbf{such that}} $Y\supseteq X$} {
            \ForEach{$w\in \bucket[u,Y,i]$} {
               \If{$w\notin R$}
	         {$R.\mathtt{append}(w)$}
	       $c\gets c-1$\\
	       \If{$c=0$}{\textbf{goto} exit}
	    }
         }
      }
      {exit }\label{alg:mark:end}
    }
    $K_u\gets \mathtt{new\ List}()$\\
    \ForEach{$w\in R$} {
      $K_u.\mathtt{append}(\mathtt{recCore}(\wh{L},q,w))$
    }
    \If{$u\neq \bot$}{
    $K_u.\mathtt{append}(u)$}
    \Return{$K_u$}

    \caption{method $\mathtt{recCore}(\wh{L},q,u)$}
\end{algorithm}

\newcommand{\appendices}{\mathsf{apps}}

\paragraph*{Partial structures.} We will be often faced with a situation where we have a graph $G$, a recursively connected elimination forest $F$ of $G$, and a prefix $K$ of $F$ (usually, a suitable core). Given access to the data structure $\Dt[F,G]$, we would like to update the part of $\Dt[F,G]$ that corresponds to the induced subgraph $G[K]$. The key point in the forthcoming implementation is that the remainder of $\Dt[F,G]$, i.e. the part that does not concern $G[K]$, does not need to be updated at all. Therefore, it is convenient to define a {\em{partial data structure}}, which represents the part of $\Dt[F,G]$ which does not need to be updated.

Let $G$ be a graph, $K$ be a subset of its vertices, and $R$ be a recursively connected elimination forest of the graph $G-K$. We define the {\em{partial data structure}} $\Dt[R,G/K]$ as follows. First, for each vertex $u\in V(G)-K$ we define $\SReach_{R,G}(u)\coloneqq N_G(\desc_R(u))$; note that $\SReach_{R,G}(u)$ {\em{may}} contain vertices of $K$. Then for each vertex $u\in V(G)-K$, in $\Dt[R,G/K]$ we store:
\begin{itemize}[nosep]
 \item a pointer $\tp(u)$ which points to a memory cell that stores $\pnt_R(u)$;
 \item a set $\SReach(u)$ equal to $\SReach_{R,G}(u)$;
 \item a set $\Up(u)$ equal to $\Up_{R,G}(u)\coloneqq N_G(u)\cap \SReach_{R,G}(u)$;
 \item a number $\height(u)$ equal to $\height(R_u)$;
 \item for each $X\subseteq \SReach_{R,G}(u) \cup \{ u \}$ and $i\in \{1,\ldots,d\}$, the bucket
       $$\bucket[u,X,i]\coloneqq \{\ w\ \colon\ w\in \chld_R(u),\ \SReach_{R,G}(w)=X,\ \height(R_w)=i\ \}.$$
\end{itemize}
Also, with each bucket $\bucket[u,X,i]$ we keep a memory cell $\parent[u,X,i]$ that stores $u$ and is pointed to (via $\tp(\cdot)$) by all the elements of $\bucket[u,X,i]$. 
Note that this is exactly the same information as in the definition of the (non-partial) data structure $\Dt[\cdot,\cdot]$, except that we interpret strong reachability sets for $u\in V(G)\setminus K$ using the $u\mapsto \SReach_{R,G}(u)$ operator.

Furthermore, in $\Dt[R,G/K]$ we store a list $\appendices$ consisting of buckets of the form $\bucket[\bot,X,i]$. We assume that the following assertion holds: the buckets $\bucket[\bot,X,i]$ placed in $\appendices$ have pairwise different pairs $(X,i)$. The buckets in $\appendices$ contain a partition of the roots of $R$: each $r\in \roots_R$ is placed in the bucket $\bucket[\bot,\SReach_{R,G}(r),\height(R_r)]$.
As before, each bucket $\bucket[\bot,X,i]\in \appendices$ is supplied with a memory cell $\parent[\bot,X,i]$ that stores $\bot$ and is pointed to by all the elements of $\bucket[\bot,X,i]$. Note that this time, the number of buckets on the list $\appendices$ is potentially unbounded; we store only the non-empty ones. Thus, again we have the following property: in $\Dt[R,G/K]$, every vertex of $V(G)\setminus K$ is placed in exactly one bucket.

We now show two lemmas about partial data structures. The first one says that given a (non-partial) data structure $\Dt[F,G]$ and a prefix $K$ of $F$, we may efficiently modify $\Dt[F,G]$ to obtain $\Dt[F-K,G/K]$. The second one gives a modification procedure in the other direction: given $\Dt[R,H/K]$ and an elimination forest $F^K$ of $H[K]$ such that $R$ is attachable to $(H,F^K)$, we can modify $\Dt[R,H/K]$ to obtain $\Dt[\wh{F},H]$, where $\wh{F}$ is the extension of $F^K$ via $R$.

\begin{lemma}\label{lem:trim-implement}
 Suppose $G$ is a graph, $F$ is a recursively connected elimination forest of $G$ of height at most $d$, and $K$ is a prefix of $F$. Then given on input the data structure $\Dt[F,G]$, one can in time $\Oh(2^d\cdot d\cdot |K|)$ modify $\Dt[F,G]$ to obtain the partial data structure $\Dt[F-K,G/K]$. Moreover, in the obtained structure $\Dt[F-K,G/K]$, the list $\appendices$ has length at most $d\cdot(1+2^d\cdot |K|)$.
\end{lemma}
\begin{proof}
 We explain consecutive modifications that should be applied to $\Dt[F,G]$ in order to turn it into $\Dt[F-K,G/K]$.
 We summarize them using pseudocode as method $\mathtt{trim}(K)$, presented as Algorithm~\ref{alg:trim}.
 
 First, we remove all the vertices of $K$ from all the buckets. 
 This can be done in total time $\Oh(|K|)$ by removing each vertex $u\in K$ from the bucket that it belongs to; recall here that $u$ stores a pointer to the list element representing it in this bucket.
 Second, we need to rename each bucket $\bucket[u,X,i]$ for $u\in K\cup \{\bot\}$ to $\bucket[\bot,X,i]$, and put all these buckets on the list $\appendices$ while ignoring buckets that became empty. Also, we should simultaneously rename the cell $\parent[u,X,i]$ to $\parent[\bot,X,i]$, 
 and change its content to $\bot$.
 Observe that the total number of buckets $\bucket[u,X,i]$ for $u\in K\cup \{\bot\}$ in $\Dt[F,G]$ is bounded by $d\cdot (1+2^d\cdot |K|)$, hence the obtained list $\appendices$ has length at most $d\cdot (1+2^d\cdot |K|)$ and can be constructed in time $\Oh(2^d\cdot d\cdot |K|)$. Note also that renaming the buckets when constructing the list $\appendices$ involves suitable updates in the records stored for the vertices of $G$. 
 
 Here, let us point out one important detail: in the above operation, it will never be the case that two different buckets $\bucket[u_1,X,i]$ and $\bucket[u_2,X,i]$ will be renamed to the same new bucket $\bucket[\bot,X,i]$.
 This follows from Lemma~\ref{lem:parent-SReach} applied to $F$: if $\bucket[u_1,X,i]$ and $\bucket[u_2,X,i]$ were simultaneously nonempty for some subset of vertices $X$, then either $X=\emptyset$ and $u_1=u_2=\bot$, or both $u_1$ and $u_2$ would be equal to the deepest vertex of $X$ in~$F$, implying $u_1=u_2$.
 
 It is straightforward to see that these modifications turn $\Dt[F,G]$ into $\Dt[F-K,G/K]$.
\end{proof}

\begin{algorithm}\label{alg:trim}
    
    \SetKwInOut{Input}{Input}
    \SetKwInOut{Result}{Result}
    \SetKw{Break}{break}
    \SetKw{Raise}{raise exception:\ }

    \vskip 0.2cm
    
    \Input{A prefix $K$ of $F$}
    \Result{Structure $\Dt[F,G]$ is updated to $\Dt[F-K,G/K]$}
    
    \vskip 0.3cm
    
    \ForEach{$v\in K$} {      
        remove $v$ from the bucket that it belongs to
    }
    
    $\appendices\gets \mathtt{new\ List}()$
    
    \ForEach{$u\in K\cup \{\bot\}$ \normalfont{\textbf{and}} $X\subseteq \SReach_{F^K}(u) \cup \{ u \}$ \normalfont{\textbf{and}} $i\in \{1,\ldots,d\}$} {
        \If {$\bucket[u,X,i]$ \normalfont{is empty}} {\normalfont{\textbf{continue}}}
        rename $\bucket[u,X,i]$ to $\bucket[\bot,X,i]$\\
        rename $\parent[u,X,i]$ to $\parent[\bot,X,i]$\\
        $\parent[\bot,X,i]\gets \bot$\\
        $\appendices.\mathtt{append}(\bucket[u,X,i])$
    }
  
    \caption{method $\mathtt{trim}(K)$}
\end{algorithm}

\begin{lemma}\label{lem:extend-implement}
 Suppose $H$ is a graph, $K$ is a subset of vertices of $H$, and $F^K$ is an elimination forest of $H[K]$ of height at most~$d$. Suppose further that $R$ is a recursively connected elimination forest of $H-K$ of height at most $d$ that is attachable to $(H,F^K)$. Let $\wh{F}$ be the extension of $F^K$ via $R$. Then given on input $H[K]$, $F^K$, and the partial data structure $\Dt[R,H/K]$, one can in time $2^{\Oh(d)}\cdot |K|^{\Oh(1)}$ modify $\Dt[R,H/K]$ to obtain the data structure $\Dt[\wh{F},H]$.
\end{lemma}
\begin{proof}
 We explain consecutive modifications that should be applied to $\Dt[R,H/K]$ in order to turn it into $\Dt[\wh{F},H]$. It is straightforward to see that they can be executed in total time $2^{\Oh(d)}\cdot |K|^{\Oh(1)}$; we leave the easy verification to the reader.
 We summarize the procedure using pseudocode as method $\mathtt{extend}(H[K],F^K)$, presented as Algorithm~\ref{alg:extend}.
 
 Consider any $r\in \roots_R$. By assumption that $R$ is attachable to $(H,F^K)$, we have that $N_H(\desc_R(r))=\SReach_{R,H}(r)$ is straight in $F^K$. Observe that the total number of straight sets in $F^K$ is bounded by $1+|K|\cdot 2^{d-1}$: each straight set is either empty, or can be described by providing its deepest (in $F^K$) element and a subset of its ancestors. Hence, the list $\appendices$ in $\Dt[R,H/K]$ has length at most $(1+|K|\cdot 2^{d-1})\cdot d$.
 
 Next, we need to appropriately attach the buckets from the list $\appendices$ as in the extension operation.
 For this, for each bucket $\bucket[\bot,X,i]$ on this list we rename it to $\bucket[m,X,i]$, where $m$ is the vertex of $X$ that is the deepest in $F^K$; in case $X=\emptyset$, no renaming is needed. Similarly, we rename the memory cell $\parent[\bot,X,i]$ to $\parent[m,X,i]$ and change its content to $m$. 

 The key observation is that, at this point, all the data stored for every vertex $w\notin K$, that is, $\SReach(w)$, $\Up(w)$, $\height(w)$, the buckets $\bucket[w,\cdot,\cdot]$, the cells $\parent[w,\cdot,\cdot]$, and the value stored in $\tp(w)$, is exactly as it should be in the data structure~$\Dt[\wh{F},H]$.
 Therefore, it remains to finish computing the relevant records for all $u\in K$, which we do as follows.
 
 We process the vertices of $K$ in a bottom-up manner over the forest $F^K$.
 When processing $u\in K$, we assume that the buckets $\bucket[u,\cdot,\cdot]$ have already been correctly computed (note that this holds at the beginning for the leaves of $F^K$). From these, we may compute $\Up(u)$, $\SReach(u)$, and $\height(u)$ using the following straightforward formulas:
 \begin{eqnarray*}
  \Up(u) & = & \anc_{F^K}(u)\cap N_{H[K]}(u);\\
  \SReach(u) & = & \Up(u)\cup \bigcup \left\{\,X\subseteq \anc_{F^K}(u)\setminus \{u\}\ \mid\ \bucket[u,X,i]\neq \emptyset \textrm{ for some } i\in \{1,\ldots,d\}\,\right\};\\
  \height(u) & = & 1 + \max \{\,i\ \mid\ \bucket[u,X,i]\neq \emptyset \textrm{ for some } X\subseteq \anc_{F^K}(u)\setminus \{u\}\,\};
 \end{eqnarray*}
 where in the last formula we assume that $\max$ evaluates to $0$ if it ranges over the empty set.
 
 Next, we add $u$ to the bucket $\bucket[\pnt_{F^K}(u),\SReach(u),\height(u)]$. Finally, we make $\tp(u)$ point to the cell $\bucket[\pnt_{F^K}(u),\SReach(u),\height(u)]$, which in turn should point to $\pnt_{F^K}(u)$.
 This completes the necessary modifications.
\end{proof}

\begin{algorithm}\label{alg:extend}
    
    \SetKwInOut{Input}{Input}
    \SetKwInOut{Result}{Result}
    \SetKw{Break}{break}
    \SetKw{Raise}{raise exception:\ }

    \vskip 0.2cm
    
    \Input{Graph $H[K]$ and an elimination forest $F^K$ of $H[K]$ of height at most $d$ such that $R$ is attachable to $(H,F^K)$}
    \Result{Structure $\Dt[R,H/K]$ is updated to $\Dt[\wh{F},H]$, where $\wh{F}$ is the extension of $F^K$ via $R$}
    
    \vskip 0.3cm
    
    \ForEach{\normalfont{bucket} $\bucket[\bot,X,i]\in \appendices$}{
        $m\gets$ deepest vertex of $X$ in $F^K$, or $\bot$ if $X=\emptyset$\\
        rename $\bucket[\bot,X,i]$ to $\bucket[m,X,i]$\\
        rename $\parent[\bot,X,i]$ to $\parent[m,X,i]$\\
        $\parent[m,X,i]\gets m$
    }
    
    free$(\appendices)$\\
   
    \ForEach{ $u\in K$ \normalfont{ in the bottom-up traversal of } $F^K$} {
       $\Up(u)\gets \anc_{F^K}(u)\cap N_{H[K]}(u)$\\
       $\SReach(u)\gets \Up(u)$\\
       $\height(u)\gets 1$\\
       \For{$i\gets 1$ \normalfont{\textbf{to}} $d$}{
         \ForEach{$X\subseteq \anc_{F^K}(u)\setminus \{u\}$}{
            \If{$\bucket[u,X,i]$\normalfont{ is not empty}}{
              $\SReach(u)\gets \SReach(u)\cup X$\\
              $\height(u)\gets 1 + i$
            }
         }
       }
       $\tp(u)\gets$ pointer to $\parent[\pnt_{F^K}(u),\SReach(u),\height(u)]$\\
       $\bucket[\pnt_{F^K}(u),\SReach(u),\height(u)].\mathtt{append}(u)$\\
    }
    
    \caption{method $\mathtt{extend}(H[K],F^K)$}
\end{algorithm}

\paragraph*{Updates.} We are ready now show that the data structure can be maintained under edge insertions and removals. The idea is to first use Lemma~\ref{lem:extract-core} to extract a small core $K$ 
that contains both endpoints of the updated edge,
then run the static algorithm of Lemma~\ref{lem:static} to compute a recursively optimal elimination forest $F^K$ of the updated $G[K]$, and finally apply the extension operation, implemented using the procedures provided by Lemmas~\ref{lem:trim-implement} and~\ref{lem:extend-implement}. The correctness of this approach is asserted by Lemma~\ref{lem:core-uberlemma}.

\begin{lemma}\label{lem:update}
 Suppose that we have access to a data structure $\Dt[F,G]$ that stores a recursively optimal elimination forest $F$ of a graph $G$ such that $F$ has height at most $d$.
 Let $H$ be a graph obtained from $G$ by inserting or removing a single edge, given as input. 
 Then one can in $2^{\Oh(d^2)}$ time either conclude that $\td(H)>d$, or modify 
$\Dt[F,G]$ to obtain a data structure $\Dt[\wh{F},H]$ where $\wh{F}$ is some recursively 
optimal elimination forest of $H$ of height at most $d$.
\end{lemma}
\begin{proof}
  We only present the implementation of edge insertion, and at the end we discuss the tiny modifications needed for the implementation of edge removal.
  The corresponding method $\ins(uv)$, where $uv$ is the new edge, is presented using pseudocode as Algorithm~\ref{alg:insert}.
  We now explain the consecutive steps.
  Let $H\coloneqq G+uv$, that is, $H$ is obtained from $G$ by adding the edge~$uv$.
 
 First, we apply the method $\mathtt{core}(d+1,\{u,v\})$ provided by Lemma~\ref{lem:extract-core} to construct a $(d+1)$-core $K$ of $(G,F)$ that contains both endpoints of $uv$.
 Also, we construct the graph $G[K]$. Note that $|K|\leq d^{\Oh(d)}$ and this 
step takes $d^{\Oh(d)}$ time.
 
 We add the edge $uv$ to $G[K]$, thus obtaining the graph $H[K]$.
 Note that, since $\td(G)\leq d$, we have $\td(H)\leq d+1$, hence also $\td(H[K])\leq d+1$.
 We apply the algorithm of Lemma~\ref{lem:static} to $H[K]$ to compute a recursively optimal elimination forest $F^K$ of $H[K]$.
 Note that this step takes $2^{\Oh(d^2)}\cdot |K|^{\Oh(1)}=2^{\Oh(d^2)}$ time.
 If $\height(F^K)>d$ then $\td(H)\geq \td(H[K])=\height(F^K)>d$, hence we may terminate the algorithm and conclude that $\td(H)>d$.
 So from now on we assume that $\height(F^K)\leq d$.
 
 As $H$ is a $(K,0)$-restricted augmentation of $G$, from Lemma~\ref{lem:core-uberlemma} we infer that $F-K$ is attachable to $(H,F^K)$ and if $\wh{F}$ is the extension of $F^K$ via $F-K$, then $\wh{F}$ is a recursively optimal elimination forest of $H$ of height at most $d$. Therefore, it remains to update the data structure $\Dt[F,G]$ to the structure $\Dt[\wh{F},H]$. This can be done in time $2^{\Oh(d)}\cdot |K|^{\Oh(1)}=d^{\Oh(d)}$ by first using Lemma~\ref{lem:trim-implement} to modify $\Dt[F,G]$ to $\Dt[F-K,G/K]$, which is equal to $\Dt[F-K,H/K]$, and then using Lemma~\ref{lem:extend-implement} to modify $\Dt[F-K,H/K]$ to $\Dt[\wh{F},H]$.
 
 This concludes the implementation of an edge insertion.
 For the implementation of an edge removal, we follow exactly the same strategy with the exception that we compute a $(d+2)$-core $K$ instead of a $(d+1)$ core, 
 because $G-uv$ is then a $(K,1)$-restricted augmentation of $G$ instead of a $(K,0)$-restricted augmentation. Also, there is no need of checking whether the height of $F^K$ is larger than $d$, 
 because it is bounded by $\td((G-uv)[K])\leq \td(G-uv)\leq \td(G)\leq d$.
\end{proof}

Lemma~\ref{lem:update} concludes the proof of Theorem~\ref{thm:td-data}.

\begin{algorithm}\label{alg:insert}
    
    \SetKwInOut{Input}{Input}
    \SetKwInOut{Result}{Result}
    \SetKw{Break}{break}
    \SetKw{Raise}{raise exception:\ }

    \vskip 0.2cm
    
    \Input{A new edge $uv\notin E(G)$}
    \Result{Structure $\Dt[F,G]$ is updated to $\Dt[\wh{F},G+uv]$, where $\wh{F}$ is a recursively optimal elimination forest of $G+uv$ of depth at most $d$}
    
    \vskip 0.3cm
    
    $K\gets \mathtt{core}(\{u,v\},d+1)$\\
    construct $G[K]+uv$\\
    $F^K \gets$ recursively optimal elimination forest of $G[K]+uv$ obtained using Lemma~\ref{lem:static}\\
    \If{$\height(F^K)>d$} {
      \Raise $\td(G+uv)>d$
    }

    $\mathtt{trim}(K)$\\
    $\mathtt{extend}(G[K]+uv,F^K)$
    
    \caption{method $\mathtt{insert}(uv)$}
\end{algorithm}

%%% Local Variables:
%%% mode: latex
%%% TeX-master: "main"
%%% End:

%\input{pseudocodes}

\section{Dynamic dynamic programming}\label{sec:dp}

In this section we show how the data structure presented in Section~\ref{sec:data-structure} can be enriched so that together with an elimination forest, it also maintains a run of an auxiliary dynamic program on this forest.
For this, we need to understand dynamic programming on elimination forests in an abstract way, which we do using a formalism of {\em{configuration schemes}}.
This formalism is based on the classic algebraic approach to graph decompositions, and in particular on the idea of recognizing finite-state properties of graphs of bounded treewidth through homomorphisms into finite algebras.
This direction has been intensively developed in the 90s, see the book of Courcelle and Engelfriet~\cite{CourcelleE12} for a broad introduction.
The definitional layer that we use is closest to the one used in the work of Bodlaender et al.~\cite{BodlaenderFLPST16}.

\paragraph*{Boundaried graphs.} 
We shall assume that all vertices in all the considered graphs come from some countable reservoir of vertices $\Omega$, say $\Omega=\N$.
A {\em{boundaried graph}} is a graph $G$ together with a subset of vertices $X\subseteq V(G)$, called the {\em{boundary}}.
On boundaried graphs we can introduce two basic operators:
\begin{itemize}[nosep]
 \item For a boundaried graph $(G,X)$ and $x\in X$, the operator $\forget$ yields 
       $$\forget((G,X),x)\coloneqq (G,X\setminus \{x\}).$$
 \item For two boundaried graphs $(G_1,X_1)$ and $(G_2,X_2)$ such that $V(G_1)\cap V(G_2)=X_1\cap X_2$, the {\em{union}} operator $\oplus$ yields
       $$(G_1,X_1)\oplus (G_2,X_2)\coloneqq ((V(G_1)\cup V(G_2),E(G_1)\cup E(G_2)),X_1\cup X_2).$$
       In other words, we take the disjoint union of $G_1$ and $G_2$ and, for each $x\in X_1\cap X_2$, we fuse the copy of $x$
       from $G_1$ with the copy of $x$ from $G_2$. The boundary of the output graph is $X_1\cup X_2$. 
\end{itemize}
Note that the operator $\forget((G,X),x)$ is defined only if $x\in X$, while the operator $(G_1,X_1)\oplus (G_2,X_2)$ is defined only if $V(G_1)\cap V(G_2)=X_1\cap X_2$. 
It is useful to think of this as of an abstract algebra on boundaried graphs, endowed with operations $\forget$ and $\oplus$.

\paragraph*{Configuration schemes.} 
We now introduce the language of {\em{configuration schemes}}, which intuitively is a way of assigning each boundaried graph $(G,X)$ a set of {\em{configurations}} $\conf(G,X)$
which contains all the essential information about the behaviour of $(G,X)$ with respect to some computational problem.

A {\em{configuration scheme}} is a pair of mappings $(\Conf,\conf)$ such that
with any set of vertices $X\subseteq\Omega$ we may associate a set of {\em{configurations}} $\Conf(X)$,
and with every boundaried graph $(G,X)$ we may associate a subset of configurations $\conf(G,X)\subseteq \Conf(X)$ {\em{realized}} by $(G,X)$. For an example what this might be, see the first paragraph of the proof of Lemma~\ref{lem:schem-kpath}.
We will need several properties of configuration schemes, the first of which is composability.

We will say that such a configuration scheme $(\Conf,\conf)$ is {\em{composable}} if the following conditions hold.
\begin{itemize}[nosep]
 \item For a boundaried graph $(G,X)$ and $x\in X$, $\conf(\forget((G,X),x))$ is uniquely determined by $\conf(G,X)$ and $x$.
 \item For boundaried graphs $(G_1,X_1)$ and $(G_2,X_2)$ whose union is defined, $\conf((G_1,X_1)\oplus (G_2,X_2))$ is uniquely determined by $\conf(G_1,X_1)$ and $\conf(G_2,X_2)$.
\end{itemize}
Letting $\Xi\coloneqq \biguplus_{X\subseteq \Omega} 2^{\Conf(X)}$, the above conditions are equivalent to saying that there exist operators 
$$\forget\colon \Xi \times \Omega \to \Xi \qquad\textrm{and}\qquad\oplus\colon \Xi \times \Xi \to \Xi$$
such that the following assertions hold:
\begin{itemize}[nosep]
 \item Operator $\forget(\cdot,\cdot)$ is defined only on pairs of the form $(C,x)$ such that $C\in 2^{\Conf(X)}$ and $x\in X$ for some $X\subseteq \Omega$; then $\forget(C,x) \in 2^{\Conf(X \setminus \{ x \} )}$.
 \item For every boundaried graph $(G,X)$ and $x\in X$, we have $$\conf(\forget((G,X),x))=\forget(\conf(G,X),x).$$
 \item If $C_1 \in 2^{\Conf(X_1)}$ and $C_2 \in 2^{\Conf(X_2)}$ for some $X_1,X_2\subseteq \Omega$, then $C_1 \oplus C_2 \in 2^{\Conf(X_1 \cup X_2)}$. 
 \item For all boundaried graphs $(G_1,X_1)$ and $(G_2,X_2)$ whose union is defined, we have $$\conf((G_1,X_1)\oplus (G_2,X_2))=\conf(G_1,X_1)\oplus \conf(G_2,X_2).$$
\end{itemize}
In other words, $\conf$ is a homomorphism from the algebra of boundaried graphs endowed with operations $\forget$ and $\oplus$ to the algebra on $\Xi$ with the same operations.
Note that $\oplus$ is commutative and associative on boundaried graphs, hence it should also have these properties on sets of configurations. 

We will say that $(\Conf,\conf)$ is {\em{efficiently composable}} if there is a non-decreasing computable function $\zeta\colon \N\to \N$ such that
\begin{itemize}[nosep]
 \item For every $X\subseteq \Omega$ we have $|\Conf(X)|\leq \zeta(|X|)$ and given $X$, one can compute $\Conf(X)$ in time $\zeta(|X|)^{\Oh(1)}$.
 \item Given a boundaried graph $(G,X)$ with $|V(G)| \leq 2$ and $X = V(G)$, one can compute $\conf(G,X)$ in constant time.
 \item Given $X\subseteq \Omega$, $x\in X$, and $C\subseteq \Conf(X)$, one can compute $\forget(C,x)$ in time~$\zeta(|X|)^{\Oh(1)}$.
 \item Given $X_1,X_2\subseteq \Omega$, $C_1\subseteq \Conf(X_1)$, and $C_2\subseteq \Conf(X_2)$, one can compute  $C_1\oplus C_2$ in time $\zeta(|X|)^{\Oh(1)}$.
\end{itemize}
Finally, we will say that $(\Conf,\conf)$ is {\em{idempotent}} if there is a non-decreasing computable function $\tau\colon \N\to \N$ such that the following condition holds:
\begin{itemize}[nosep]
 \item Consider any $X\subseteq \Omega$ and a multiset $\Mm$ whose elements are subsets of $\Conf(X)$. 
       Suppose in $\Mm$ there is $C\subseteq \Conf(X)$ such that for each $c\in C$, there are at least $\tau(|X|)$ elements $D\in \Mm\setminus \{C\}$ such that $c\in D$.
       Then $$\bigoplus_{D\in \Mm\setminus \{C\}} D\ =\ \bigoplus_{D\in \Mm} D.$$
\end{itemize}
A configuration scheme $(\Conf,\conf)$ that is both efficiently composable and idempotent shall be called {\em{tractable}}. 
The corresponding functions $\zeta$ and $\tau$ shall be called the {\em{witnesses}} of tractability of $(\Conf,\conf)$.

\paragraph*{Recognizing properties.}
A {\em{graph property}} is simply a set $\Pi$ consisting of graphs, interpreted as graphs that admit $\Pi$.
A graph property $\Pi$ is {\em{recognized}} by a configuration scheme $(\Conf,\conf)$ if there exists a subset of {\em{final}} configurations $F\subseteq \Conf(\emptyset)$ such that for every graph $G$, we have 
$$G\in \Pi\qquad \textrm{if and only if}\qquad \conf(G,\emptyset)\cap F\neq \emptyset.$$

As mentioned before, the formalism introduced above should be standard for a reader familiar with the work on graph algebras and recognition of $\MSO$-definable languages of graphs.
Recall here that $\MSO_2$ is a logic on graphs that extends the standard first order logic $\FO$ by allowing quantification over subsets of vertices and over subsets of edges.
A graph property $\Pi$ is {\em{$\MSO_2$-definable}} if there exists a sentence $\varphi$ of $\MSO_2$ such that for every graph $G$, $\varphi$ holds in $G$ if and only if $G\in \Pi$.
Then the classic connection between graph algebras and $\MSO_2$ yields the following statement.
As this is not the main focus of our work, we only give a sketch of the proof, but both the statement and the tools are standard.

\begin{lemma}\label{lem:mso-tractable}
 Every $\MSO_2$-definable graph property is tractable.
\end{lemma}
\begin{proof}[Sketch]
 Let us first recall some basic definitions and facts from logic.
 The {\em{quantifier rank}} of a formula is the maximum number of nested quantifiers in it.
 For $X\subseteq \Omega$, let $\MSO_2[X]$ be the standard $\MSO_2$ logic on graphs where in the signature, apart from the standard relations for encoding graphs, we also have each $x\in X$ introduced as a constant.
 It is known that for every $q\in \N$ and $X$, there is only a finite number of pairwise non-equivalent $\MSO_2[X]$-statements of quantifier rank at most $q$.
 Moreover, given $q$ and $X$, one can compute a set $\Types^q(X)$ consisting of one $\MSO_2[X]$-statement of quantifier rank at most $q$ per each class of equivalence. Thus, $|\Types^q(X)|$ is bounded by a computable function of $q$ and $X$.
 Then, with every boundaried graph $(G,X)$ we can associate its {\em{$q$-type}} $\type^q(G,X)\subseteq \Types^q(X)$, which just comprises all those statements $\psi\in \Types^q(X)$ that are satisfied in $(G,X)$.
 
 We move to the sketch of the proof.
 Let $\Pi$ be the property in question, and let $\varphi$ be the $\MSO_2$-statement defining $\Pi$. Let $q$ be the quantifier rank of $\varphi$.
 We define the following configuration scheme $(\Conf,\conf)$:
 \begin{itemize}[nosep]
  \item for $X\subseteq \Omega$, we let $\Conf(X)\coloneqq \Types^q(X)$; and
  \item for a boundaried graph $(G,X)$, we let $\conf(G,X)\coloneqq \type^q(G,X)$.
 \end{itemize}
 The fact that this configuration scheme is efficiently composable and idempotent is a standard fact about $\MSO_2$ logic, which can be proved e.g. using Ehrenfeucht-Fra\"isse games.
 By taking the set of final configurations to be $F\coloneqq \{\varphi\}$, we see that this configuration scheme recognizes $\Pi$.
\end{proof}

Lemma~\ref{lem:mso-tractable} gives tractability of any $\MSO_2$-definable graph property, however the witnesses $\zeta,\tau$ of this tractability are non-elementary functions.
We now give an explicit configuration scheme for the property of our interest: containing a $k$-path.

\begin{lemma}\label{lem:schem-kpath}
 For $k\in \N$, let $\Pi_k$ be the graph property comprising all graphs that contain a simple path on $k$ vertices.
 Then $\Pi_k$ is recognized by a tractable configuration scheme with witnesses $\zeta(x)=(k+1)\cdot 2^{x+1}\cdot x!$ and $\tau(x)=x+2$.
\end{lemma}
\begin{proof}
 We first define the configuration scheme.
 Let $s,t$ be two new special vertices, not belonging to $\Omega$; intuitively, in our scheme these will be markers representing the beginning and the end of the constructed path.
 Recall that a {\em{linear forest}} is an acyclic undirected graph whose every component is a path.
 For $X\subseteq \Omega$, we define the set of configurations $\Conf(X)$ as follows: $\Conf(X)$ consists of all pairs $(H,i)$ such that:
 \begin{itemize}[nosep]
  \item $H$ is a linear forest with vertex set $X\cup \{s,t\}$, where the degrees of $s$ and $t$ are at most $1$; and
  \item $i\in \{0,1,\ldots,k-1,\infty\}$.
 \end{itemize}
 Moreover, we require that if $E(H)=\emptyset$, then $i=0$.
 
It is straightforward to see that the number of different linear forests with vertex set of size $p$ is at most $2^{p-1}\cdot p!$: a linear forest can be chosen by 
selecting a permutation $\pi=(u_1,\ldots,u_p)$ of the vertex set ($p!$~choices), and then deciding, for each $i\in \{1,\ldots,p-1\}$, whether $u_i$ and $u_{i+1}$ are adjacent ($2^{p-1}$ choices).
Therefore, we have $|\Conf(X)|\leq (k+1)\cdot 2^{|X|+1}\cdot |X|!=\zeta(|X|)$, as required.

Now, for a boundaried graph $(G,X)$ and a configuration $(H,i)\in \Conf(X)$, we shall say that $(H,i)$ is {\em{realized}} in $(G,X)$ if there exists a family of paths $\{P_e\colon e\in E(H)\}$ in $G$
satisfying the following conditions:
\begin{itemize}[nosep]
 \item For each $xy\in E(H)$, $P_e$ is a path whose one endpoint is $x$, provided $x\in X$, and the second endpoint is $y$, provided $y\in X$. If $x$ or $y$ (or both) belongs to $\{s,t\}$, then the corresponding endpoint of $P_e$ can be any vertex of $V(G)$. Moreover, $V(P_e)\cap X=\{x,y\}\cap X$.
 \item For all $e,e'\in E(H)$, the paths $P_e$ and $P_{e'}$ are vertex-disjoint, apart from possibly sharing an endpoint in case $e$ and $e'$ share an endpoint.
 \item The total number of edges on paths $P_e$ is equal to $i$ if $i<\infty$, or is at least $k$ if $i=\infty$.
\end{itemize}
For a boundaried graph $(G,X)$, let $\conf(G,X)\subseteq \Conf(X)$ be the set of configurations realized in $(G,X)$.

Observe that a graph $G$ contains a simple path on $k$ vertices if and only if $\conf(G,\emptyset)$ contains configuration $(H_0,k-1)$, where $H_0$ is the graph on vertex set $\{s,t\}$ with the edge $st$ present.
Thus, by setting $F\coloneqq \{(H_0,k-1)\}$ we see that the scheme $(\Conf,\conf)$ recognizes $\Pi_k$. It remains to show that $(\Conf,\conf)$ is suitably efficiently composable and idempotent.

\begin{claim}\label{cl:composable}
 The configuration scheme $(\Conf,\conf)$ is efficiently composable with witness $\zeta(x)=(k+1)\cdot 2^{x+1}\cdot x!$.
\end{claim}
\begin{clproof}
 The bound $|\Conf(X)|\leq\zeta(|X|)$ has already been shown. 
 Also, for every boundaried graph $(G, X)$ such that $|V(G)|\leq 2$ and $X = V(G)$, the value of $\conf(G, X)$ can be computed in constant time trivially by the definition.
 Thus, to prove the claim, it remains to define suitable operators $\forget$ and $\oplus$, and to show that they are computable in time $\zeta(|X|)^{\Oh(1)}$.
 
 We start with defining the $\forget(\cdot,\cdot)$ operator.
 For any boundaried graph $(G, X)$ and $x \in X$ consider a configuration $(H, i) \in \conf(\forget((G,X),x))$ and a corresponding family of paths $\{P_e\colon e\in E(H)\}$ in $G$, as described in the definition of $\conf(\cdot)$. 
 There are three possibilities of how $x$ can interact with this family of paths:
 \begin{itemize}[nosep]
  \item $x$ is an internal vertex of exactly one path $P_e$. This means that $(H', i) \in \conf(G, X)$, where $e = yz$ and $H'=(V(H) \cup \{x\}, E(H) \cup \{xy, xz\} \setminus \{yz\})$.
  \item $x$ is an endpoint of some path $P_e$. From the definition this is possible if and only if the corresponding endpoint of $e$ belongs to $\{s, t\}$; say $e=ys$, the other case being analogous. 
        Again, this means that $(H', i) \in \conf(G, X)$, where $H'=(V(H) \cup \{x\}, E(H)\cup \{xy, xs\} \setminus \{ys\})$.
  \item $x$ does not belong to any of these paths. This means that $((V(H) \cup \{x\}, E(H)), i) \in \conf(G, X)$. 
 \end{itemize}
 
 These observations show a way of defining $\forget(\cdot,\cdot)$. For a linear forest $H$ and $x \in V(H)$ such that $x$ has degree $2$ in $H$, 
 let $\join(H, x)$ be the graph obtained from $H$ by adding an edge connecting the neighbors of $x$ and removing $x$ itself. Then for $C\subseteq \Conf(X)$ we can write the definition of $\forget$ as follows:
 \begin{eqnarray*}
   \forget(C, x) & \coloneqq & \forget_0(C, x) \cup \forget_2(C, x);\\
   \forget_0(C, x) & \coloneqq & \{\, (H-x, i)\ \colon\ (H, i) \in C, x \in V(H), \degree_H(x) = 0\,\};\\
   \forget_2(C, x) & \coloneqq & \{\, (\join(H, x), i)\ \colon\ (H, i) \in C, x \in V(H), \degree_H(x) = 2\,\}.
 \end{eqnarray*}
 Based on the previous observations, it is easy to see that this operator satisfies the desired property: $\conf(\forget((G,X),x))=\forget(\conf(G,X),x)$ for every boundaried graph $(G, X)$ and $x \in X$. 
 Moreover, $\forget(C, x)$ can be computed in time $\zeta(|X|)^{\Oh(1)}$ by inspecting every element of $C$ and applying the formula above; recall here that $|C|\leq \zeta(|X|)$.
 
 We now move to the discussion of the $\oplus$ operator.
 Consider any pair of boundaried graphs $(G_1, X_1)$ and $(G_2, X_2)$ whose union is defined, and let $(G, X) \coloneqq (G_1,X_1)\oplus (G_2,X_2)$. 
 Consider a configuration $(H, i) \in \conf(G, X)$ and a corresponding family of paths $\{P_e\colon e\in E(H)\}$ in $G$. 
 Note that each path $P_e$ is either a single edge, in which case it can be present both in $G_1$ and $G_2$, or it has some internal vertices. 
 In the latter case all internal vertices and edges of $P_e$ must belong either to $G_1$ or to $G_2$. They cannot belong to both, as $P_e$ is internally disjoint with $X$, and in $G$ there is no edge between $V(G_1)\setminus X_1$
 and $V(G_2)\setminus X_2$.
 In particular, this means that each path from the family $\{P_e\colon e\in E(H)\}$ is either entirely contained in $G_1$ or entirely contained in $G_2$. Based on this observation, we can define $\oplus$ as follows.
 
 Two configurations $(H_1,i_1)\in \Conf(X_1)$ and $(H_2,i_2)\in \Conf(X_2)$ are {\em{mergeable}} if $E(H_1)\cap E(H_2)=\emptyset$ and 
 the graph $H_1\oplus H_2\coloneqq (V(H_1)\cup V(H_2),E(H_1)\cup E(H_2))$ is a linear forest with $s$ and $t$ having degrees at most $1$. 
 Then for $C_1\subseteq \Conf(X_1)$ and $C_2\subseteq \Conf(X_2)$, we define operator $\oplus$ as follows
 $$C_1\oplus C_2\coloneqq C_1 \cup C_2 \cup \{\,(H_1\oplus H_2,i_1+i_2)\ \colon\ (H_1,i_1)\in C_1, (H_2,i_2)\in C_2,\textrm{ and they are mergeable}\},$$
 where the value $i_1+i_2$ is replaced with $\infty$ in case it exceeds $k-1$.
 
 Based on the previous observation, it is easy to see that this operator satisfies the desired property:
 $\conf((G_1,X_1)\oplus (G_2,X_2))=\conf(G_1,X_1)\oplus \conf(G_2,X_2)$ for all boundaried graphs $(G_1,X_1)$ and $(G_2,X_2)$ 
 whose union is defined. That the operator $\oplus$ defined above is commutative and associative is obvious.
 Finally, the value of $C_1\oplus C_2$ can be computed in time $\zeta(|X|)^{\Oh(1)}$ by inspecting all pairs $((H_1,i_1),(H_2,i_2))\in C_1 \times C_2$, 
 whose number is bounded by $\zeta(|X_1|)\zeta(|X_2|) \leq \zeta(|X|)^2$, and applying the formula above.
\end{clproof}

\begin{claim}\label{cl:idempotent}
 The configuration scheme $(\Conf,\conf)$ is idempotent with witness $\tau(x)=x+2$.
\end{claim}
\begin{clproof}
 Consider any $X\subseteq \Omega$ and a multiset $\Mm$ whose elements are subsets of $\Conf(X)$. Suppose that there is $C \in \Mm$ such that for each $c\in C$, 
 there are at least $|X|+2$ elements $D\in \Mm\setminus\{C\}$ such that $c\in D$. Let $S\coloneqq \bigoplus (\Mm\setminus \{C\})$. 
 We need to prove that $S = S \oplus C$. 
 
 That $S \subseteq S \oplus C$ is directly implied by the definition of operator $\oplus$. To prove that $S \oplus C \subseteq S$, let us consider any configuration $(H, i) \in S \oplus C$; we need to argue that in fact $(H,i)\in S$.
 
 From the definition of $\oplus$ there exists a multiset of configurations $M = \{(H_1, i_1), (H_2, i_2), \ldots, (H_m, i_m)\}$ for some  $m \leq |\Mm|$, 
 with each configuration in $M$ chosen from some distinct element of multiset $\Mm$, such that $H = \bigoplus_{j=1}^m H_j$ and either $i = \sum_{j=1}^m i_j$ in case $i<\infty$, or $\sum_{j=1}^m i_j\geq k$ in case $i=\infty$. 
 In particular, graphs $H_j$ for $j\in \{1,\ldots,m\}$ are pairwise mergeable.
 Note that $H=\bigoplus_{j=1}^m H_j$ is a linear forest on a vertex set of size $|X|+2$, hence it has at most $|X|+1$ edges. As graphs $H_j$ have pairwise disjoint edge sets due to being mergeable, we conclude that all but at most $|X|+1$ graphs $H_j$ have empty edge sets. Since in the definition of a configuration we required that emptiness of the edge set entails that the second coordinate is $0$, for those graphs $H_j$ for which $E(H_j)=\emptyset$ we also have $i_j=0$.
 Now observe that configurations $(H_j,i_j)$ with $E(H_j)=\emptyset$ and $i_j=0$ can be safely removed from $M$ without changing the union of configurations in $M$. Thus, by restricting $M$ if necessary, we may assume that $|M|\leq |X|+1$.
 
 Now, if none of the configurations from $M$ were selected from $C$, then $(H, i) \in S$ and we are done. 
 On the other hand, if some configuration $(H_j, i_j)$ was selected from $C$, then by the assumptions on $C$ and the fact that $|M|\leq |X|+1$, there is at least one other element $D \in \Mm\setminus \{C\}$, 
 such that no other configuration in $M$ was selected from $D$ and $(H_j, i_j) \in D$. Therefore, we can select $(H_j, i_j)$ from $D$ instead of $C$, which again implies that $(H, i) \in S$.  
\end{clproof}

Claims~\ref{cl:composable} and~\ref{cl:idempotent} verify that the scheme $(\Conf,\conf)$ has all the required properties. This concludes the proof of Lemma~\ref{lem:schem-kpath}.
\end{proof}

\newcommand{\Gb}{\mathbf{G}}
\newcommand{\Hb}{\mathbf{H}}

\paragraph*{Maintaining a scheme.}
We next show that whenever a graph property admits a tractable configuration scheme, it can be efficiently maintained in a dynamic graph of bounded treedepth. Let $\Pi$ be a graph property that is recognized by a tractable configuration scheme $(\Conf,\conf)$ with polynomial-time computable witnesses~$\zeta(\cdot)$ and $\tau(\cdot)$. Let $G$ be a graph and $F$ be a recursively connected elimination forest of $G$. We enrich the data structure $\Dt[F,G]$ presented in Section~\ref{sec:data-structure} to the data structure $\Dt^\Pi[F,G]$ as follows.
 
 For each $w\in V(G)$, define the graph 
 $$G_w\coloneqq \left(\,\desc_F(w)\cup \SReach_{F,G}(w),\ \{\,e\in E(G)\ \colon\ e\cap \desc_F(w)\neq \emptyset\,\}\,\right).$$
 In other words, the vertex set of $G_w$ consists of the descendants of $w$ plus $\SReach_{F,G}(w)$, while in the edge set we include only those edges of $G$ that are incident on vertices of $\desc_F(w)$.
 Thus, $\SReach_{F,G}(w)$ is an independent set in $G_w$. Further, we define the boundaried graph
 $$\Gb_w\coloneqq \left(\,G_w,\ \SReach_{F,G}(w)\,\right).$$
 
 In addition to all the data stored in $\Dt[F,G]$, in $\Dt^\Pi[F,G]$ we include the following. For each vertex~$u$, subset $X\subseteq \SReach_{F,G}(u)\cup \{u\}$, $i\in \{1,\ldots,d\}$, and configuration $c\in \Conf(X)$, 
 we store the {\em{mug}} $\bucket[u,X,i,c]$ defined as follows:
 $$\bucket[u,X,i,c]\coloneqq \{\,w\in \bucket[u,X,i]\ \colon\ c\in \conf(\Gb_w)\,\}.$$
 In other words, the mug $\bucket[u,X,i,c]$ comprises all vertices $w$ from the bucket $\bucket[u,X,i]$ for which $c$ is realized in $\Gb_w$.
 Note that the mugs are not necessarily disjoint, in particular one vertex may be stored in up to $|\Conf(X)|\leq \zeta(d)$ mugs. Some vertices may even not belong to any mug. Similarly as for $\Dt[F,G]$ structure, we only store nonempty mugs, so the storage space used for the mugs is bounded by $\Oh(n \cdot \zeta(d))$.

 Similarly as for the buckets, the mugs are represented as doubly linked lists. In the lists corresponding to mugs we only store indices of the required vertices, representing their copies, while the actual objects corresponding to vertices are stored in the buckets as before. Each vertex object $w$ in bucket $\bucket[u,X,i]$, besides storing pointers to the next and the previous element in $\bucket[u,X,i]$, also stores $ \zeta(d)$ additional pointers: for each $c\in \Conf(X)$, we store the pointer to the position of an element in $\bucket[u,X,i,c]$ containing index of $w$, or null pointer if the $w$ does not belong to $\bucket[u,X,i,c]$.
 Thus, whenever we remove a vertex $w$ from its bucket, we may also remove it from all the mugs to which it belongs in time $\Oh(\zeta)$. These pointers occupy space $\Oh(n \cdot \zeta(d))$. So the total space for $\Dt^\Pi[F,G]$ is $\Oh(n \cdot (d+ \zeta(d)))$.
 
 \newcommand{\mmb}{\mathsf{member}}
 
 Finally, $\Dt^\Pi[F,G]$ also stores a boolean flag $\mmb$, which is set to true if and only if $G\in \Pi$. Note that this flag can be computed by checking whether $\conf(G,\emptyset)=\oplus_{r\in \roots_F} \conf(\Gb_r)$ contains a final configuration.
 It is again clear that $\Dt^{\Pi}[F,G]$ can be initialized for an edgeless graph $G$ in $\Oh(n)$ time, because there are additionally $\Oh(n)$ empty mugs to initialize.

 We now explain how the mugs are maintained upon updates of $\Dt^\Pi[F,G]$. For this, it will be again convenient to speak about partial data structures $\Dt^{\Pi}[R,G/K]$, where $K$ is a subset of vertices of $G$ and $R$ is an elimination forest of $G-K$. We define them in the same way as in Section~\ref{sec:data-structure}, with the exception that each bucket $\bucket[u,X,i]$ is enriched with mugs $\bucket[u,X,i,c]$ for $c\in \Conf(X)$ as in the definition of $\Dt^{\Pi}[\cdot,\cdot]$. We now argue that suitable analogs of Lemmas~\ref{lem:trim-implement},~\ref{lem:extend-implement}, and~\ref{lem:update-conf} hold. Note that in the following statements, the configuration scheme $(\Conf,\conf)$ is considered fixed.
 
 \begin{lemma}\label{lem:trim-implement-scheme}
 Suppose $G$ is a graph, $F$ is a recursively connected elimination forest of $G$ of height at most $d$, and $K$ is a prefix of $F$. Then given on input the data structure $\Dt^\Pi[F,G]$, one can in time $\Oh(2^d\cdot d\cdot |K|\cdot \zeta(d))$ modify $\Dt^\Pi[F,G]$ to obtain the partial data structure $\Dt^\Pi[F-K,G/K]$. Moreover, in the obtained structure $\Dt^\Pi[F-K,G/K]$, the list $\appendices$ has length at most $1+2^d\cdot |K|$.
\end{lemma}
\begin{proof}
 Perform exactly the same modifications as described in the proof of Lemma~\ref{lem:trim-implement}, and additionally do the following:
 \begin{itemize}[nosep]
  \item For each $u\in K$, remove $u$ also from all the mugs in which it is contained. Note that this boils down to removing the single list element corresponding to $u$ not just from one list representing the bucket containing $u$, 
 but also from at most $\zeta(d)$ lists representing the mugs containing $u$.
  \item When renaming bucket $\bucket[u,X,i]$ to $\bucket[\bot,X,i]$, rename also all the mugs of the form $\bucket[u,X,i,c]$ to $\bucket[\bot,X,i,c]$. Note that there are at most $\zeta(d)$ such mugs.
 \end{itemize}
 It is straightforward to see that these modifications turn $\Dt^\Pi[F,G]$ into $\Dt^\Pi[F-K,G/K]$. The bound on the length of $\appendices$ can be argued as in the proof of Lemma~\ref{lem:trim-implement}.
\end{proof}

\begin{lemma}\label{lem:extend-implement-scheme}
 Suppose $H$ is a graph, $K$ is a subset of vertices of $H$, and $F^K$ is an elimination forest of $H[K]$ of height at most~$d$. Suppose further that $R$ is a recursively connected elimination forest of $H-K$ of height at most $d$ that is attachable to $(H,F^K)$. Let $\wh{F}$ be the extension of $F^K$ via $R$. Then given on input $H[K]$, $F^K$, and the partial data structure $\Dt^\Pi[R,H/K]$, one can in time $2^{\Oh(d)}\cdot |K|^{\Oh(1)}\cdot \tau(d)\zeta(d)^{\Oh(1)}$ modify $\Dt^\Pi[R,H/K]$ to obtain the data structure $\Dt^\Pi[\wh{F},H]$.
\end{lemma}
\begin{proof}
We augment the procedure that modifies $\Dt[R,H/K]$ to $\Dt[\wh{F},H]$ that was described in the proof of Lemma~\ref{lem:extend-implement}. We denote $\tau\coloneqq \tau(d)$ and $\zeta\coloneqq \zeta(d)$; recall that these can be computed in time $d^{\Oh(1)}$.

Recall that the first step of the procedure of Lemma~\ref{lem:extend-implement} was to rename the buckets from the list $\appendices$ so as to model attaching trees $R_r$ for $r\in \roots_R$ in the construction  of $\wh{F}$. We apply the same renaming scheme: whenever a bucket $\bucket[\bot,X,i]$ is renamed to $\bucket[m,X,i]$, this renaming applies also to all the mugs $\bucket[\bot,X,i,c]$ that are sub-lists of $\bucket[\bot,X,i]$. Note that there are at most $\zeta$ such mugs per bucket.
 Recall that the key observation in the proof of Lemma~\ref{lem:extend-implement} was that after the renaming, all the data stored for vertices $w\notin K$ did not require further updating, and even the buckets $\bucket[u,\cdot,\cdot]$ for 
 $u\in K$ were as they should be in $\Dt[\wh{F},H]$, except they were lacking vertices of $K$.
 Observe that now, the same also applies to the mugs: mugs $\bucket[w,\cdot,\cdot,\cdot]$ for $w\notin K$ do not require updating, while mugs $\bucket[u,\cdot,\cdot,\cdot]$ for $u\in K\cup \{\bot\}$ 
 are as they should be in $\Dt^\Pi[\wh{F},H]$,
 except they lack vertices of $K$. 
 
 %This is because for $w\notin K$ we have $\Hb_w=\Gb_w$, where
 %$\Hb_w$ is the boundaried graph defined analogously to $\Gb_w$, but with respect to the graph $H$ and its elimination forest $\wh{F}$.
 
 Next, we may compute the values $\SReach(w)$, $\Up(w)$, $\height(w)$ for $w\in K$ and appropriately fill up the buckets $\bucket[w,\cdot,\cdot]$ for $w\in K\cup \{\bot\}$ exactly as in the proof of Lemma~\ref{lem:extend-implement}.  
 Hence, it remains to fill up the mugs $\bucket[w,\cdot,\cdot,\cdot]$ for $w\in K\cup \{\bot\}$ by inserting the lacking vertices of~$K$. 
 Similarly as in the proof of Lemma~\ref{lem:update}, for this it suffices to add each $u\in K$ to the mug
 \begin{equation}\label{eq:hamster}
 \bucket[\pnt_{F^K}(u),\SReach(u),\height(u),c]\qquad\textrm{for each }c\in \conf(\Hb_u).
 \end{equation}
 Here, $\Hb_u$ is a suitable boundaried graph defined as in the description of $\Dt^\Pi[\cdot,\cdot]$, but with respect to the graph $H$ and its elimination forest $\wh{F}$.
 Note that this can be easily done in time $\Oh(\zeta)$ provided we have computed the set $\conf(\Hb_u)$. 
 
 We now present a procedure that updates the mugs $\bucket[w,\cdot,\cdot,\cdot]$ for all $w\in K\cup \{\bot\}$ by processing vertices $u\in K$ in a bottom-up manner over the forest $F^K$.
 When processing $u$ we assume that all $v\in \chld_{F^K}(u)$ have already been inserted in appropriate mugs as prescribed in~\eqref{eq:hamster}, that is, all the mugs $\bucket[u,\cdot,\cdot,\cdot]$ are already
 correctly constructed. Based on this we compute $\conf(\Hb_u)$, so that $u$ itself can be inserted to appropriate mugs as prescribed in~\eqref{eq:hamster}.
 
 First, we construct a set~$W$ of vertices with $W\subseteq \chld_{F'}(u)$ as follows:
 for each $X\subseteq \SReach(u)\cup \{u\}$, $i\in \{1,\ldots,d\}$, and $c\in \Conf(X)$, include in $W$ the first $\tau$ elements of the mug $\bucket[u,X,i,c]$, or all of them in case their number is at most $\tau$.
 Note that this can be done in time $\Oh(\tau)$ per considered mug, so in total time $\Oh(2^d\cdot d \tau\zeta)$ per vertex $u$. Also, we have $|W|\leq 2^d\cdot d \tau\zeta$.
 Next, we construct the multiset
 $$\Mm\coloneqq \{\{\,\conf(\Hb_v)\ \colon\ v\in W\,\}\}.$$
 This can be done in time $\Oh(|W|\cdot \zeta)$ as follows: for each $v\in W$, we deduce $\conf(\Hb_v)$ by inspecting the list element of $v$ and checking in which mugs it is contained.
 
 Let $\mathbf{I}_{uu'}$ be the boundaried graph $((\{u,u'\},\{uu'\}),\{u,u'\})$; that is, it consists only of two boundary vertices $u,u'$ and the edge $uu'$.
 We finally compute $\conf(\Hb_u)$ using the following formula:
 \begin{equation}\label{eq:walrus}
   \conf(\Hb_u) = \forget\left(\quad \bigoplus_{u'\in \Up(u)} \conf(\mathbf{I}_{uu'})\ \oplus\ \bigoplus_{D\in \Mm} D,\quad u\quad \right).
 \end{equation}
 Note that formula~\eqref{eq:walrus} can be computed in time $$(d+|\Mm|)\cdot \zeta^{\Oh(1)}=2^d\cdot d \tau\zeta^{\Oh(1)} ,$$
 because it involves $\Oh(d+|\Mm|)$ operations $\oplus$ and $\forget$, plus $\Oh(d)$ operations $\conf(\cdot,\cdot)$ applied to a two-vertex graph: each of these operations can be carried out in time $\zeta^{\Oh(1)}$ 
 due to the efficient composability of $(\Conf,\conf)$. Once $\conf(\Hb_u)$ is computed, the vertex $u$ can be added to appropriate mugs as described in~\eqref{eq:hamster} in time $\Oh(\zeta)$.
 
 It remains to argue that formula~\eqref{eq:walrus} is correct. First, observe that
 $$
   \Hb_u = \forget\left(\quad \bigoplus_{u'\in \Up(u)} \mathbf{I}_{uu'}\ \oplus\ \bigoplus_{v\in \chld_{F'}(u)} \Hb_v,\quad u\quad \right).
 $$
 Hence, by the composability of $(\Conf,\conf)$ we have
 \begin{equation}\label{eq:seal}
   \conf(\Hb_u) = \forget\left(\quad \bigoplus_{u'\in \Up(u)} \conf(\mathbf{I}_{uu'})\ \oplus\  \bigoplus_{v\in \chld_{F'}(u)} \conf(\Hb_v),\quad u\quad \right).
 \end{equation}
 However, the construction of $W$ and the idempotence of $(\Conf,\conf)$ implies that
 \begin{equation}\label{eq:penguin}
   \bigoplus_{D\in \Mm} D = \bigoplus_{v\in W} \conf(\Hb_v) = \bigoplus_{v\in \chld_{F'}(u)} \conf(\Hb_v).
 \end{equation}
 Now~\eqref{eq:seal} and~\eqref{eq:penguin} imply the correctness of~\eqref{eq:walrus}.
 
 \medskip
 
 This concludes the description of updating the mugs.
 To see that the time complexity of the update is as promised, note that the time spent on processing a single $u\in K$ as above is bounded by $2^d\cdot d \tau\zeta^{\Oh(1)}$.
 Hence, the total time complexity of $2^{\Oh(d)}\cdot |K|^{\Oh(1)}\cdot \tau\zeta^{\Oh(1)}$ follows.
 
 We are left with recomputing the $\mmb$ flag.
 In the paragraphs above, we have described the computation of $\conf(\Hb_u)$ for any vertex $u$ based on the access to correctly updated mugs $\bucket[u,\cdot,\cdot,\cdot]$,
 in time $2^d\cdot d \tau\zeta^{\Oh(1)}$.
 We can apply the same reasoning for $u=\bot$, and thus compute within the same running time the value $\conf(\Hb_\bot)=\conf(H,\emptyset)$.
 To verify whether $H\in \Pi$, which is the new value of $\mmb$, it suffices to check whether whether $\conf(H,\emptyset)$ contains one of final configurations of $(\Conf,\conf)$. 
\end{proof}

\begin{lemma}\label{lem:update-conf}
  Suppose that we have access to the data structure $\Dt^\Pi[F,G]$, where $G$ is a graph and $F$ is a recursively optimal elimination forest $G$ of height at most $d$.
 Let $H$ be a graph obtained from $G$ by inserting or removing a single edge, given as input. 
 Then one can in $2^{\Oh(d^2)}\cdot \tau(d)\zeta(d)^{\Oh(1)}$ time either conclude that $\td(H)>d$, or modify $\Dt^\Pi[F,G]$ to obtain the data structure $\Dt^\Pi[\wh{F},H]$ where $\wh{F}$ is some recursively optimal elimination forest of $H$ of height at most $d$.
\end{lemma}
\begin{proof}
 We apply the same reasoning as in Lemma~\ref{lem:update}, except that we use of Lemmas~\ref{lem:trim-implement-scheme} and~\ref{lem:extend-implement-scheme} instead of Lemmas~\ref{lem:trim-implement} and~\ref{lem:extend-implement}, respectively.
\end{proof}

From Lemma~\ref{lem:update-conf} we immediately obtain the following conclusion.

\begin{lemma}\label{lem:scheme-maintain}
 Let $\Pi$ be a graph property that is recognized by a tractable configuration scheme with polynomial-time computable witnesses~$\zeta(\cdot)$ and $\tau(\cdot)$.
 Then there exists a data structure that maintains a dynamic graph $G$ of treedepth at most $d$ under edge insertions and deletions with update time $2^{\Oh(d^2)}\cdot \tau(d)\zeta(d)^{\Oh(1)}$
 while offering a query on whether $G\in \Pi$ in time $\Oh(1)$.
 Upon inserting an edge that breaks the invariant $\td(G)\leq d$, the data structure detects this and refuses to carry out the operation. The data structure uses $\Oh(n (d+\zeta(d)))$ memory.
 The initialization of the data structure for an edgeless $n$-vertex graph takes time $\Oh(n)$.
\end{lemma}

Finally, by combining Lemmas~\ref{lem:schem-kpath} and~\ref{lem:scheme-maintain} we immediately obtain the following.

\begin{lemma}\label{lem:kpath-bnd-td}
 There exists a data structure that maintains a dynamic graph $G$ of treedepth smaller than $k$ under edge insertions and deletions with update time $2^{\Oh(k^2)}$
 while offering a query on whether $G$ contains a simple path on $k$ vertices in time $\Oh(1)$.
 Upon inserting an edge that breaks the invariant $\td(G)<k$, the data structure detects this and refuses to carry out the operation. The data structure uses $\Oh(n \cdot 2^{\Oh(k \log k)})$ memory.
 The initialization of the data structure for an edgeless $n$-vertex graph takes $\Oh(n)$ time.
\end{lemma}

Let us remark that by combining Lemma~\ref{lem:scheme-maintain} with Lemma~\ref{lem:mso-tractable} in the same manner, we may recover the result of Dvo\v{r}\'ak et al.~\cite{DvorakKT14}:
for every fixed $\MSO_2$-definable property $\Pi$, there is a data structure that for a dynamic graph of treedepth at most $d$ maintains the information about $\Pi$-membership with update time $f(d)$, for some function $f$.

\section{Property Recognition under Merging and Splitting}\label{sec:scheme-ms}

\CycleNote

In this section we make a step towards the data structure for dynamic cycle detection.
More precisely, we provide a data structure for the dynamic maintenance of a suitable partition of a dynamic graph $G$ into subgraphs.
The intuition is that this partition will later be used to maintain the tree of biconnected components of $G$.

%This interface is later used in Section~\ref{sec:cyc_det_ds} to implement the dynamic cycle detection structure. 

Before we proceed, we need to slightly extend the results from previous sections.  

\paragraph*{Path queries.} We will need the following result, which essentially says that, given access to the data structure for the {\sc{$k$-Path}} problem provided by the combination of Lemmas~\ref{lem:schem-kpath} and~\ref{lem:scheme-maintain}, one can efficiently answer queries about the existence of a path on at least $k$ vertices with a given pair of endpoints.

\begin{lemma}\label{lem:path_queries}
 Let $k\in \N$, let $\Pi_k$ be the graph property comprising all graphs that contain a simple path on $k$ vertices, and let $(\Conf,\conf)$ be the configuration scheme for the property $\Pi_k$ provided by Lemma~\ref{lem:schem-kpath}. Suppose we are given a graph $G$, a pair of its vertices $u$ and $v$, and access to the data structure $\Dt^{\Pi_k}[F,G]$ constructed using the scheme $(\Conf,\conf)$, where $F$ is a recursively optimal elimination forest of $G$ of height at most $d$. Then one can in time $2^{\Oh(k)}\cdot (dk)^{\Oh(d)}$ answer the following queries:
 \begin{itemize}[nosep]
  \item Given $u,v\in V(G)$ and $i\leq k$, find a simple path in $G$ on exactly $i$ vertices and with endpoints $u$ and $v$, or correctly conclude that no such path exists.
  \item Given $u,v\in V(G)$, decide whether in $G$ there exists a simple path on at least $k$ vertices with endpoints $u$ and $v$.
 \end{itemize}
\end{lemma}
\begin{proof}
 We adopt the notation for the configuration scheme $(\Conf,\conf)$ introduced in the proof of Lemma~\ref{lem:schem-kpath}, as well as the notation for the scheme $\Dt^{\Pi_k}[\cdot,\cdot]$ introduced before Lemma~\ref{lem:extend-implement-scheme}. By a {\em{$u$-$v$ path}} we mean a simple path with endpoints $u$ and $v$. Moreover, the length of the path is its number of edges, unless specified otherwise.
 
 Let $r$ be the root of $F$ that is an ancestor of $u$ in $F$. 
 If $r$ is not an ancestor of $v$ as well, then $u$ and $v$ lie in different connected components of $G$, hence we immediately conclude that the sought path does not exist.
 Note that the roots of~$F$ that are ancestors of~$u$ and~$v$ can be computed in the required time bound using access to~$\Dt^{\Pi_k}[F,G]$.
 Let $T=F_r$; then $u,v\in V(T)$.
 
 For $X\subseteq \Omega$, a pair of different vertices $x,y\in X$, and for $j\in \{2, 3,\ldots,k-1,\infty\}$, we define the configuration $c_X^{xy,j}\in \Conf(X)$ as $(J_X^{xy},j)$, where $J_X^{xy}$ is the graph on vertex set $X\cup \{s,t\}$ where only the edge $xy$ is present and no other edge.
 For $w\in V(T)$, we write $c_w^{xy,j}\coloneqq c_{\SReach_{F,G}(w)}^{xy,j}$ for brevity.
 From the definition of the scheme $(\Conf,\conf)$ (see the proof of Lemma~\ref{lem:schem-kpath}) we immediately see the following.
 
 \begin{claim}\label{cl:conf-paths}
   For every $w\in V(T)$, a pair of different vertices $x,y\in \SReach_{F,G}(w)$, and $j\in \{2,3,\ldots,k-1\}$, we have $c_w^{xy,j}\in \conf(\Gb_w)$ if and only if in $G_w$ there is a $x$-$y$ path of length exactly $j$. Moreover, $c_w^{xy,\infty}\in \conf(\Gb_w)$ if and only if in $G_w$ there is a $x$-$y$ path of length at least $k$.
 \end{claim}

 We now define a prefix $K$ of $T$ through a recursive marking procedure as follows.
 Given a vertex $w\in V(T)$, the marking procedure first marks $w$ if it has not been marked before.
 Then, it iterates over every pair of distinct vertices $x,y\in \SReach_{F,G}(w)\cup \{w\}$ and each $j\in \{2,3,\ldots,k-1,\infty\}$, and inspects each mug $\bucket[w,X,i,c_X^{xy,j}]$ that is present in $\Dt^{\Pi_k}[F,G]$ such that $x,y\in X$.
 The procedure marks $d$ arbitrary vertices contained in any of these mugs, or all of the vertices in all mugs if the number of them is smaller then $d$.
 Additionally, the procedure marks any children of $w$ in $F$ that are ancestors of either $u$ or~$v$.
 Finally, the procedure is applied recursively to all the children of~$w$ in~$F$ that got marked.
 We let $K$ be the set of all the vertices of $T$ that get marked after applying the procedure to the root $r$.
 Since the marking procedure applied to a vertex~$w$ marks at most $\binom{d}{2}\cdot (k-1)\cdot d+2\leq d^3k$ children of~$w$ and the height of $T$ is at most $d$, we conclude that
 $$|K|\leq (d^3k)^0+(d^3k)^1+\ldots+(d^3k)^{d-1}\leq (dk)^{\Oh(d)}.$$
 Also, given access to $\Dt^{\Pi_k}[F,G]$ we may compute $K$ and $G[K]$ in time $(dk)^{\Oh(d)}$ using a similar implementation as in Lemma~\ref{lem:extract-core}.
 
 By construction we have $u,v\in K$.
 We now observe that $G[K]$ contains all the necessary information about $u$-$v$ paths with endpoints $u$ and $v$. The proof is similar to that of Claim~\ref{cl:idempotent} in the proof of Lemma~\ref{lem:schem-kpath}.
 
 \begin{claim}\label{cl:core-paths-correct}
  For each $i\in \{1,2\ldots,k-1\}$, there is a $u$-$v$ path of length $i$ in $G$ if and only if there is one in~$G[K]$. Moreover, there is a $u$-$v$ path of length at least $k$ in $G$ if and only if there is one in $G[K]$.
 \end{claim}
 \begin{clproof}
   The claim is trivial for $i=1$ and, for both assertions, the right-to-left implication holds vacuously. Hence, we need to prove the left-to-right implication for $i\in \{2,\ldots,k-1,\infty\}$, where $\infty$ corresponds to the case of a path of length at least $k$. More precisely, by a path of length $i\in \{2,\ldots,k-1,\infty\}$ we mean a path consisting of exactly $i$ edges in case $i<\infty$, or at least $k$ edges in case $i=\infty$.
   We need to prove that, if in $G$ there is a $u$-$v$ path of length $i\in \{2,\ldots,k-1,\infty\}$, then there is also one such that all of its vertices belong to~$K$. 
  
   For a vertex $w\in V(T)\setminus K$, let $f(w)$ denote the smallest depth in $T$ of a descendant of $w$ that is not in~$K$.
   In formulas,
   $$f(w)\coloneqq \min\ \{\,\depth_T(a)\ \colon\ a\in \desc_T(w)\setminus K\,\}.$$
  For a path $P$ in~$G$ and $\ell\in \{1,\ldots,d\}$, we let $f_\ell(P)=|\{\,w\in V(P)\setminus K\ \mid\ f(w)=\ell\,\}|$. Finally, we let
  $$f(P)=(f_1(P),f_2(P),\ldots,f_d(P)).$$
  Among $u$-$v$ paths of length~$i$ in $G$, let us choose $P$ to be an arbitrary path that lexicographically minimizes the vector~$f(P)$.
  We suppose for a contradiction that $P$ contains some vertices outside of $K$.
  Let $Q$ be a maximal subpath of $P$ that starts and ends in vertices of $K$ and such that all the internal vertices of $Q$ do not belong to $K$.
  Note that such a $Q$ exists due to $V(P)\setminus K\neq \emptyset$ and $u,v\in K$. Let $x$ and $y$ be the endpoints of $Q$; then $x,y\in K$.
  
  Since $Q$ is connected and $K$ is a prefix of $F$, there exists $a\in \App_F(K)$ such that $x,y\in \SReach_{F,G}(a)$ and $V(Q)\setminus \{x,y\}\subseteq \desc_F(a)$. Let $j\coloneqq |E(Q)|$ if $|E(Q)|<k$, or $j\coloneqq \infty$ otherwise. Then, according to Claim~\ref{cl:conf-paths}, $Q$ witnesses that $c^{xy,j}_a\in \conf(\Gb_a)$. Since $a\notin K$, by the construction of $K$ we infer that there is a set $B\subseteq K$ consisting of $d$ siblings of $a$ such that for each $b\in B$, we have $x,y\in \SReach_{F,G}(b)$ and $c_b^{xy,j}\in \conf(\Gb_b)$. By Claim~\ref{cl:conf-paths}, for each $b\in B$ there exists a path $Q_b$ such that $Q_b$ has endpoints $x$ and $y$, length $j$, and all its internal vertices belong to $\desc_T(b)$.
  
  We now claim that $P$ cannot intersect each of the sets $\desc_T(b)$ for $b\in B$.
  Suppose otherwise, that is, for each $b\in B$, there exists a vertex $p_b\in V(P)\cap \desc_T(b)$. Let also $p_a$ be any internal vertex of $Q$. Observe that, since $F$ is an elimination forest of $G$, on $P$ between every two vertices of $\{p_b\colon b\in B\}\cup \{p_a\}$ there must be a vertex of $\anc_T(a)\setminus \{a\}$. Since the former set has size $d+1$ and the latter has size at most $\height(T)-1\leq d-1$, this is a contradiction.
  
  Therefore, there exists $b\in B$ such that $V(P)\cap \desc_T(b)=\emptyset$.
  In particular, no internal vertex of~$Q_b$ lies on $P$. Let $P'$ be the path obtained from $P$ by replacing the subpath $Q$ with $Q_b$. Then $P'$ is a $u$-$v$ path. Moreover, as the edge count of $Q_b$ is either equal to that of $Q$, or at least $k$ anyway, we conclude that $P'$ also has length $i$. Finally, observe that, since $b\in K$, for each $z'\in V(Q_b)\setminus K$ we have $f(z')>\depth_T(b)=\depth_T(a)$. As $f(z)=\depth_T(a)$ for each $z\in V(Q)\setminus \{x,y\}$, we conclude that $f(P')$ is lexicographically strictly smaller than $f(P)$. This is a contradiction to the choice of~$P$.
 \end{clproof}

 By Claim~\ref{cl:core-paths-correct} we conclude that for both queries considered in the lemma statement, it suffices to answer them in graph $G[K]$ instead of $G$. As $|K|\leq (dk)^{\Oh(d)}$, we may do it using any suitably efficient static FPT algorithm for the considered problem, applied in $G[K]$. Precisely, for the first query (about a $u$-$v$ path on exactly $i\leq k$ vertices) we may apply a standard color-coding algorithm~\cite{AlonYZ95}, which then runs in time $2^{\Oh(i)}\cdot |K|^{\Oh(1)}\leq 2^{\Oh(k)}\cdot (dk)^{\Oh(d)}$. For the second query (about a $u$-$v$ path on more than $k$ vertices), we may use the static algorithm of Bez\'akov\'a et al.~\cite{BezakovaCDF19}, which runs in time $2^{\Oh(k)}\cdot |K|^{\Oh(1)}=2^{\Oh(k)}\cdot (dk)^{\Oh(d)}$. Let us note here that this algorithm works even under the parameterization where $k$ is the difference between the requested path length and the (shortest path) distance between $u$ and $v$, which is a smaller parameter than just the requested path length. 
\end{proof}

%%% Local Variables:
%%% mode: latex
%%% TeX-master: "main"
%%% End:

\paragraph*{Maintaining a partition.}
Equipped with the extended path queries provided by Lemma~\ref{lem:path_queries}, we are ready to explain the main goal of this section. We start by a definition of a nice partition. %For two disjoint graphs $G_1$ and $G_2$ with distinguished vertices $v_1$ and $v_2$, we denote their coalescence as $(G_1,v_1) \cdot (G_2,v_2)$. The coalescence graph $(G_1,v_1) \cdot (G_2,v_2)$ is formed by identyfing $v_2$ and $v_2$, that is $v_1$ and $v_2$ are replaced by a single vertex $v$ adjacent to the same vertices in $G_1$ as $v_1$ and to the same vertices in $G_2$ as $v_2$.

\begin{definition}
A \emph{nice partition} of a graph $G$ is a family $\Prt=\{H_1, \ldots, H_\ell\}$ of subgraphs of $G$ satisfying the following conditions:
\begin{itemize}[nosep]
 \item every edge of $G$ belongs to exactly one among the graphs $H_1,\ldots,H_\ell$;
 \item graphs $H_1,\ldots,H_\ell$ are connected and each is the union of some biconnected components of $G$.
\end{itemize}
Graphs $H_1,\ldots,H_\ell$ will be called the {\em{parts}} of the partition $\Prt$.
\end{definition}

Note that for a graph $G$, the partition of $G$ into biconnected components is the finest nice partition of $G$, while the partition into connected components is the coarsest one. Further, any two parts in a nice partition are either disjoint or intersect at exactly one vertex. Finally, if $\Prt$ is a nice partition of $G$ and $C$ is a cycle in~$G$, then $C$ is entirely contained in a single part of $\Prt$.

The next lemma states that there is data structure that can maintain a nice partition of a graph while offering access to it through multiple query and update operations.

\begin{lemma}\label{lem:kpath-bnd-td-ms}
 Let $d$ and $k$ be fixed parameters where $k \leq d$. Suppose we are given a dynamic graph $G$ on $n$ vertices and access to a dictionary on the edges of $G$, using $\mdict$ memory, where each dictionary operation runs in time bounded by $\tdict$. 
 Then there exists a data structure $\Tdep_{d,k}[G,\Prt]$ that maintains a dynamic nice partition $\Prt$ of $G$. The data structure works under the assumption that, at all times, $\td(H) \leq d$ for all $H \in \Prt$. The data structure uses $(n \cdot 2^{\Oh(d \log d)} + \mdict)$ space and offers three main queries: 
 
 \medskip
 
 \noindent\begin{tabular}{| p{\textwidth} |}
    \hline
   \underline{$\pathlb(u,v,ux,vy)$}\\[0.2cm]
   \begin{tabular}{ l p{0.87\textwidth} }
   takes: & vertices $u,v \in V(G)$ and edges $uv,vy \in E(G)$ such that there is $H \in \Prt$ satisfying $ux,vy \in E(H)$ \\
   returns: & $\true$ if $H$ contains a simple $u$ to $v$ path on at least $k$ vertices, otherwise $\false$ \\
   time: & $d^{\Oh(d)}+\Oh(\tdict)$
   \end{tabular}\\[0.3cm]
    \hline
    
   \underline{$\pathub(i,u,v,ux,vy)$}\\[0.2cm]
   \begin{tabular}{ l p{0.87\textwidth} }
   takes: & parameter $i \leq k$, two vertices $u,v \in V(G)$ and two edges $uv,vy \in E(G)$ such that there is $H \in \Prt$ satisfying $ux,vy \in E(H)$ \\
   returns: & a simple path from $u$ to $v$ in $H$ on precisely $i$ vertices, if it exists; otherwise $\nil$\\
   time: & $d^{\Oh(d)}+\Oh(\tdict)$
   \end{tabular}\\[0.3cm]
    \hline
    
  \underline{$\art(vx,vy)$}\\[0.2cm]
  \begin{tabular}{ l p{0.87\textwidth} }
  takes: & edges $vx,vy\in E(G)$ such that there is $H\in \Prt$ satisfying $vx,vy\in V(H)$\\
  returns: & $\true$ if $x$ and $y$ are not in the same connected component of $H-v$, otherwise $\false$\\
  time: & $d^{\Oh(d)}$
   \end{tabular}\\
    \hline
    
\end{tabular}

\medskip

\noindent Furthermore, structure $\Tdep_{d,k}[G,\Prt]$ offers the following update operations:

\medskip
 
 \noindent\begin{tabular}{| p{\textwidth} |}
    \hline
\underline{$\ins(u,v,ux,vy)$}\\[0.2cm]
\begin{tabular}{ l p{0.87\textwidth} }
 takes: & vertices $u,v \in V(G)$ and edges $ux,vy \in E(G)$ such that $uv\notin E(G)$ and there is $H \in \Prt$ satisfying $ux,vy \in E(H)$\\
 does: & inserts an edge $uv$ to $H$, thus effectively adding $uv$ to $G$; in case $\td(H)>d$ after the insertion, throws an exception\\
 time: & $2^{\Oh(d^2)}+d^{\Oh(d)}\tdict$
   \end{tabular}\\[0.3cm]
    \hline
    
  \underline{$\rem(uv)$}\\[0.2cm]
\begin{tabular}{ l p{0.87\textwidth} }
  takes: & edge $uv \in E(G)$ which is not a bridge in $G$\\
  does: & removes edge $uv$ from the part $H\in \Prt$ that contains it, thus effectively removing $uv$ from $G$\\
  time: & $2^{\Oh(d^2)}+d^{\Oh(d)}\tdict$
   \end{tabular}\\[0.3cm]
    \hline
    
    \underline{$\merge(vx,vy)$}\\[0.2cm]
\begin{tabular}{ l p{0.87\textwidth} }
  takes: & edges $vx,vy \in E(G)$ such that the parts $I,J\in \Prt$ that satisfy $vx\in I$ and $vy\in J$ are different and $J$ is biconnected\\ 
  does: & merges $I$ and $J$ into one graph in $\Prt$, that is, it substitutes $I$ and $J$ in $\Prt$ with the graph $H=G[V(I) \cup V(J)]$; in case $\td(H)>d$, it throws an exception\\
  time: & $2^{\Oh(d^2)}+d^{\Oh(d)}\tdict$
   \end{tabular}\\[0.3cm]
    \hline
    
  \underline{$\splitt(vx,vy)$}\\[0.2cm]
\begin{tabular}{ l p{0.87\textwidth} }
  takes: & edges $vx,vy \in E(G)$ such that there is $H \in \Prt$ satisfying the following: $vx,vy \in E(H)$, $v$ is a cut-vertex in $H$, and $H$ has exactly two biconnected components $I$ and $J$ where $vx\in E(I)$ and $vy\in E(J)$\\
  does: & splits $H$ in $\Prt$ into $I$ and $J$\\
  time: & $2^{\Oh(d^2)}+d^{\Oh(d)}\tdict$
   \end{tabular}\\
    \hline
\end{tabular}
\smallskip

\noindent In addition to that, the data structure offers the following simple queries: \begin{itemize}[nosep]
\item $\edge(uv)$, which checks whether $uv \in E(G)$;
\item $\bridge(uv)$, which checks whether the part of $\Prt$ containing $uv$ is a single edge;
\item $\same(uv,u'v')$, which checks whether $uv,u'v'\in E(G)$ belong to the same part in $\Prt$.
\end{itemize}
The running times of these queries are bounded by $\Oh(\tdict+d)$.
Finally, $\Tdep_{d,k}[G,\Prt]$ also offers the following simple update methods:
\begin{itemize}[nosep]
\item $\new(u,v)$: for two vertices $u,v\in V(G)$ promised to be not in the same connected component of $G$, add the edge $uv$ to $G$ and create a new part $H \in \Prt$ consisting only of $uv$ and its endpoints.
\item $\destroy(uv)$: for an edge $uv \in E(G)$ such that the part $H\in \Prt$ that contains $uv$ consists only of $uv$, remove $uv$ from $G$ and remove the part $H$ in $\Prt$.
\end{itemize}
These methods take $\Oh(\tdict)$ time.
\end{lemma}
\begin{proof}
In $\Tdep_{d,k}[G,\Prt]$, each vertex $v\in V(G)$ can be represented through several aliases, which we call {\em{copies}}: one {\em{global copy}} $v$, and one {\em{local copy}} $v_H$ for each $H\in \Prt$ such that $v\in V(H)$. Each local copy $v_H$ maintains a pointer $\glob(v_H)=v$ pointing to the corresponding global copy.

The data structure $\Tdep_{d,k}[G,\Prt]$ maintains the following internal data structures:
\begin{itemize}
 \item A global array $\verts$ which stores all global copies of the vertices of~$G$. This array is never modified. The pointers $\glob(\cdot)$ point to the corresponding cells in this array.
 
 \item For each $H\in \Prt$, a data structure $\Dt_{\gl}^{\Pi_k}[T,H]$ that stores a recursively optimal elimination tree $T$ of $H$ of depth at most $d$ (note here that we assume that $H$ is connected, hence $T$ will be always a tree). This data structure is a slight modification of the structure $\Dt^{\Pi_k}[T,H]$ given by Lemma~\ref{lem:update-conf} for the configuration scheme provided by Lemma~\ref{lem:schem-kpath}. Let us now explain this modification. In $\Dt_{\gl}^{\Pi_k}[T,H]$, all the data about vertices of $H$, such as the parent pointer and the bucket elements, are stored using their local copies associated with $H$, with three exceptions. Namely, for each $v\in V(H)$, the sets $\SReach(v_H)$ and $\Up(v_H)$, and the sets $X$ indexing buckets $\bucket[v_H,X,i]$ and mugs $\bucket[v_H,X,i,c]$, are stored in $\Dt_{\gl}^{\Pi_k}[T,H]$ using the global copies of vertices, rather than the local copies (hence, also the configuration $c$ is encoded using global copies, as it depends on $X$). The intuition behind this modification is that it allows us to keep these sets up to date when local copies get destroyed or created during merges and splits.
 
 Finally, the root of $T$ also stores a reference to the whole data structure $\Dt_{\gl}^{\Pi_k}[T,H]$. Note that, thus, given a local copy $u_H$ of a vertex of $H$, we can obtain a reference to $\Dt_{\gl}^{\Pi_k}[T,H]$ in time $\Oh(d)$ by following the parent pointers up to the root of $T$, and finding the reference there.

\item A global dictionary $\edges$ wherein the keys are the edges of $G$ (i.e., pairs of global copies of vertices). Each edge $uv \in E(G)$ is associated with the following value. Let $H \in \Prt$ be such that $uv \in E(H)$. $\Dict[uv]$ stores the local copy of the endpoint of $uv$ that is deeper in the elimination tree $T$, stored in $\Dt_{\gl}^{\Pi_k}[T,H]$. Let us denote this value by $\lowe(uv)$, i.e. $\edges[uv]=\lowe(uv)$. %Since we leave the choice of dictionary implementation unspecified, we denote by $\tdict$ the time needed per one dictionary operation (see Section~\ref{sec:kpath} for a more comprehensive overview of dictionary data structure).   
\end{itemize}

Let us first bound the memory usage. Array $\verts$ uses $\Oh(n)$ memory, while dictionary $\edges$ uses $\mdict$ memory. Since each $\Dt_{\gl}^{\Pi_k}[T,H]$ structure does not store any additional data compared to $\Dt_{\gl}^{\Pi_k}[T,H]$, only encodes it differently, by Lemma~\ref{lem:scheme-maintain} and Lemma~\ref{lem:schem-kpath} it occupies $|V(H)| \cdot 2^{\Oh(d \log d)}$ memory. In total this takes $2^{\Oh(d \log d)} \sum_{H \in \Prt} |V(H)| =2^{\Oh(d \log d)} |E(G)| =n \cdot 2^{\Oh(d \log d)}$, where the last equality follows from $|E(G)| \leq d n$. 

\paragraph*{Global and local indexing.}
Before we move on to implementing the interface of $\Tdep_{d,k}[G,\Prt]$ specified in the lemma statement, let us elaborate a bit more on the structure $\Dt_{\gl}^{\Pi_k}[T,H]$ that we maintain for each $H \in \Prt$. We would like to use methods that we have already implemented for $\Dt^{\Pi_k}[T,H]$: $\core()$ (Lemma~\ref{lem:extract-core}), $\trim()$ (Lemmas~\ref{lem:trim-implement-scheme} and Lemma~\ref{lem:trim-implement}), $\extend()$ (Lemma~\ref{lem:extend-implement-scheme} and Lemma~\ref{lem:extend-implement}), and edge updates (Lemmas~\ref{lem:update-conf} and Lemma~\ref{lem:update}). The caveat is that in $\Dt_{\gl}^{\Pi_k}[T,H]$ we use global references for $\SReach(v_H)$, $\Up(v_H)$ and the sets $X$ indexing buckets and mugs.
We argue next that this is in fact not an issue, if we are willing to accept polynomial overheads in the running time.

The crucial observation is that the partition into buckets and mugs in $\Dt_{\gl}^{\Pi_k}[T,H]$ is exactly the same as in $\Dt^{\Pi_k}[T,H]$, just with the local indexing translated into a global one. Moreover, given a local copy~$v_H$ of a vertex $v\in V(H)$, and any set $Y_\gl$ consisting of global copies of vertices lying on the path in $T$ from $v$ to the root of $T$, we can reconstruct its local counterpart $Y_\loc$ --- consisting of the $H$-associated local copies of vertices of $Y_\gl$ --- as follows. Starting from $v_H$ and following the parent pointers in $T$, obtain the set of local copies of ancestors of $v$ in $T$, and then check for which of those ancestors their global copies --- accessed using $\glob()$ --- are contained in $Y_\gl$. This takes $\Oh(d \cdot |Y|) \leq \Oh(d^2)$ time. Note that the sets $\SReach(v_H)$, $\Up(v_H)$ and the sets $X$ indexing buckets and mugs are such that we can compute their local counterparts in this way. Moreover, having $X$ translated into local copies $X_\loc$, allows translating any configuration $c \in \conf(X)$. Translating $Y_\loc$ back to $Y_\gl$ is automatic using the $\glob()$ pointers. So to sum up, the structure $\Dt_{\gl}^{\Pi_k}[T,H]$ can simulate the methods we implemented for $\Dt^{\Pi_k}[T,H]$ --- $\core()$ and $\trim()$ --- with an additional $\Oh(d^2)$ overhead in the running time.

Further explanation is, however, needed for the method $\extend()$, which normally operates on a partial data structure $\Dt^{\Pi_{k}}[R,H/K]$, a graph $H[K]$, and a forest $F^K$ on vertices of $K$ such that $R$ is attachable to $F^K$. We similarly assume that after the modification, the method should work on a modified partial data structure $\Dt^{\Pi_{k}}_\gl[R,H/K]$, where sets $\SReach(v_H)$, $\Up(v_H)$, and sets indexing buckets and mugs are stored using global copies instead of local ones. Note that we also assume that the forest $F^K$ is provided to the method using a parent function on the local copies of vertices of $K$. Then given a set $Y_\gl$ represented using global copies, for instance $Y_\gl=\SReach(v_H)$ for some local copy $v_H\in V(H)\setminus K$ provided on input, we can compute the corresponding local counterpart $Y_\loc$ by translating the global references for all $y\in Y\setminus K$ as in the previous paragraph, while for each $y\in Y\cap K$ we can iterate through the whole $K=V(F^K)$ to find a local copy that corresponds to the given global copy. This takes time $\Oh(|Y|(d+|K|))\leq \Oh(d^2+d|K|)$. The translation from global copies to local ones is again straightforward using the $\glob()$ pointers. We conclude that due to the translation between global and local copies, we get an $\Oh(d^2+d|K|)$ overhead in the implementation of method $\extend()$.

%Note that for this procedure, $F^K$ is only provided through the parent pointers on the local copies

%a tree structure on $K$ with a parent access, there are no sets $\SReach(v_H)$, $\Up(v_H)$, buckets or mugs associated with vertices $v_H \in K$. The sets  
%$\SReach(v_H)$, $\Up(v_H)$, buckets and mugs are associated only with vertices of $R$ by a partial structure $\Dt[R,H/K]$. In this situation, for any vertex $v_H \in R$, we can also reconstruct and set $Y$ of interest, as its local copies are either on a path from $v_H$ to the root in $R$, or they are in $K$, so we iterate through the path and $K$ to find them. This takes $|Y| \cdot d \cdot |K| \leq d^2 |K|$ time. 

Summing up, in $\Dt_{\gl}^{\Pi_k}[T,H]$ we can implement $\core()$ and $\trim()$ with an additional overhead of $\Oh(d^2)$ in the running time, while in the implementation of $\extend()$ we get an $\Oh(d^2+d|K|)$ additional overhead. As can be seen in Lemma~\ref{lem:update} and Lemma~\ref{lem:update-conf}, these three operations are sufficient to implement edge insertion and removal, provided that we are given the local copies of the endpoints. Moreover, we have $|K|\leq d^{\Oh(d)}$ in these applications. Thus, we conclude that edge addition and removal in $\Dt_{\gl}^{\Pi_k}[T,H]$ take time $2^{\Oh(d^2)}$. 

\paragraph*{Auxiliary methods.}
Next, we introduce a few auxiliary methods of the data structure $\Tdep_{d,k}[G,\Prt]$, which will be helpful in the forthcoming implementations of the public methods.

First, we will use a method $\cpy(K)$. This method is given a subset of vertices $K\subseteq V(G)$, provided using global copies. It constructs a copy $G'$ of the induced subgraph $G[K]$ using fresh copies. That is, for each $v\in K$ a new copy $v_{G'}$ is created, and the pointer $\glob(v_{G'})\coloneqq v$ is appropriately set. For each pair of vertices $u,v\in K$ we check whether $uv\in E(G)$ by a look-up in $\edges$. If so, we add $u_{G'}v_{G'}$ to the edge set of $G'$. Thus, $\cpy(K)$ takes time $\Oh(|K|^2\cdot \tdict)$.

Next, we will use a method $\retrieve(u,uv)$. This method takes a global vertex $u \in V(G)$ and a global edge $uv \in E(G)$, and returns a local copy $u_H$ of $u$, where $H\in \Prt$ is the part such that $uv \in E(H)$. This method is presented using pseudocode as Algorithm~\ref{alg:retrieve} and works as follows. First, we look up the value $x=\edges[uv]$, which is equal to either to $u_H$ or to $v_H$, whichever is lower in the stored elimination tree of $H$. If $\glob(x)=u$, then $x=u_H$ and we can return $x$. Otherwise $x=v_H$ and $u_H$ is among the ancestors of $v_H$, so we may look for $u_H$ by following the parent pointers. Thus, this method takes $\Oh(\tdict+d)$ time. 

Finally, we will also use method $\find(uv)$. This method  finds (and returns via a reference) the structure $\Dt_{\gl}^{\Pi_k}[T,H]$ where $H \in \Prt$ is such that $uv \in E(H)$. The implementation is as follows. First, we look up $x=\edges[uv]$, which is equal to either $u_H$ or $v_H$. Then we follow the parent pointers from $x$ to the root of $T$, where $T$ is the elimination tree stored in the structure $\Dt_{\gl}^{\Pi_k}[T,H]$ associated with $H$. The reference to $\Dt_{\gl}^{\Pi_k}[T,H]$ can be found at this root. Thus, this method also takes $\Oh(\tdict + d)$ time.

 \begin{algorithm}\label{alg:retrieve}
     
     \SetKwInOut{Input}{Input}
     \SetKwInOut{Output}{Output}
 
     \vskip 0.2cm
     
     \Input{A vertex $u \in V(G)$ and an edge $uv \in E(G)$}
     \Output{A local copy $u_H$ of $u$, where $H \in \Prt$ is such that $uv \in E(H)$}
     
     \vskip 0.1cm

      $x \gets \edges[uv]$

      \While{ $\glob(x) \neq u$}
      {
      $x \gets \mathtt{parent}(x)$\\
      }
     { \Return{ $x$ } } 
     \caption{Method $\retrieve(u,uv)$. }
 \end{algorithm}

\paragraph*{Implementation of the interface.}
We are now ready to implement the public methods of $\Tdep_{d,k}[G,\Prt]$, specified in the lemma statement. Let us emphasize that all public methods operate on global vertices, so there is no access to local copies from the outside. Some methods are very easy to implement, so we describe them briefly first. 
\begin{itemize}[nosep]
 \item Method $\edge(uv)$ can be implemented by a single look-up in $\edges$, so it takes $\Oh(\tdict)$ time.
 \item Method $\bridge(uv)$ is equally simple. Using $\find(uv)$ we get access to the structure $\Dt_{\gl}^{\Pi_k}[T,H]$, and then we verify whether the elimination tree $T$ consists only of local copies of $u$ and $v$. This takes $\Oh(\tdict+d)$ time.
 \item Method $\same(uv,u'v')$ looks up $uv$ and $u'v'$ in $\edges$, thus obtaining $\lowe(uv)$ and $\lowe(u'v')$, which are local copies of either $u$ or $v$, and of either $u'$ or $v'$, respectively. From these local copies we follow parent pointers to the respective roots of elimination trees. If these roots are equal, then the method returns $\true$, otherwise it results $\false$. This takes $\Oh(\tdict + d)$ time.
 \item Method $\new(uv)$ creates a new data structure $\Dt_{\gl}^{\Pi_k}[T,H]$ for $H$ consisting of $u$ and $v$ and a single edge $uv$; this structure uses new copies of $u$ and $v$. Then it adds $uv$ to $\edges$ with the $H$-local copy of $\lowe(uv)$ as its value. This takes $\Oh(\tdict)$ time.
 \item Method $\destroy(uv)$ first uses $\find(uv)$ to get a reference to the structure $\Dt_{\gl}^{\Pi_k}[T,H]$, where $H$ is the part that contains $uv$. Then it checks the assertion that $uv$ is the sole edge of $H$. Finally, the method destroys $\Dt_{\gl}^{\Pi_k}[T,H]$ and removes $uv$ from $\edges$. This takes $\Oh(\tdict+d)$ time.
\end{itemize}
We now move on to describing the more complicated update and query operations.

\paragraph*{Insert and Remove.} The insert procedure $\ins(u,v,ux,vy)$ gets as arguments two vertices $u,v$ and two edges $ux,vy$ such that $ux,vy\in E(H)$ for some $H \in \Prt$. Moreover, we assume that $uv \notin E(H)$, so that the edge $uv$ may be added to $H$. First, we may obtain a reference to the structure $\Dt_{\gl}^{\Pi_k}[T,H]$ using $\find(ux)$, and we can also find the local copies $u_H$ and $v_H$ using $\retrieve(u,ux)$ and $\retrieve(v,vy)$. Now we can call $\Dt_{\gl}^{\Pi_k}[T,H].\ins(u_Hv_H)$ to insert the edge $u_Hv_H$ to $H$. As we argued, this method takes time $2^{\Oh(d^2)}$. 

%This is going to extract the $(d+1)$-core $K$ (which consists of local copies of vertices of $H$), rebuild the elimination tree tree $T^K$ for $H[K]$, and extend $T^K$ with $\appendices$ (see Lemma~\ref{lem:update}).  This procedures need to be adapted slightly to handle global copies of vertices stored in $\SReach(v_H), \Up(v_H)$ and used for indexing the buckets. Note that the partition into buckets is the same using local indices and using global ones. Extracting the core only depends on the bucket information, so it is performed in the same manner as in Lemma~\ref{lem:extract-core}. Now rebulding $T^K$ can be performed on the local copies, and when $T^K$ is constructed, sets $\SReach(v_H)$, $\Up(v_H)$ and the bucket indices may be substituted by global copies, using $\glob(v_H)$ pointers. We now have to attach the buckets in $\appendices$ to $T^K$ as in Lemma~\ref{lem:extend-implement}, however each bucket $\bucket[\bot,X,i]$ is indexed with a set $X$ containing global vertices. Thus, for each of $(1+|K|\cdot 2^{d-1})\cdot d$ buckets $\bucket[\bot,X,i]$, for each element $x \in X$, we iterate through the entire $T^K$ to find the local copy of $x$. Once we have local copies for the entire $X$, we find an attachement local vertex $m$ as in the proof of Lemma~\ref{lem:extend-implement}.  

There is one more clean-up operation that needs to be performed, namely the update of information stored in $\edges$. Recall that the implementation of method $\ins(u_Hv_H)$ in $\Dt_{\gl}^{\Pi_k}[T,H]$ first extracts a $(d+1)$-core $K$, then trims the current state of a data structure to a partial data structure, and then extends the partial data structure to the data structure for updated $H$ using a new elimination forest computed for the subgraph induced by $K$. Consequently, the values $\edges[uv]$ for $u,v\in K$ require updating; it is easy to see that the values for other edges remain valid. As the graph $H[K]$ and its elimination forest $F^K$ are computed explicitly in method $\ins(u_Hv_H)$, finding the new value $\edges[uv]$ for each edge $uv$ with $u,v\in K$ requires $\Oh(|K|^2)$ time. Since there are $\Oh(d|K|)$ edges for which the values need to be updated in this way, this takes time $\Oh(d|K|^3\cdot \tdict)$, which is bounded by $d^{\Oh(d)}\cdot \tdict$ due to $|K|\leq d^{\Oh(d)}$. Hence the total time needed for an edge insertion is $2^{\Oh(d^2)}+d^{\Oh(d)}\tdict$. 

Method $\rem(uv)$ is completely analogous (cf. Lemma~\ref{lem:update}) and requires same asymptotic running~time.

\paragraph*{Merge.} The $\merge(vx,vy)$ procedure (see Algorithm~\ref{alg:merge} for a pseudocode) gets on input two edges $vx,vy \in E(G)$ such that there are different parts $I,J \in \Prt$ where $vx \in E(I)$ and $vy \in E(J)$. Moreover, we assume that $J$ is biconnected. The goal is to merge $I$ and $J$ into one new part of $\Prt$, i.e., substitute $I$ and $J$ in $\Prt$ with $H=G[V(I) \cup V(J)]$. 

We starting getting references to structures $\Dt_{\gl}^{\Pi_k}[T_I,I]$ and $\Dt_{\gl}^{\Pi_k}[T_J,J]$, and retrieving the local copies $v_I$, $x_I$, $v_J$, $y_J$ of the input vertices. This can be done using the $\find()$ and $\retrieve()$ method. Then, we extract $(d+1)$-cores $K_I$ and $K_J$ of $(I,T_I)$ and $(J_,T_J)$, respectively, so that $v_I,x_I\in K_I$ and $v_J,y_J\in K_J$. For this, we apply methods $\Dt_{\gl}^{\Pi_k}[T_I,I].\core(\{ v_I,x_I \},d+1)$ and $\Dt_{\gl}^{\Pi_k}[T_J,J].\core(\{ v_J,y_J \},d+1)$ respectively. We let $K\coloneqq \glob(K_I \cup K_J)$.

We now use methods $\Dt_{\gl}^{\Pi_k}[T_I,I].\trim(K_I)$ and $\Dt_{\gl}^{\Pi_k}[T_J,J].\trim(K_J)$ to modify these data structures into partial data structures $\Dt_{\gl}^{\Pi_k}[T_I-K_I,I/K_I]$ and $\Dt_{\gl}^{\Pi_k}[T_J-K_J,J/K_J]$, respectively (see Lemma~\ref{lem:trim-implement-scheme}). This takes time $(d|K|)^{\Oh(1)}\leq d^{\Oh(d)}$. Now comes the crucial point of the proof: these partial data structures can be easily merged into one partial data structure $\Dt_{\gl}^{\Pi_k}[R,H/K]$, where $R$ is the disjoint union of forests $T_I-K_I$ and $T_J-K_J$. For this, we may just concatenate lists $\Dt_{\gl}^{\Pi_k}[T_I-K_I,I/K_I].\appendices$ and $\Dt_{\gl}^{\Pi_k}[T_J-K_J,J/K_J].\appendices$, and declare the concatenated list to be $\appendices$ in $\Dt_{\gl}^{\Pi_k}[R,H/K]$. Note that thus, the data structure $\Dt_{\gl}^{\Pi_k}[R,H/K]$ uses the local copies inherited from $\Dt_{\gl}^{\Pi_k}[T_I-K_I,I/K_I]$ for the vertices of $V(I)\setminus K_I$, and the local copies inherited from $\Dt_{\gl}^{\Pi_k}[T_J-K_J,J/K_J]$ for vertices of $V(J)\setminus K_J$. 
Since sets $\SReach(\cdot)$, $\Up(\cdot)$, and sets indexing buckets and mugs are stored using global copies of vertices, it is straightforward to verify that the records stored for the vertices of $V(I)\setminus K_I$ and of $V(J)\setminus K_J$ in $\Dt_{\gl}^{\Pi_k}[T_I-K_I,I/K_I]$ and $\Dt_{\gl}^{\Pi_k}[T_J-K_J,J/K_J]$, respectively, are exactly as they should be in $\Dt_{\gl}^{\Pi_k}[R,H/K]$. Thus, merging the data structures indeed only amounts to concatenating the list $\appendices$. Let us note here that the concatenated list $\appendices$ does not contain any duplicate buckets. This is because for each bucket $\bucket[\bot,X,i]\in \Dt_{\gl}^{\Pi_k}[T_J-K_J,J/K_J].\appendices$, the set $X$ contains some vertex of $K_J\setminus \{v_J\}$, due to $J$ being biconnected and $\{v_J,y_J\}\subseteq K_J$.

% We now append $\appendices_J$ to $\appendices_I$, thus creating a concatenated list $\appendices$. We now observe that $\appendices$ is a partial data structure that can be attached to $T^K$ (see Section~\ref{sec:data-structure} for the definition of a partial data structure). To be more precise, $\appendices$ consists of a set $R$ of elimination forests (represented by their roots), grouped into buckets $\bucket[\bot,X,i]$, and we claim that it is in fact a partial data structure $\Dt[R,H/\glob(K)]$ for $H=G[V(I) \cup V(J)]$ (according to the definition in Section~\ref{sec:data-structure}). Observe that each bucket in $\appendices$ is either a bucket of $\appendices_I$ or $\appendices_J$. Hence, every bucket contains trees, where the vertices are already equipped with global sets $\SReach$, $\Up$ and with the bucket partitioning indexed with global sets. This information remains up to date thanks to global indexing. What remains to observe is that $\appendices$ does not contain duplicate buckets. 

%Since the vertex sets of $I$ and $J$ are disjoint except of vertex $v$, the only way this could happen is that there are two buckets $\bucket[\bot,\{ v \},i]$, one coming from $\appendices_I$ and the other from $\appendices_J$. However, there cannot be a bucket $\bucket[\bot,\{ v \}, i]$ in $\appendices_J$, because $v$ is not an articulation point of $J$, as, due to our assertions, $J$ is biconnected. So there we have a partial data structure $\Dt_{\gl}^{\Pi_k}[R,H/\glob(K)]$ given by list $\appendices$, 

Next, we construct a copy $H_K$ of $G[K]$ using $H_K\coloneqq \cpy(K)$. We may also compute a recursively optimal elimination tree $T^K$ for $H_K$ using the algorithm of Lemma~\ref{lem:static}. We refer to the corresponding procedure as $\elimination(H_K)$, and we note that it takes time $2^{\Oh(d^2)}\cdot |K|^{\Oh(1)}=2^{\Oh(d^2)}$. In case it turns out that $\td(H_K)>d$, then we roll back all the changes and terminate the method due to the conclusion that $\td(H)>d$; otherwise, we may assume that $\height(T^K)\leq d$.
Now, from Lemma~\ref{lem:merging} we conclude that $R$ is attachable to $(H_K,T^K)$ and if $\wh{T}$ is the extension of $T^K$ via $R$, then $\wh{T}$ is a recursively optimal elimination forest of $H$ of height at most $d$.

We may finally apply the method $\extend(H_K,T^K)$, provided by Lemma~\ref{lem:extend-implement-scheme}, to modify the data structure $\Dt_{\gl}^{\Pi_k}[R,H/K]$ into $\Dt_{\gl}^{\Pi_k}[\wh{T},H]$. By Lemma~\ref{lem:extend-implement-scheme} and the discussion about the treatment of global indexing, this takes time polynomial in $2^d$ and $|K|$, which is $d^{\Oh(d)}$.

At the end we need to update the values stored in $\edges$, similarly as in the $\ins()$ method. Again, this takes time $(d|K|)^{\Oh(1)}\cdot \tdict\leq d^{\Oh(d)}\cdot \tdict$. So all in all, the whole implementation takes time $2^{\Oh(d^2)}+d^{\Oh(d)}\cdot \tdict$.
% 
% At this point, having $\appendices$, $T^K$ and $H_K$ at our disposal, we apply (simulate) the procedure (see Lemma~\ref{lem:extend-implement} and~\ref{lem:extend-implement-scheme} for specification), which extends $\Dt_{\gl}^{\Pi_k}[R,H/\glob(K)]$ to $\Dt_{\gl}^{\Pi_k}[T_H,H]$, which can be implemented using global indices as described for $\ins()$ procedure. We refer to this procedure as $\extend(\appendices,K,H_K)$.
% Due to Lemma~\ref{lem:merging}, the elimination tree $T_H$ of $\Dt_{\gl}^{\Pi_k}[T_H,H]$ is recursively optimal. The running time of the entire procedure up until now is dominated by $2^{\Oh(d^2)}$ factor for the static elimination tree computation, emerging from Lemma~\ref{lem:static}.  
% Clearly, similarily as for $\ins()$, we need to update the edges with both endpoints in $\glob(K)$ in $\edges$ dictionary. This, again, takes $d^{\Oh(d)}\tdict$ time.

\paragraph*{Split.} The method $\splitt(vx,vy)$ (see Algorithm~\ref{alg:split} for pseudocode) takes as arguments edges $vx,vy \in E(G)$ such that there is a part $H \in \Prt$ such that $vx,vy \in E(H)$, $v$ is a cut-vertex in $H$, and $H$ has exactly two biconnected components $I$ and $J$, where $vx\in E(I)$ and $vy\in E(J)$. The goal is to split $H$ in $\Prt$ into $I$ and $J$. 

We first get access to the data structure $\Dt_{\gl}^{\Pi_k}[T,H]$ and retrieves the local copies $v_H,x_H,y_H$ of $v,x,y$, respectively. For this, we may use methods $\retrieve()$ and $\find()$. We then extract a $(d+1)$-core $K$ satisfying $v_H,x_H,y_H\in K$ using the method $\Dt_{\gl}^{\Pi_k}[T_H,H].\core(\{ v_H, x_H, y_H \},d+2)$ method. Next, we extract the graph $H_K\coloneqq \cpy(K)$. By Lemma~\ref{lem:biconn}, $H_K$ has exactly two biconnected components $H_I\coloneqq H[K^I]$ and $H_J\coloneqq H[K^J]$, where $K^I\coloneqq K\cap V(I)$ and $K^J\coloneqq K\cap V(J)$. We can construct these components in time $\Oh(|E(H_K)|)=d^{\Oh(d)}$ by applying on $H_K$ any static linear-time algorithm for decomposing a graph into biconnected components. Note that $V(H_I)\cap V(H_J)=\{v\}$.

We apply the method $\Dt_{\gl}^{\Pi_k}[T,H].\trim(K)$, provided by Lemma~\ref{lem:trim-implement-scheme}, to modify the data structure $\Dt_{\gl}^{\Pi_k}[T,H]$ into the partial data structure $\Dt_{\gl}^{\Pi_k}[R,H/K]$, where $R=T-K$. Note that by Lemma~\ref{lem:trim-implement-scheme}, the list $\appendices$ in the data structure $\Dt_{\gl}^{\Pi_k}[R,H/K]$ has length at most $\Oh(2^{d}\cdot |K|)\leq d^{\Oh(d)}$. 
By Lemma~\ref{lem:splitting}, $R$ can be decomposed into two subforests $R^I$ and $R^J$ by placing each tree of $R$ either in $R^I$ or in $R^J$, where trees whose vertex sets are contained in $V(I)$ go to $R^I$ and trees whose vertex sets are contained in $V(J)$ go to $R^J$. Note that since $K$ contains at least one vertex of $V(I)\setminus \{v\}$ and at least one vertex of $V(J)\setminus \{v\}$, while both $I$ and $J$ are biconnected, for each tree $S$ in $R$ we have that $|N_H(V(S))|>1$. This means that $N_H(V(S))$ either contains vertices of $V(I)\setminus \{v\}$ and then $V(S)\subseteq V(I)$, or $N_H(V(S))$ contains vertices of $V(J)\setminus \{v\}$ and then $V(S)\subseteq V(J)$. Observing this, we may split the partial data structure $\Dt_{\gl}^{\Pi_k}[R,H/K]$ into $\Dt_{\gl}^{\Pi_k}[R^I,I/K^I]$ and $\Dt_{\gl}^{\Pi_k}[R^I,I/K^I]$ as follows: for each bucket $\bucket[\bot,X,i]\in \Dt_{\gl}^{\Pi_k}[R,H/K].\appendices$, place this bucket on $\Dt_{\gl}^{\Pi_k}[R^I,I/K_I].\appendices$ if $X$ contains a vertex of $V(I)\setminus \{v\}$, and otherwise place this bucket on $\Dt_{\gl}^{\Pi_k}[R^J,J/K_J].\appendices$. Note here that since $X\subseteq K$ each such check can be done in time $\Oh(|X|\cdot |K|)=d^{\Oh(d)}$, so the total time required for splitting is bounded by $d^{\Oh(d)}$. Once the lists $\Dt_{\gl}^{\Pi_k}[R^I,I/K_I].\appendices$ and $\Dt_{\gl}^{\Pi_k}[R^J,J/K_J].\appendices$ are constructed, we observe that the data structures $\Dt_{\gl}^{\Pi_k}[R^I,I/K_I].\appendices$ and $\Dt_{\gl}^{\Pi_k}[R^J,J/K_J].\appendices$ are effectively constructed as well: the records stored at vertices of $V(H)\setminus K$, inherited from $\Dt_{\gl}^{\Pi_k}[R,H/K]$, are exactly as they should be in $\Dt_{\gl}^{\Pi_k}[R^I,I/K_I]$ and $\Dt_{\gl}^{\Pi_k}[R^J,J/K_J]$. Note that both these data structures inherit local copies of vertices from $\Dt_{\gl}^{\Pi_k}[R,H/K]$.

Next, we create copies $H'_I\coloneqq \cpy(K^I)$ and $H'_J\coloneqq \cpy(K^J)$ of $H_I$ and $H_J$, respectively. We compute recursively optimal elimination trees $T^{K^I}$ and $T^{K^J}$ of $H'_I$ and $H'_J$, respectively, using the $\elimination()$ method that applies the algorithm of Lemma~\ref{lem:static}. Note that this step takes time $2^{\Oh(d^2)}\cdot |K|^{\Oh(1)}=2^{\Oh(d^2)}$ and, due to the assumption that $\td(H)\leq d$, the obtained trees $T^{K^I}$ and $T^{K^J}$ have height at most $d$.
By Lemma~\ref{lem:splitting}, $R^I$ is attachable to $(H'_I,T^{K^I})$ and $\wh{T}^I$ --- the extension of $T^{K^I}$ via $R^I$ --- is a recursively optimal elimination forest of $H'_I$. Similarly for $R^J$, $H'_J$, and $T^{K^J}$. 

As we have already constructed data structures $\Dt_{\gl}^{\Pi_k}[R^I,I/K_I]$ and $\Dt_{\gl}^{\Pi_k}[R^J,J/K_J]$, we may now use their $\extend()$ methods, provided by Lemma~\ref{lem:extend-implement-scheme}, to turn them into data structures $\Dt_{\gl}^{\Pi_k}[\wh{T}^I,I]$ and $\Dt_{\gl}^{\Pi_k}[\wh{T}^J,J]$. Again, this takes time $d^{\Oh(d)}$. Finally, we again need to update the values stored in $\edges$ for edges with both endpoints in $K$, which again takes time $d^{\Oh(d)}\cdot \tdict$. We conclude that the total running time of the method is $2^{\Oh(d^2)}+d^{\Oh(d)}\cdot \tdict$.

 \begin{algorithm}\label{alg:merge}
     
     \SetKwInOut{Input}{Input}
     \SetKwInOut{Output}{Output}
    \SetKw{Raise}{raise exception:\ }
 
     \vskip 0.2cm
     
     \Input{Edges $vx,vy \in V(G)$ such that there are different $I,J \in \Prt$ where $vx \in E(I)$ and $vy \in E(J)$, and $J$ is biconnected}
     \Output{$I$ and $J$ are substituted with $H=G[V(I)\cup V(J)]$ in $\Prt$}
     
     \vskip 0.1cm

      $\Dt_{\gl}^{\Pi_k}[T_I,I], \Dt_{\gl}^{\Pi_k}[T_J,J] \gets \find(vx), \find(vy)$\\
      $v_I,v_J, x_I, y_J \gets \retrieve(v,vx), \retrieve(v,vy),\retrieve(x,vx),\retrieve(y,vy)$\\
      $K_I,K_J \gets \Dt_{\gl}^{\Pi_k}[T_I,I].\core(\{ v_I,x_I \},d+1),\Dt_{\gl}^{\Pi_k}[T_J,J].\core(\{ v_J,y_J \},d+1)$\\
      $K \gets \glob(K_I \cup K_J)$
      
      turn $\Dt_{\gl}^{\Pi_k}[T_I,I]$ into $\Dt_{\gl}^{\Pi_k}[T_I-K_I,I/K_I]$ by applying $\Dt_{\gl}^{\Pi_k}[T_I,I].\trim(K_I)$
      
      turn $\Dt_{\gl}^{\Pi_k}[T_J,J]$ into $\Dt_{\gl}^{\Pi_k}[T_J-K_J,J/K_J]$ by applying $\Dt_{\gl}^{\Pi_k}[T_J,J].\trim(K_J)$
      
      merge $\Dt_{\gl}^{\Pi_k}[T_I-K_I,I/K_I]$ and $\Dt_{\gl}^{\Pi_k}[T_J-K_J,J/K_J]$ into $\Dt_{\gl}^{\Pi_k}[R,H/K]$ by concatenating their $\appendices$ lists
      
      $H_K\gets \cpy(K)$\\
      $T^K\gets \elimination(H_K)$\\
      
      \If{$\elimination(H_K)$\normalfont{reported that}$\td(H_K)>d$}{
         roll back all the changes\\
         \Raise $\td(H)>d$
      }
      
      turn $\Dt_{\gl}^{\Pi_k}[R,H/K]$ into $\Dt_{\gl}^{\Pi_k}[\wh{T},H]$ using $\extend(H_K,T^K)$
      
      $\update(\edges)$\\
     %{ \Return{ $x$ } } 
     \caption{Method $\merge(vx,vy)$ }
 \end{algorithm}

  \begin{algorithm}\label{alg:split}
     
     \SetKwInOut{Input}{Input}
     \SetKwInOut{Output}{Output}
 
     \vskip 0.2cm
     
     \Input{edges $vx,vy \in V(G)$ such that there is $H \in \Prt$ such that $vx,vy \in E(H)$, $v$ is a cut-vertex in $H$, and $H$ has exactly two biconnected components $I$ and $J$ such that $vx\in E(I)$ and $vy\in E(J)$}
     \Output{Substitutes $H$ with $I$ and $J$ in $\Prt$}
     
     \vskip 0.1cm

     $v_H,x_H, y_H \gets \retrieve(v,vx),\retrieve(x,vx),\retrieve(y,vy)$\\
     
     $\Dt_{\gl}^{\Pi_k}[T,H] \gets \find(vx)$ \\
     
     $K \gets \Dt_{\gl}^{\Pi_k}[T_H,H].\core(\{ v_H, x_H, y_H \},d+2)$\\
     $K^I,K^J\gets V(I)\cap K,V(J)\cap K$\\
     
     $H_K \gets \cpy(K)$\\
     $H_I,H_J \gets \bicom(H_K)$ \\

     turn $\Dt_{\gl}^{\Pi_k}[T,H]$ into $\Dt_{\gl}^{\Pi_k}[T-K,H/K]$ using $\Dt_{\gl}^{\Pi_k}[T,H].\trim(K)$\\
     split $\Dt_{\gl}^{\Pi_k}[T-K,H/K]$ into $\Dt^{\Pi_k}_{\gl}[R^I,I/K^I]$ and $\Dt^{\Pi_k}_{\gl}[R^J,J/K^J]$ by splitting the $\appendices$ list\\
     
     $H'_I,H'_J \gets \cpy(K^I),\cpy(K^J)$\\
     $T^{K^I},T^{K^J} \gets \elimination(H'_I), \elimination(H'_J)$\\
     turn $\Dt^{\Pi_k}_{\gl}[R^I,I/K^I]$ into $\Dt^{\Pi_k}_{\gl}[\wh{T}^I,I]$ using $\extend(H'_I,T^{K^I})$\\
     turn $\Dt^{\Pi_k}_{\gl}[R^J,J/K^J]$ into $\Dt^{\Pi_k}_{\gl}[\wh{T}^J,J]$ using $\extend(H'_J,T^{K^J})$
     
     \caption{Method $\splitt(vx,vy)$}
 \end{algorithm}

\paragraph{Path queries.} Recall that the path queries to $\Tdep_{d,k}[G,\Prt]$ get on input two vertices $u,v$ and two edges $ux,vy$ such that there exists a part $H\in \Prt$ satisfying $ux,vy\in E(H)$. The query $\pathlb(u,v,ux,vy)$ should answer whether in $H$ there is a simple path on at least $k$ vertices with endpoints $u$ and $v$. In the query $\pathlb(i,u,v,ux,vy)$ we similarly ask for a $u$-to-$v$ path on exactly $i$ vertices in $H$, where $i\leq k$ is a given parameter.

These queries are implemented as follows.
Using method $\find(ux)$ we may get access to the structure $\Dt^{\Pi_k}[T,H]$ stored within $\Tdep_{d,k}[G,\Prt]$, where $ux,vy\in E(H)$. Note that $u,v$ are provided on input using global copies, but we may compute the corresponding local copies $u_H,v_H$ using $\retrieve(u,ux)$ and $\retrieve(v,vy)$. Now, to answer the queries we may use the procedures provided by Lemma~\ref{lem:path_queries} for the data structure $\Dt_{\gl}^{\Pi_k}[T,H]$. Again, the need of translation of global copies to local ones incurs an additional $d^{\Oh(1)}$ overhead in the running times of these procedures, but anyway they work in time $2^{\Oh(k)}(dk)^{\Oh(d)}\leq d^{\Oh(d)}$. Hence, the total running time of the queries is bounded by $d^{\Oh(d)}+\Oh(\tdict)$, as required.

% Using method
% 
% As discussed, we can maintain the structure $\Dt_{\gl}^{\Pi_k}[T_H,H]$ for each $H \in \Prt$ upon all updates described above. Due to Lemma~\ref{lem:path_queries}, given access to $\Dt^{\Pi_k}[T_H,H]$, and two vetrices $v_H,w_H \in V(T_H)$, we can answer queries:
% \begin{itemize}
%  \item is there a path from $v_H$ to $w_H$ on at least $k$ vertices
%  \item return a path from $v_H$ to $w_H$ on $i$ vertices if one exists.
% \end{itemize}
% We do not have access to $\Dt^{\Pi_k}[T_H,H]$, but we have access to $\Dt_{\gl}^{\Pi_k}[T_H,H]$. Since the two queries do not modify $T_H$, they can be simulated by $\Dt_{\gl}^{\Pi_k}[T_H,H]$, given the two local copies $v_H$ and $w_H$ of vertices in $V(T_H)$. Thus, both methods $\pathlb(u,v,ux,vy)$ and $\pathub(i,u,v,ux,vy)$ in addition to global vertices $u$ and $v$ receive (global) edges adjacent to them in $H$. They use method $\retrieve()$ to retrieve the local copies, and next the queries provided by Lemma~\ref{lem:path_queries} are simulated. Due to Lemma~\ref{lem:path_queries}, the time needed for both methods is $2^{\Oh(k)}(dk)^{\Oh(d)}+\tdict$, where $\tdict$ comes from querying for local copies and finding the appropriate structures.

\paragraph{Cut-vertices.} Finally, method $\art(vx,vy)$ (see Algorithm~\ref{alg:articul} for pseudocode) takes two edges $vx,vy \in E(G)$ such that there exists a part $H\in \Prt$ satisfying $vx,vy\in E(H)$, and should answer whether $x$ and $y$ are in the same connected component of $H-v$.

First, we get access to the structure $\Dt_{\gl}^{\Pi_k}[T,H]$ stored within $\Tdep_{d,k}[G,\Prt]$ using $\find(vx)$, and we find the relevant local copies $v_H,x_H,y_H$ using $\retrieve()$. Next, we extract a $2$-core $K$ of $(H,T)$ satisfying $v_H,x_H,y_H\in K$. For this we use method $\Dt_{\gl}^{\Pi_k}[T,H].\core(\{v_H,x_H,y_H\},2)$, which runs in time $d^{\Oh(d)}$ and also constructs the graph $H_K\coloneqq H[K]$. Finally, by Lemma~\ref{lem:biconn} we infer that $x$ and $y$ are in the same connected component of $H-v$ if and only if they are in the same connected component of $H_K-v$. Since $|K|\leq d^{\Oh(d)}$, this can be easily checked in time $d^{\Oh(d)}$.

% The method simply extracts a $2$-core $K$ from $\Dt_{\gl}^{\Pi_k}[T,H]$, where $K$ contains $v_H$, $x_H$ and $y_H$ and checks if $v$ separates $x$ from $y$ in graph $H[\glob(K)]$. The correctnes of this method is guaranteed by Lemma~\ref{lem:biconn}. The running time is dominated by $2$-core computation, which takes $d^{\Oh(d)}$ time according to Lemma~\ref{lem:extract-core}. 
  \begin{algorithm}\label{alg:articul}
     
     \SetKwInOut{Input}{Input}
     \SetKwInOut{Output}{Output}
 
     \vskip 0.2cm
     
     \Input{edges $vx,vy \in E(G)$ such that there is $H \in \Prt$ satisfying $vx,vy\in E(H)$}
     \Output{$\true$ if $x$ and $y$ are in the same connected component of $H-v$, and $\false$ otherwise}
     
     \vskip 0.1cm

     $v_H,x_H, y_H \gets \retrieve(v,vx),\retrieve(x,vx),\retrieve(y,vy)$\\
     
     $\Dt_{\gl}^{\Pi_k}[T,H] \gets \find(vx)$ \\
     
     $H_K,K \gets \Dt_{\gl}^{\Pi_k}[T,H].\core(\{ v_H, x_H, y_H \},2)$\\
     
     $a \gets $ are $x_H$ and $y_H$ in the same connected component of $H_K-v_H$?\\
     
     { \Return{ $a$ } } 
     \caption{Method $\art(vx,vy)$}
 \end{algorithm}
\end{proof}

%%% Local Variables:
%%% mode: latex
%%% TeX-master: "main"
%%% End:

\section{The \LCycle-cycle detection data structure}\label{sec:cyc_det_ds}

\CycleNote

Our goal in this section is to derive an analogue of Lemma~\ref{lem:kpath-bnd-td} for the $\geq k${\sc{-Cycle}} problem. That is, we are going to design a data structure that for a dynamic graph $G$, modified over time by edge insertions and removals, recognizes whether $G$ contains a cycle on at least $k$ vertices. More precisely, the data structure works under the invariant that this is never the case, and refuses any edge insertion that breaks this~invariant.

The main idea is to combine the data structure presented in the previous section with the top-tree data structure~\cite{AlstrupIP,AlstrupJ}, which we recall next.

\paragraph{Top trees.} The top trees data structure~\cite{AlstrupIP,AlstrupJ} maintains a dynamic forest $\Upsilon$ subject to edge insertions and deletions. Contrary to forests we maintained so far, $\Upsilon$ is not rooted, that is, it is a union of unrooted trees. 

The top tree data structure is designed to aggregate information along paths in the maintained forest $\Upsilon$. The idea is that each tree $S$ in forest $\Upsilon$ is assigned a pair $\partial S$ of (not necessarily distinct) vertices in $S$, called the {\em{boundary}} of $S$. The purpose of distinguishing the boundary vertices is to be able to easily access the information associated with the unique path between them within $S$. This information has to be aggregated when the forest $\Upsilon$ changes, and to do this efficiently, the path information is distributed in a binary tree called the \emph{top tree}. We outline next the most important properties of top trees, but for details we refer to~\cite{AlstrupIP,AlstrupJ}. 

So as mentioned before, for each pair $(S,\partial S)$ a (rooted and binary) top tree $\Lambda(S,\partial S)$ is constructed and maintained. The top tree construction is based on a recursive partition of $S$ into clusters, each cluster $C$ having its own boundary pair of vertices $\partial C$. %
For a connected subtree $C$ of $S$, a boundary vertex is a vertex of $C$ that is either a vertex of $\partial S$, or a neighbor of a vertex not in $C$. We say that $C$ is a {\em{cluster}} of $S$ if it has at most two boundary vertices $x,y$; then $\partial C = \{x,y\}$. Moreover, if $x \neq y$, we call $C$ a path cluster. The recursive partition into clusters starts with the root cluster, which is $S$ itself. A cluster $C$ is partitioned into clusters $A$ and $B$, where $A$ and $B$ share one boundary vertex, and $(\partial A \cup \partial B) \setminus (\partial A \cap \partial B) \subseteq \partial C \subseteq (\partial A \cup \partial B)$. % $x$, i.e., $\partial A \cap \partial B = \{ x \}$, such that $x \notin \partial C$, while other boundary vertices of $A$ and $B$ are the vertices of $\partial C$, i.e., $\partial C = \partial A \cup \partial B \setminus \{ x \}$. 
In the top tree $\Lambda(S,\partial S)$ cluster $C$ becomes a parent of $A$ and $B$. The leafs of $\Lambda(S,\partial S)$ are  clusters containing single edge, one per each edge of $S$. 

The partition procedure (which we do not describe here, see~\cite{AlstrupJ}) guarantees that the height of each top tree $\Lambda(S,\partial S)$ is $\Oh(\log |S|)$. 
The top trees data structure provides the following methods for updating the forest $\Upsilon$:
\begin{itemize}[nosep]
 \item $\link(u,v)$: Given vertices $u,v$ lying in different trees of $\Upsilon$, links these trees by adding the edge $uv$ to $\Upsilon$.
 \item $\cut(e)$: Remove the edge $e$ from $\Upsilon$.
 \item $\expose(u,v)$: Return $\nil$ if $u$ and $v$ are not in the same tree of $\Upsilon$. Otherwise, if $S$ is the tree that contains $u$ and $v$, then the top tree $\Lambda(S,\partial S)$ is rebuilt so that $u$ and $v$ become the boundary vertices of its root cluster. This new root cluster is returned by the method.
\end{itemize} 

As presented in~\cite{AlstrupIP,AlstrupJ}, the top tree data structure can be enriched by additionally storing some piece of information about each cluster. Here, one assumes that these pieces of information are compositional, in the sense that the information for child clusters can be aggregated to the information about the parent cluster. Actually, the example provided by Alstrup et al. in~\cite{AlstrupJ} in Lemma~$5$ is exactly the one we need here: for a path cluster $C$ one can store the length of the path within $C$ between the boundary vertices of $C$. This leads to an implementation of the following query with worst-case running time $\Oh(\log n)$:
\begin{itemize}[nosep]
 \item $\pathlf(u,v)$: Given vertices $u$ and $v$, return the length of the $u$-to-$v$ path in $\Upsilon$, or $\infty$ if $u$ and $v$ are in different trees of $\Upsilon$.
\end{itemize}
In our application, we will also need to efficiently report $u$-to-$v$ paths assuming they are short. The following lemma summarizes all that we need from the top tree data structure.

\begin{lemma}\label{lem:TT}
  There is a top trees data structure $\Topt[\Upsilon]$ which maintains a dynamic forest $\Upsilon$ on $n$ vertices subject to edge additions and removals. The data structure provides the following methods:
\begin{itemize}[nosep]
 \item $\link(u,v)$: Given vertices $u$ and $v$, if $u$ and $v$ lie in the same tree of $\Upsilon$, do nothing and return $\false$. Otherwise, add the edge $uv$ to $\Upsilon$ and return $\true$.
 \item $\cut(e)$: Given an edge $e$ of $\Upsilon$, remove $e$ from $\Upsilon$.
 \item $\pathlf(u,v)$: Given vertices $u$ and $v$, return the length of the $u$-to-$v$ path in $\Upsilon$, or $\infty$ if $u$ and $v$ are in different trees of $\Upsilon$.
 \item $\pathf(u,v)$: Given vertices $u$ and $v$, return the $u$-to-$v$ path in $\Upsilon$, or $\nil$ if $u$ and $v$ are in different trees in $\Upsilon$.
\end{itemize}
All the methods run in worst case $\Oh(\log n)$ time with the exception of $\pathf(u,v)$, which runs in worst case time $\Oh(|\pi| \cdot \log n)$, where $\pi$ is the returned path. The data structure uses $\Oh(|\Upsilon|)$ space.
\end{lemma}
\begin{proof}
 The data structure is the variant of the top tree data structure described in~\cite{AlstrupIP,AlstrupJ} that in addition to the basic data, also stores, for each cluster $C$, the length of the path within $C$ between the vertices of the boundary of $C$. As shown in~\cite{AlstrupIP,AlstrupJ}, the first three operations can be implemented in this data structure with worst-case running time $\Oh(\log n)$. Hence, it remains to implement the query $\pathf(u,v)$.
 
 %For $\pathlf(u,v)$, we first apply $\expose(u,v)$. This method either reports that $u$ and $v$ are in different trees of $\Upsilon$ --- in which case we may return $\infty$ as the answer to the query --- or returns the cluster associated with the tree $S$ containing $u$ and $v$, after setting $\partial S$ be consist of $u$ and $v$. Then we can simply return $\len(S,\partial S)$ as the answer to the query.
 
 For $\pathf(u,v)$, we first apply $\expose(u,v)$.  This method either reports that $u$ and $v$ are in different trees of $\Upsilon$ --- in which case we may return $\nil$ as the answer to the query --- or returns the cluster associated with the tree $S$ containing $u$ and $v$, after setting $\partial S$ to consist of $u$ and $v$. Then, we construct the $u$-to-$v$ path in $\Upsilon$ using a recursive procedure $\pathf(C)$ that for a path cluster $C$, given as a node in $\Topt[\Upsilon]$ and where $\partial C$ consists of two different vertices, 
 returns the path $\pi(C)$ connecting the vertices of $\partial C$ within $C$. This procedure works as follows. If $C$ is a leaf cluster, then return the single edge of $C$. Otherwise, if $A$ and $B$ are the child clusters of $C$, then return the concatenation of $\pi(A)$, computed using a recursive call $\pathf(A)$, and $\pi(B)$, computed using a recursive call $\pathf(B)$. Here, if either $\partial A$ or $\partial B$ consist of only one vertex, then ignore the corresponding recursive call.
 
 Observe, that we only call the recursion on path clusters. Now, for the analysis, we distinguish two types of recursive calls: the first type recurses on both children or is called on a path cluster consisting of a single edge, and the second type recurses on one child only. Observe that the total number of calls of the first type is bounded by $\Oh(|\pi(S)|)$, as each such call can be attributed either to a distinct edge of $\pi(S)$, or to splitting $\pi(S)$ into two nonempty paths. Furthermore, after a sequence of $\Oh(\log |S|)$ calls of the second type there must be a call of the first type.
 Thus, the procedure $\pathf(S)$ constructs the sought path $\pi(S)$ in time $\Oh(|\pi(S)| \cdot \log n)$.
\end{proof}

\paragraph*{The data structure.}
Equipped with the top trees data structure we are ready to prove the main result of this section, stated in the following lemma.

\begin{lemma}\label{lem:kcycle-restricted}
Let $k$ be a fixed parameter.
Suppose we are given a dynamic graph $G$, updated by edge insertions and removals, and access to a dictionary on the edges of $G$, which uses $\mdict$ memory and each dictionary operation runs in time bounded by $\tdict$. 
Then there is a dynamic data structure that maintains $G$ under such updates as long as the following invariant holds: $G$ does not contain a simple cycle on at least $k$ vertices. If an edge insertion violates this invariant, the data structure refuses to carry out the insertion and reports this outcome. The data structure uses $(n \cdot 2^{\Oh(k^2 \log k)} + \mdict)$ memory and guarantees worst-case running time of $\Oh(k \log n)+2^{\Oh(k^2 \log k)} \tdict + 2^{\Oh(k^4)}$ per operation. The initialization for an edgeless graph on $n$ vertices takes time $\Oh(n)$.
\end{lemma}
\begin{proof}
We may assume that $\tdict\leq \Oh(\log n)$.
The data structure maintains a top trees data structure $\Topt[\Upsilon]$, provided by Lemma~\ref{lem:TT}, in which we store a maximal spanning forest $\Upsilon$ of $G$. We note that the $\Topt[\Upsilon]$ structure is completely oblivious of the edges of $E(G)\setminus E(\Upsilon)$, so we will need to ensure that $\Upsilon$ remains maximal while implementing edge insertions and deletions in $G$. 

Moreover, we maintain the structure $\Tdep_{k^2,k}[G,\Prt]$ provided by Lemma~\ref{lem:kpath-bnd-td-ms}, i.e., the structure $\Tdep_{d,k}[G,\Prt]$ for $d=k^2$. This structure maintains a nice partition $\Prt$ of $G$ that will at all times coincide with the partition of $G$ into biconnected components, except that methods that we use for edge insertion and removal may temporarily change $\Prt$ into a coarser nice partition, but by the time a method terminates, $\Prt$ becomes the partition into biconnected components again. As we assume the invariant that $G$ never contains a cycle of length at least $k$, by Lemma~\ref{lem:cycle-td} all the biconnected components of $G$ always have treedepth bounded by $k^2$. The memory usage is dominated by memory used by $\Tdep_{k^2,k}[G,\Prt]$ structure, which is $(n \cdot 2^{\Oh(k^2 \log k)} + \mdict)$ due to Lemma~\ref{lem:kpath-bnd-td-ms}.

%We store all structures $\Tdep[H,k^2,k]$ (pointers to them), one per each biconnected component, in a doubly linked list $\bicom$, where they can be added and removed in constant time.
%In addition to that, we maintain a dictionary $\edges$ data structure, where the keys are (global) edges of $G$. Each edge $uv$ in $G$ belongs to precisely one biconnected component $H(uv)$, so the value of $uv$ in $\edges$ is a pair $(u_{H(uv)},v_{H(uv)})$ of local copies of its endpoints $u$ and $v$ in $H(uv)$. Clearly, $\edges$ is modified as the edges appear and dissapear from $G$. This completes the list of structures maintained by the cycle detection data structure. Additionally, the cycle detection data structure implements a method $\find(u_H)$, which, given a local copy $u_H$ of some vertex $u$, returns an instance $\Tdep[H,k^2,k]$ corresponding to $H$. This method is possible to implement in $\Oh(k^2)$ time. Since $H$ is connected, the structure $\Tdep[H,k^2,k]$ maintains a single elimination tree, and we can use its parent pointers to iterate all the way from $u_H$ to the root of the elimination tree. The pointer to the instance $\Tdep[H,k^2,k]$ can be stored in the root and retrieved. Since the depth of the elimination tree is bounded by $k^2$, we get $\Oh(k^2)$ time for iterating from $u_H$ to the root.

We now move on to describing the procedures for inserting and removing edges.

\paragraph{Insert.} We start with procedure $\ins(uv)$ for inserting an edge (see Algorithm~\ref{alg:addcd} for a pseudocode). We start by verifying that $uv$ is not yet present in $G$, as otherwise there is nothing to do. Then we apply the query $\Topt[\Upsilon].\pathlf(u,v)$ to check the distance between $u$ and $v$ in $\Upsilon$ (see Lemma~\ref{lem:TT} for the interface of $\Topt[\Upsilon]$). We distinguish three cases.

The first case is when $\Topt[\Upsilon].\pathlf(u,v)$ returns $\infty$, meaning that $u$ and $v$ are not in the same connected component of $G$. Then, the newly added edge $uv$ becomes a bridge and hence it forms a new biconnected component of $G$. We use the method $\Tdep_{k^2,d}[G,\Prt].\new(uv)$ method to add up a new biconnected component $H=(\{ u,v \}, \{ uv \})$ to $\Prt$. We also update $\Upsilon$ by calling $\Topt[\Upsilon].\link(u,v)$. This case is handled in $\Oh(\log n)$ time, as the dominating operations are $\Topt[\Upsilon].\pathlf(u,v)$ and $\Topt[\Upsilon].\link(u,v)$. 

The second case is when $\Topt[\Upsilon].\pathlf(u,v)$ returns a value that is larger or equal to $k$. Clearly, adding $uv$ in this case violates the invariant, as $uv$ closes the cycle of length at least $k$. Hence, if this is the case, the method terminates by throwing a suitable exception, and in particular the edge insertion is not carried out. Similarly as before, handling this case takes $\Oh(\log n)$ time.

We are left with the case when $\Topt[\Upsilon].\pathlf(u,v)$ returns a value smaller than $k$. In such case, we retrieve the entire path $\pi$ in $\Upsilon$ between $u$ and $v$ using $\Topt[\Upsilon].\pathf(u,v)$; this takes time $\Oh(k\log n)$. We use this path to merge the biconnected components along $\pi$. More precisely, observe that $\pi$ can be partitioned into subpaths $\pi_1,\ldots,\pi_\ell$ so that $\pi_1$ is contained in a biconnected component $H_1$, $\pi_2$ is contained in another biconnected component $H_2$, and so on. What should happen after inserting the edge $uv$ to $G$, is that all these biconnected component $H_1,\ldots,H_\ell$ should be merged into one part of $\Prt$, to which $uv$ should be inserted --- this new, merged part $H$ is a new biconnected component after the insertion. To simulate this, we iterate through the edges on $\pi$ and if for two consecutive edges $xy$ and $yz$ we find out that they belong to different parts in $\Prt$ --- which can be checked using method $\same(xy,yz)$ --- then we merge those parts using method $\merge(xy,yz)$. Note that at each point, method $\merge$ may report that the treedepth of the obtained graph is larger than $d=k^2$, in which case we roll back the changes and report  violation of the invariant. This is correct by Lemma~\ref{lem:cycle-td}, as we conclude that after adding $uv$, the biconnected component of $G$ containing $uv$ would have treedepth larger than $k^2$. 

Finally, all the merges have been executed, and hence all the edges of $\pi$ are now in the same part of $\Prt$. We now verify whether adding $uv$ to $G$ does not break the invariant using query $\Tdep_{k^2,k}[G,\Prt].\pathlb(u,v,e_u,e_v)$, where $e_u$ and $e_v$ are respectively the first and the last edge of $\pi$. Recall that this query checks whether the part of $\Prt$ containing $e_u$ and $e_v$ contains a simple $u$-to-$v$ path on at least $k$ vertices, which at this point is equivalent to checking whether adding $uv$ to $G$ would create a cycle on at least $k$ vertices. If no violation of the invariant is discovered, we may safely add the edge $uv$ using method $\Tdep_{k^2,k}[G,\Prt].\ins(u,v,e_u,e_v)$.

As for the running time, observe that $\ell<k$, hence we use methods $\same$ and $\merge$ of $\Tdep_{k^2,k}[G,\Prt]$ at most $k$ times, which is followed by a single $\pathlb$ query and a single $\ins$ update. By Lemma~\ref{lem:kpath-bnd-td-ms}, this takes time $2^{\Oh(k^4)}+k^{\Oh(k^2)}\tdict$, which results in the total running time of $2^{\Oh(k^4)}+k^{\Oh(k^2)}\tdict + \Oh(k \log n)$ for the insertion procedure.

 \begin{algorithm}\label{alg:addcd}
     
     \SetKwInOut{Input}{Input}
     \SetKwInOut{Output}{Output}
    \SetKw{Raise}{raise exception:\ }
 
     \vskip 0.2cm
     
     \Input{An edge $uv$}
     \Output{The data structure is updated by inserting $uv$ to $G$}
     
     \vskip 0.1cm
     \If{$\Tdep_{k^2,k}[G,\Prt].\edge(uv)$} { \Return{}}
     $p \gets \Topt[\Upsilon].\pathlf(u,v)$ \\
     
     \If{ $p=\infty$ } 
     {
        $\Tdep_{k^2,k}[G,\Prt].\new(uv)$ \\
        $\Topt[\Upsilon].\link(u,v)$ \\
        \Return
     }
     \If{ $k\leq p<\infty$ } 
     {
          \Raise {\normalfont{long cycle detected}}   
     }
     $\pi \gets \Topt[\Upsilon].\pathf(u,v)$ 
     \tcp*{$\pi[i]$ denotes the $i$th edge of $\pi$}
     \For{ $i \gets 1$ \normalfont{\textbf{to}} $p-1$ }
     { 
        \If{ $\Tdep_{k^2,k}[G,\Prt].\same(\pi[i],\pi[i+1])$ }
        {
           $\Tdep_{k^2,k}[G,\Prt].\merge(\pi[i],\pi[i+1])$\\
           \If{\normalfont{exception is caught}}
           {
           roll back all the changes \\
           \Raise {\normalfont{long cycle detected}}
           }
        }
     
     }
     \If{$\Tdep_{k^2,k}[G,\Prt].\pathlb(u,v,\pi[1],\pi[p-1])$}
     {
       roll back all the changes \\
       \Raise {\normalfont{long cycle detected}}
     }
     $\Tdep_{k^2,k}[G,\Prt].\ins(u,v,\pi[1],\pi[p-1])$\\
     
     \caption{Edge insertion algorithm}
 \end{algorithm}
 
\paragraph{Remove.} We now move on to describing the method $\rem(uv)$ for removing an edge $uv$ from $G$ (see Algorithm~\ref{alg:remcd} for the pseudocode). Again, we start by checking whether $uv$ is present in $G$, as otherwise there is nothing to do. Then we check whether $uv$ is a bridge using method $\bridge(uv)$. If so, we remove it from $G$ and update the forest $\Upsilon$ appropriately. 

Hence, from now on we may assume that $uv$ is not a bridge. Therefore, in the biconnected component $H$ of $G$ that contains $uv$ there is a $u$-to-$v$ path that does not use the edge $uv$. Observe that any such path cannot have more than $k-1$ vertices, for otherwise there would a cycle of length at least $k$ in $G$, violating the invariant assumptions. We find such a path $\pi$ by applying the $\pathub(i,u,v,uv,uv)$ method for all $i$ from $3$ to $k-1$. 
Observe that if we now remove $uv$ from $H$, then all the cut-vertices in the obtained graph $H-uv$ lie on $\pi$. Hence, we may carry out the removal of $uv$ as follows. 

First, we need to update $\Topt[\Upsilon]$ in case $uv \in E(\Upsilon)$. For this, we first remove $uv$ from $E(\Upsilon)$ by cutting $\Upsilon$ using $\Topt[\Upsilon].\cut(uv)$ method. We then find a substitution edge on $\pi$, that is, an edge of $\pi$ that connects two different trees of $\Upsilon$, and we insert it into $\Upsilon$ in place of $uv$.

Second, we remove $uv$ from $H$ by applying $\Tdep_{k^2,k}[G,\Prt].\rem(uv)$. This way $uv$ is removed from the part $H$ of $\Prt$, hence $H$ may break into several biconnected component $H_1,\ldots,H_\ell$. What remains to do is to update $\Prt$ by actually splitting $H$ into $H_1,\ldots,H_\ell$, so that we maintain the invariant that after the method is completed, $\Prt$ is the partition of the graph into biconnected components. Note here that each of the biconnected components $H_1,\ldots,H_\ell$ intersects $\pi$ on a non-empty subpath, so we may suppose that $\pi$ is split into subpaths $\pi_1,\ldots,\pi_\ell$ so that $\pi_i$ is the intersection of $H_i$ with $\pi$. 

We now execute the split as follows. We iterate through the consecutive edges on $\pi$ and whenever we find two consecutive edges $xy$ and $yz$ such that $\art(xy,yz)$ yields $\true$, we should apply the $\splitt(xy,yz)$ method to appropriately split two parts of $\Prt$. There is a technical caveat here in that we cannot do it immediately, because the $\splitt$ method assumes that both parts resulting from the split should be biconnected. This, however, can be easily fixed by temporarily adding the edge $yv$ before the split, and removing it after the split. It is straightforward to see that thus, the prerequisites of $\splitt$ are met. Moreover, the temporary insertion of the edge $yv$ could not increase the treedepth above $k^2$, because the obtained graph is a minor of the original biconnected component $H$.

This concludes the description of the method,
As for the running time, we use $\Oh(k)$ operations in the data structure $\Topt[\Upsilon]$, each taking $\Oh(\log n)$ time, and $\Oh(k)$ operations in the data structure $\Tdep_{k^2,k}[G,\Prt]$, each taking at most $2^{\Oh(k^4)}+k^{\Oh(k^2)}\cdot \tau$ time. Hence, the total running time of the method is $2^{\Oh(k^4)}+k^{\Oh(k^2)} \tdict + \Oh(k \log n)$.

 \begin{algorithm}\label{alg:remcd}
     
     \SetKwInOut{Input}{Input}
     \SetKwInOut{Output}{Output}
    \SetKw{Break}{break}
 
     \vskip 0.2cm
     
     \Input{An edge $uv$}
     \Output{The data structure is updated by removing $uv$ from $G$}
     
     \vskip 0.1cm

     \If{{\normalfont{\textbf{not}}} $\Tdep_{k^2,k}[G,\Prt].\edge(uv)$}
     {
     \Return
     }
     
     \If{$\Tdep_{k^2,k}[G,\Prt].\bridge(uv)$} 
     {
        $\Tdep_{k^2,k}[G,\Prt].\destroy(uv)$ \\
        $\Topt[\Upsilon].\cut(u,v)$ \\
        \Return
     }
     
     \For{$p \gets 3$ \normalfont{\textbf{to}} $k-1$}
     {
       $\pi \gets \Tdep_{k^2,k}[G,\Prt].\pathub(p,u,v,uv,uv)$\\
       \If{$\pi \neq \nil$}
       {\Break}
     }
     $\Tdep_{k^2,k}[G,\Prt].\rem(uv)$ \\
     
     \If{$uv \in E(\Upsilon)$ } 
        {
        $\Topt[\Upsilon].\cut(uv)$ \;
        \For{$i\gets 1$ \normalfont{\textbf{to}} $p-1$}
        {\If{$\Topt[\Upsilon].\link(\pi[i])$}{\Break}}
        }
     
     \For{$i \gets 1$ \normalfont{\textbf{to}} $p-1$}
     {
        \If{$\Tdep_{k^2,k}[G,\Prt].\art(\pi[i],\pi[i+1])$}
        {   
            $y\gets$ the common endpoint of $\pi[i]$ and $\pi[i+1]$\\
            \If{\normalfont{\textbf{not}} $\Tdep_{k^2,k}[G,\Prt].\edge(y,v)$} {
               $\Tdep_{k^2,k}[G,\Prt].\ins(y,v,\pi[i],\pi[p-1])$\\
               $a \gets \true$\\
            }
            $\Tdep_{k^2,k}[G,\Prt].\splitt(\pi[i],\pi[i+1])$\\
            \If{ $a$ } {
            $\Tdep_{k^2,k}[G,\Prt].\rem(yv)$\\
            }
        }
     }
     
     \caption{Edge removal algorithm}
 \end{algorithm}
\end{proof}

\section{Postponing insertions and the main results}\label{sec:kpath}

In this section we combine the results from the previous section, in particular Lemmas~\ref{lem:kpath-bnd-td} and~\ref{lem:kcycle-restricted}, with the technique of postponing insertions in order to prove our main results: data structure for the
dynamic \LPath and \LCycle problems.

\paragraph*{Postponing insertions.}
We now present a generic technique for turning data structures working under structural restrictions into general ones, at the cost of introducing amortization and using a dictionary.
As we mentioned, the idea is not new: it was used by Eppstein et al.~\cite{EppsteinGIS96} for planarity testing.
We formulate the technique in an abstract way, because apart from being useful in our setting, it also applies to some other problems discussed by Alman et al.~\cite{AlmanMW17}.

Suppose $\Univ$ is some universe. A family of subsets $\Ff\subseteq 2^{\Univ}$ is {\em{downward closed}} if for all $E \subseteq F\subseteq \Univ$, $F\in \Ff$ entails $E\in \Ff$.
A  data structure $\Dt$  %\todo{MiPi: I removed the $\Dt[U,\Ff]$ notation. I found it confusing, as in the previous section we used the stuff in brackets to say what is stored inside the structure, and not the things that are fixed. So it should be rather $\Dt[X]$ if we follow this convention.} 
{\em{strongly supports $\Ff$ membership}} if $\Dt$ maintains a subset $X$ of $\Univ$ under $\ins(x)$ and $\rem(x)$ operations (just as in the dictionary, but without the associated records), 
and in addition to this it offers a boolean query $\member()$ that checks whether $X\in \Ff$.
We also consider the following weaker notion: a  data structure $\Dt$ {\em{weakly supports $\Ff$ membership}} if, again, it implements $\ins(x)$ and $\rem(x)$ operations on a dynamic subset $X\subseteq \Univ$, 
but works under the restriction that we always have $X\in \Ff$. Whenever performing an $\ins(\cdot)$ operation would violate the invariant $X\in \Ff$,
the data structure should detect this and refuse to carry out the operation.

The following lemma shows that data structures supporting weak $\Ff$ membership can be turned into ones supporting strong $\Ff$ membership at the cost of introducing amortization and allowing access to a dictionary $\Lt$ on $\Univ$.
The proof essentially repeats the same argument as~\cite[Corollary~1]{EppsteinGIS96}.

\begin{lemma}\label{lem:queue}
Suppose $\Univ$ is a universe and we have access to a dictionary $\Lt$ on $\Univ$.
Let $\Ff\subseteq 2^{\Univ}$ be downward closed and
suppose there is a  data structure $\Dt$ that weakly supports $\Ff$ membership.
Then there is a data structure $\Dt'$ that strongly supports $\Ff$ membership, where each $\member(\cdot)$ query takes $\Oh(1)$ time 
and each update uses amortized $\Oh(1)$ time and amortized $\Oh(1)$ calls to operations on $\Lt$ and~$\Dt$.
\end{lemma}
\begin{proof}
  $\Dt'$ maintains (a copy of) $\Dt$ and an additional queue $Q$, in which $\Dt'$ stores elements whose insertion is {\em{postponed}}.
 We will denote the set stored in $\Dt'$ by $X$, while the sets stored in (the copy of) $\Dt$ and in $Q$ are $X_{\Dt}$ and $X_{Q}$, respectively.
 We maintain the invariant that $X$ is the disjoint union of $X_{\Dt}$ and $X_Q$.
 
 The queue $Q$ is implemented as a doubly linked list. 
 In addition to the above, we maintain a dictionary $\Lt$ that stores $X$.
 In $\Lt$, the record associated with each $x\in X$ is either a pointer that points to the list element representing $x$ in $Q$, in case $x\in X_Q$, or a marker $\bot$ in case $x\in X_{\Dt}$. 
 Thus, given $x\in \Univ$, we may use the $\lookup(x)$ operation in $\Lt$ to check whether $x$ belongs to $X_{\Dt}$ or $X_Q$ and, if the latter holds, access the corresponding list element on $Q$.
 In the following, whenever we insert an element to $\Dt$ or $Q$, we insert it to $\Lt$ as well (possibly together with a pointer to list element). Same for removals.
 
 We maintain the following invariant:
 \begin{quote}
  $(\star)$\quad If $Q$ is non-empty and $x$ is the front element of $Q$, then $X_{\Dt}\cup \{ x \}\notin \Ff$.
 \end{quote}
 Note that since $\Dt$ stores $X_{\Dt}$, we obviously have $X_{\Dt}\in \Ff$. Hence if $Q$ is empty, then $X_{\Dt}=X\in \Ff$.
 On the other hand, if $Q$ is non-empty, then invariant~$(\star)$ together with upward-closedness of $\Ff$ implies $X\notin \Ff$.
 Thus, the $\member()$ query amounts to checking whether $Q$ is empty, which can be done in constant time.
 
 We now explain how updating $\Dt'$ works, starting with the $\ins(x)$ operation on $\Dt'$.
 By applying $\lookup(x)$ in $\Lt$, we may further assume that $x\notin X$.
 If $Q$ is not empty, then we push $x$ to the back of $Q$.
 Otherwise, we try to add $x$ to $X_{\Dt}$ by applying $\ins(x)$ on $\Dt$. 
 If this operation succeeds (i.e. $X_{\Dt}\cup \{x\}\in \Ff$), then there is nothing more to do. 
 Otherwise, if $\Dt$ refuses to insert the element $x$, then we have a situation where $X_Q=\emptyset$, $X_{\Dt} \in \Ff$, but $X_{\Dt} \cup \{x\} \notin \Ff$.
 Hence, we push $x$ to the back of $Q$; note that invariant~$(\star)$ is thus satisfied, as $x$ becomes the only element in $Q$.
 This concludes the implementation of $\ins(x)$ in $\Dt'$.
  
 We now move to the $\rem(x)$ operation in $\Dt'$.
 First, we apply $\lookup(x)$ operation offered by $\Lt$. If $x \notin X$, then there is nothing to do.
 Otherwise, we have two cases: either $x\in X_{\Dt}$ or $x\in X_Q$.
 
 If $x\in X_Q$, then we remove $x$ from $Q$; recall here that $\lookup(x)$ provided us with the pointer to the corresponding list element, so this can be done in constant time.
 However, at this point invariant~$(\star)$ might have ceased to hold.
 Hence, we apply the $\flush()$ operation, defined as follows. 
 We iteratively take the front element $x$ from $Q$ and try to insert it to $\Dt$ by applying $\ins(x)$ on $\Dt$.
 If $x$ gets successfully inserted into $\Dt$, we remove $x$ from $Q$ and proceed with the iteration.
 Otherwise, if $\Dt$ refuses to insert $x$, we break the iteration.
 Thus, the iteration stops when either~$Q$ becomes empty, or the first element $x$ of $Q$ satisfies $X_{\Dt}\cup \{x\}\notin \Ff$. So invariant~$(\star)$ is restored.
 
 We are left with the case when $x\in X_{\Dt}$. In this case we remove $x$ from $X_{\Dt}$ by calling $\rem(x)$ on~$\Dt$.
 Again, as this might have broken invariant~$(\star)$; we restore it by calling the $\flush()$ operation.
 This concludes the implementation of $\rem(x)$ in $\Dt'$.
 
 We are left with discussing the complexity.
 Observe that each operation $\ins(\cdot)$ uses $\Oh(1)$ operations on $\Dt$ and $\Oh(1)$ operations on $\Lt$.
 Similarly for $\rem(\cdot)$, except that the application of $\flush()$ may perform an unbounded number of moves of elements from $X_Q$ to $X_{\Dt}$, each involving $\Oh(1)$ operations on $\Dt$ and $\Lt$.
 However, each element inserted to $\Dt'$ is moved from $X_Q$ to $X_{\Dt}$ by $\flush()$ at most once, so the operations used for successful moves from $X_Q$ to $X_{\Dt}$ can be charged to $\ins(\cdot)$ previously applied on $\Dt'$.
 It follows that in the amortized sense, each operation in $\Dt'$ uses $\Oh(1)$ time and $\Oh(1)$ operations on $\Dt$ and $\Lt$.
\end{proof}

We remark that Lemma~\ref{lem:queue} can be applied to two problems considered by Alman et al.~\cite{AlmanMW17}: {\sc{Edge Clique Cover}} and {\sc{Point Line Cover}}.
In the first problem, given a graph $G$ and parameter $k$, we ask whether the edges of the $G$ can be covered by at most $k$ cliques in $G$.
In the second problem, given a set $S$ of points in the plane and parameter $k$, we ask whether all these points can be covered using at most $k$ lines.
Alman et al.~\cite{AlmanMW17} gave data structures for the dynamic variants of these problems, however working under the promise that there is always a solution of size at most $g(k)$.
They achieved: update time $\Oh(4^{g(k)})$ and query time $2^{2^{\Oh(k)}}+\Oh(16^{g(k)})$ for {\sc{Edge Clique Cover}}; and
update time $\Oh(g(k)^3)$ and query time $\Oh(g(k)^{g(k)+2})$ for {\sc{Point Line Cover}}.
It was left open whether the assumption about the promise can be lifted.
By combining the data structures of Alman et al.~\cite{AlmanMW17} for $g(k)=k+1$ with Lemma~\ref{lem:queue}, we obtain data structures that achieve:
amortized update time $\Oh(4^k)$ and query time $2^{2^{\Oh(k)}}$ for {\sc{Edge Clique Cover}}; and amortized update time $\Oh(k^3)$ and query time $2^{\Oh(k\log k)}$ for {\sc{Point Line Cover}}.
This assumes access to a dictionary on the edges of the graph in the case of {\sc{Edge Clique Cover}}, and on the point set in the case of {\sc{Point Line Cover}}.
%\todo[inline]{Y: Shouldn't we add this also to the ``our results'' section in the intro?!\\ MiPi: Done, added this to Our techniques.}

\newcommand{\Qq}{\mathcal{Q}}

\paragraph*{Deriving the results.} We now gather all our tools to prove the main results of this work.

\begin{theorem}[Main result for \LPath]\label{thm:main-formal}
 Let $k$ be a fixed parameter. Suppose we are given a dynamic graph $G$ on $n$ vertices, updated by edge insertions and removals, and we have access to a dictionary on the edges of $G$ where the dictionary operations have amortized running time bounded by $\tdict$. 
 Then there is a data structure that, under such updates,
 maintains whether $G$ contains a simple path on $k$ vertices with amortized update time $2^{\Oh(k^2)}+\Oh(\tdict)$. The data structure uses $(n \cdot 2^{\Oh(k \log k)} + \mdict)$ memory.
 The initialization of the data structure for an edgeless graph on $n$ vertices takes $\Oh(n)$ time.
\end{theorem}
\begin{proof}
 Let $V$ be the invariant vertex set of $G$ and let $\Univ\coloneqq \binom{V}{2}$.
 We note that a dynamic graph with vertex set $V$ can be equivalently treated as a dynamic subset of $\Univ$.
 Let then $\Qq_k\subseteq 2^{\Univ}$ comprise all sets $F\subseteq \Univ$ such that the graph $(V,F)$ does {\em{not}} contain a simple path on $k$ vertices.
 Similarly, let $\Tt_k\subseteq 2^{\Univ}$ comprise all sets $F\subseteq \Univ$ such that the graph $(V,F)$ has treedepth smaller than $k$.
 Note that both $\Qq_k$ and $\Tt_k$ are downward closed and, by Lemma~\ref{lem:path-td}, we have $\Qq_k\subseteq \Tt_k$.
  
 By Lemma~\ref{lem:kpath-bnd-td},
 there is a data structure $\Dt$ that weakly supports $\Tt_k$ membership with update time $2^{\Oh(k^2)}$, and moreover offers $\Qq_k$ membership queries in constant time.
 By Lemma~\ref{lem:queue}, we can now turn $\Dt$ into a data structure $\Dt'$ that strongly supports $\Tt_k$ membership, where each update takes amortized time $2^{\Oh(k^2)}$ and uses amortized $\Oh(1)$ operations on the dictionary over $\Univ$; the latter ones take amortized $\Oh(\tdict)$ time.
 Now, we can easily implement $\Qq_k$ membership queries in $\Dt'$ as follows:
 if the currently maintained set $X\subseteq \Univ$ does not belong to $\Tt_k$, then it also does not belong to $\Qq_k$, 
 and otherwise we may simply query the data structure $\Dt$ that is maintained within $\Dt'$ (see the proof of Lemma~\ref{lem:queue}).
 Finally, observe that strongly supporting $\Qq_k$ membership is equivalent to the requirement requested in the theorem statement.
\end{proof}

The same reasoning as presented in the proof of Theorem~\ref{thm:main-formal}, but with Lemma~\ref{lem:kpath-bnd-td} replaced with Lemma~\ref{lem:kcycle-restricted}, yields the following. We note here that as mentioned in Section~\ref{sec:prelims}, if in Lemma~\ref{lem:kcycle-restricted} we assume {\em{amortized}} time $\tdict$ per dictionary operation instead of worst-case, then we obtain the same running time guarantees, but for amortized complexity.

\begin{theorem}[Main result for \LCycle]\label{thm:main-formal-cycle}
 Let $k$ be a fixed parameter. Suppose we are given a dynamic graph $G$ on $n$ vertices, updated by edge insertions and removals, and we have access to a dictionary on the edges of $G$, which uses $\mdict$ memory and where dictionary operations have amortized bounded by $\tdict$. 
 Then there is a data structure that, under such updates,
 maintains whether $G$ contains a simple cycle on at least $k$ vertices with amortized update time $2^{\Oh(k^4)}+k^{\Oh(k^2)}\tdict+\Oh(k\log n)$. The data structure takes $(n \cdot 2^{\Oh(k^2 \log k)}+ \mdict)$ memory.
 The initialization of the data structure for an edgeless graph on $n$ vertices takes $\Oh(n)$ time.
\end{theorem}

Theorems~\ref{thm:main-formal} and~\ref{thm:main-formal-cycle} add up to a proof of Theorem~\ref{thm:main-intro}.

%%% Local Variables:
%%% mode: latex
%%% TeX-master: "main"
%%% End:

\section{Lower bounds}\label{sec:lower-bounds}

In this section we prove lower bounds for running times of operations of dynamic data structures for detecting cycles and paths.
We observe that there is a constant $\epsilon > 0$ such that any dynamic data structure for detecting cycles containing exactly $k$~vertices or for detecting paths containing at most $k$~vertices between two pre-specified vertices requires $\Omega(m^\epsilon)$ update time or $\Omega(m^\epsilon)$ query time, where $m$ is the number of edges currently present.
These lower bounds apply to the word-RAM computation model and are conditional on the so-called triangle and 3SUM conjectures (see below) but they work even for incremental data structures, i.e., structures that are required to support edge insertions but not edge removals.

We also show an unconditional lower bound for detecting cycles of length at least~$k$ in the cell-probe model.
The lower bound shows that the edge insertion and removal operations or the cycle detection operation need to take $\Omega(\log n)$ time.
This lower bound is obtained by a reduction from the dynamic connectivity problem.

\paragraph{Conditional lower bounds.}
We first state the conjectures that underlie our lower bounds.
\begin{conjecture}[Triangle Conjecture]
  Consider the word-RAM model of computation with word-length $\Oh(\log b)$ for inputs of length~$b$.
  There exists a constant $\epsilon > 0$ such that each algorithm running in this model that takes as input an undirected graph $G$ and correctly reports whether $G$ contains a triangle, has running time $\Omega(m^{1 + \epsilon})$, where $m$ is the number of edges in~$G$.
\end{conjecture}
The Triangle Conjecture has been used to derive many previous lower bounds, see \cite{abboud_popular_2014,AlmanMW17} for examples.

The second conjecture was introduced by Alman et al.~\cite{AlmanMW17}.
It refers to the following type of reachability oracle.
Let $\ell$ be an integer.
We say that a directed acyclic graph $G$ is \emph{$\ell$-layered} if it admits a vertex partition into $\ell$ sets $L_1, L_2, \ldots, L_\ell$ such that each edge~$(u, v)$ in $E(G)$ has endpoints in two consecutive, increasing layers, that is, there is an integer $i \in [\ell - 1]$ such that $u \in L_i$ and $v \in L_{i + 1}$.
\newcommand{\reachable}{\ensuremath{\mathtt{reachable}}}
An \emph{$\ell$-layered reachability oracle ($\ell$LRO)} is an algorithm that receives an $\ell$-layered directed acyclic graph~$G$, and computes a data structure that supports an operation~$\reachable(\cdot, \cdot)$.
This operation receives two vertices~$u, v$ such that $u \in L_1$ and $v \in L_\ell$, and answers whether in $G$ there is a directed path from $u$ to $v$.
The \emph{preprocessing time} of the reachability oracle is the time needed to compute the data structure, while the \emph{query time} is the time needed to carry out the \reachable{} query.
\begin{conjecture}[$\ell$-layered reachability oracle ($\ell$LRO) Conjecture]
  Consider the word-RAM model of computation with word-length $O(\log b)$ for inputs of length~$b$.
  There exists a constant $\epsilon > 0$ such that each $\ell$-layered reachability oracle in this model has preprocessing time $\Omega(m^{1 + \epsilon})$ or query time $\Omega(m^\epsilon)$.
\end{conjecture}
As Alman et al.~\cite{AlmanMW17} showed, the 3LRO Conjecture is implied by both the Triangle Conjecture and the well-known 3SUM Conjecture~\cite{abboud_popular_2014,patrascu_polynomial_2010}\footnote{In the 3SUM problem we are given a set~$S$ of integers and we want to decide whether there are three integers $a, b, c \in S$ such that $a + b = c$.
  The 3SUM Conjecture states that each algorithm solving 3SUM takes $\Omega(n^{2 - o(1)})$ time, where $n = |S|$.} and thus is possibly a weaker assumption.
For a treatment of the extensive recent usage of the Triangle and 3SUM Conjectures, see  Abboud and Williams~\cite{abboud_popular_2014}.
% Observe furthermore that for each $i \in \N$, $i \geq 3$, the $\ell$LRO conjecture implies the $(\ell + 1)$LRO conjecture (insert a new layer $L_{\ell + 1}$ which is a copy of $L_\ell$ and add the edge $(u, u')$ for each vertex $u \in L_\ell$ and its copy $u' \in L_{\ell + 1}$).

An \emph{incremental graph structure} is a data structure that maintains an undirected graph~$G$ on a fixed set of vertices and supports edge insertions.
If it also supports edge removals, we call it \emph{fully dynamic graph structure}.
We show lower bounds for incremental and fully dynamic graph structures that are initialized with an additional parameter~$k$ and support the following queries:
\newcommand{\exactcycle}[1]{\ensuremath{(=#1)\text{-}\mathtt{cycle}()}}
\newcommand{\leqpathst}[1]{\ensuremath{(\leq #1)\text{-}\mathtt{path}()}}
\begin{itemize}[nosep]
\item \exactcycle{k} -- which returns whether $G$ contains a (simple) cycle on exactly $k$ vertices; and
\item \leqpathst{k} -- which returns whether $G$ contains a (simple) path containing at most $k$ vertices between two vertices~$s$ and~$t$, which have been specified when initializing the data structure.
\end{itemize}
The time needed for inserting one edge is called the \emph{update time} and the time needed for the above queries is called the \emph{query time}.
\begin{theorem}\label{thm:cond-lb}
  Consider the word-RAM model of computation with word-length $O(\log b)$ for inputs of length~$b$.
  In this model,%
  \begin{enumerate}[nosep,label=(\roman*)]
  \item unless the Triangle Conjecture fails, there is a constant $\epsilon > 0$ such that each incremental graph structure that supports \exactcycle{3} requires $\Omega(m^\epsilon)$ update time or $\Omega(m^{1 + \epsilon})$ query time, or requires $\Omega(n^{1+\epsilon})$ time for initialization on an edgeless graph; and
  \item unless the 3LRO Conjecture fails, there is a constant $\epsilon > 0$ such that each fully dynamic graph structure that supports \exactcycle{5} or \leqpathst{5} requires $\Omega(m^\epsilon)$ update time or $\Omega(m^{ \epsilon})$ query time, or requires $\Omega(n^{1+\epsilon})$ time for initialization on an edgeless graph.
  \end{enumerate}
  Here, $n$ and $m$ respectively denote the number of vertices and edges present in the graph.
\end{theorem}
\begin{proof}
  (i): Assume that there exists an incremental graph structure~$\Dt$ that supports \exactcycle{3} and has update time~$\Oh(m^{o(1)})$, query time~$\Oh(m^{1 + o(1)})$, and initialization time~$\Oh(n^{1+o(1)})$.
  To decide whether a given graph $G$ contains a triangle, initialize $\Dt$ in $\Oh(n^{1+o(1)})$ time, insert all edges of $G$ into $\Dt$ -- which takes $\Oh(m^{1 + o(1)})$ time -- and then call \exactcycle{3}, which also takes $\Oh(m^{1 + o(1)})$ time; this is clearly correct.
  Hence the Triangle Conjecture fails.

  (ii): Assume that there exists a fully dynamic graph structure~$\Dt$ that supports \leqpathst{5} and has update time~$\Oh(m^{o(1)})$ and query time~$\Oh(m^{o(1)})$, and initialization time~$\Oh(n^{1+o(1)})$.
  Construct a $3$LRO as follows.
  Let $G$ be the input (directed) graph and $L_1, L_2, L_3$ its layers.
  Initialize $\Dt$ with the vertices of $G$ and two additional new vertices $s$ and $t$.
  Insert all the edges of $G$ into $\Dt$ one by one, forgetting their orientation.
  Take the resulting $\Dt$ to be the data structure of the 3LRO.
  Thus, the preprocessing time is $\Oh(m^{1 + o(1)})$.
  A query of the 3LRO is handled as follows.
  On input of $u \in L_1$ and $v \in L_3$, (a) insert the edges $su$ and $vt$ into $\Dt$, (b) call \leqpathst{5}, and (c) remove the edges $su$ and $vt$ from $\Dt$.
  Since all the edges in $G$ connect vertices from consecutive layers, while $s$ and $t$ are made adjacent only to $u\in L_1$ and $v\in V_3$, respectively, the answer to \leqpathst{5} will be positive if and only if in $G$ there is a directed path from $u$ to~$v$.
  Clearly, a query takes $\Oh(m^{o(1)})$ time and thus the 3LRO Conjecture fails.

  If $\Dt$ instead supports \exactcycle{5}, then proceed in the same way for preprocessing, that is, initialize $\Dt$ with the vertices of $G$ and two additional new vertices $s$ and~$t$, and insert all edges of $G$ into $\Dt$ while forgetting their orientation.
  However, during the preprocessing insert furthermore the edge $st$ into $\Dt$.
  For the query, proceed analogously: On input of $u \in L_1$ and $v \in L_3$, (a) insert the edges $su$ and $vt$ into $\Dt$, (b) call \exactcycle{5}, and (c) remove the edges $su$ and $vt$ from $\Dt$. Let $\wh{G}$ be the undirected graph that is queried in (b).
  Clearly, if in $G$ there is a directed path from $u$ to $v$, then in $\wh{G}$ there is a cycle on five vertices, and so the query will answer positively.
  On the other hand, if the query answers positively, then  the corresponding five-vertex cycle $C$ in $\wh{G}$ needs to involve $s$ and $t$. Indeed, since $G$ consists of three layers and edges in $G$ are only between two adjacent layers, each cycle in $\wh{G}$ whose vertex set is fully contained in $V(G)$ has even length.
  Therefore, $C$ needs to contain $s$, $t$, and the edge $st$. Since besides each other, $s$ is only adjacent to $u$ and $t$ is only adjacent to $v$, the remaining three vertices of $C$ must form a three-vertex path from $u$ to $v$ in $G$, as required.
  Since the preprocessing takes $\Oh(m^{1 + o(1)})$ time and the query $\Oh(m^{o(1)})$ time, the 3LRO fails.
  \end{proof}

\paragraph{Unconditional lower bound.}
\newcommand{\connected}{\mathtt{connected}}
\newcommand{\geqcycle}[1]{\ensuremath{(\geq#1)\text{-}\mathtt{cycle}}}

We next show that a dynamic graph structure for a graph~$G$ (see above for the definition) that is initialized with an additional parameter~$k$ and that supports a query \geqcycle{k} --- which tests whether there is a (simple) cycle of length at least~$k$ in $G$ --- can be used to maintain connectivity information in~$G$.
Hence, lower bounds for dynamic connectivity apply to fully dynamic graph structures that support \geqcycle{k}.
Formally, a \emph{dynamic connectivity structure} is a fully dynamic graph structure for a graph $G$ that supports the query $\connected(\cdot, \cdot)$, which receives as input two vertices $u, v \in V(G)$ and answers positively if and only if $u$ and $v$ are in the same connected component.
Demaine and P\u{a}tra\c{s}cu~\cite{10.1145/1007352.1007435} proved the following.

\begin{theorem}[\cite{10.1145/1007352.1007435}, Theorem~6]
  Consider the cell-probe model of computation.
  Every dynamic connectivity structure for a graph~$G$ in this model has either $\Omega(\log n)$ amortized update time or $\Omega(\log n)$ amortized query time, where $n = |V(G)|$.
  This holds even if $G$ is promised to remain a disjoint union of paths at all times.
\end{theorem}

The following lemma shows that this lower bound transfers to detecting cycles.

\begin{lemma}\label{lem:conn2kcyc}
  Suppose that there is a fully dynamic graph structure $\Dt$ for a graph~$G$ which supports the \geqcycle{$3$} query so that the amortized update and query time are upper bounded by $t(n)$, where $n = |V(G)|$.
  Then there is a dynamic connectivity data structure $\Dt'$ which maintains a graph which is promised to remain a forest; the data structure supports inserting and removing edges and the $\connected(u,v)$ query in amortized $\Oh(t(n))$ time per operation.
\end{lemma}
\begin{proof}
  We show how $\Dt'$ can be implemented using $\Dt$.
  We are promised that the dynamic graph given to $\Dt'$ is a forest at all times.
  Structure $\Dt'$ initializes structure $\Dt$, and uses it in the following way.
  When an edge $uv$ is added or removed to or from $G$, $\Dt'$ inserts or removes $uv$ to or from $\Dt$ accordingly.
  To carry out a query $\connected(u,v)$, structure $\Dt'$ calls on $\Dt$ the following sequence of instructions: inserting the edge $uv$, a call to $\geqcycle{3}$, and removing the edge $uv$.
  If the answer to $\geqcycle{3}$ in $\Dt$ was positive, $\Dt'$ answers positively to query $\connected(u,v)$, otherwise the answer is negative.
\end{proof}

Thus, we have the following.
\begin{corollary}\label{cor:3-cyc-lb}
  Consider the cell-probe model of computation.
  Each dynamic graph structure for a graph~$G$ in this model that supports the \geqcycle{$3$} query has either $\Omega(\log n)$ amortized update time or $\Omega(\log n)$ amortized query time, where $n = |V(G)|$.
\end{corollary}

%%% Local Variables:
%%% mode: latex
%%% TeX-master: "main"
%%% End:

\section{Conclusions}\label{sec:conclusions}

In this work we presented a fully dynamic data structure that for a dynamic graph $G$, promised to be of treedepth at most $d$ at all times, maintains an elimination forest of $G$ of optimum height. The data structure offers $2^{\Oh(d^2)}$ update time in the worst case. We used this result to give data structures for the dynamic variants of the \LPath and the \LCycle problems. For \LPath, the data structure offers amortized update time $2^{\Oh(k^2)}+\Oh(\tau)$, where $\tau$ is the amortized operation time in a dictionary on the edges of the maintained graph. For \LCycle, the amortized update time is $2^{\Oh(k^4)}+k^{\Oh(k^2)}\cdot \tau+\Oh(k\log n)$.

As argued in Section~\ref{sec:introduction}, the results for \LPath and \LCycle are tight as far as the dependence on $n$ is concerned. However, there is a lot of room for improvements of the parametric factor, that is, the dependency on $k$. In particular, the data structure for \LPath\ of Alman et al.~\cite{AlmanMW17} achieves a better parametric factor of~$k^{\Oh(k)}$, at the cost of allowing a factor that is polylogarithmic in $n$. Would it be possible to obtain amortized update time $k^{\Oh(k)}$, without any additional factors depending on $n$? Perhaps more interestingly, static fpt algorithms for \LPath, designed for instance using color-coding~\cite{AlonYZ95} or algebraic coding~\cite{BjorklundHKK17,Koutis08,Williams09}, achieve parametric factor $2^{\Oh(k)}$ in their running times. Can one design a data structure for dynamic \LPath with amortized update time $2^{\Oh(k)}\cdot \log^c n$ for some constant $c$, or even $2^{\Oh(k)}$? Similar questions can be also asked about \LCycle, where the $2^{\Oh(k^4)}$ factor seems even more amenable to improvements.

%for the dynamic $k$-path problem that achieves amortized update time $2^{\Oh(k^2)}$. 
%We remark that it is straightforward to modify the data structure so that within the same complexity it also maintains an example $k$-path, in case there is any.
%Also, it is easy to verify that the space usage of the data structure is $2^{\Oh(k^2)}\cdot n+\Oh(m)$.
% 
% Observe that while static fixed-parameter algorithms for $k$-path achieve parametric factor of the running time of the form $2^{\Oh(k)}$, this is the case neither for our data structure (factor $2^{\Oh(k^2)}$) 
% nor for the data structure of Alman et al.~\cite{AlmanMW17} (factor $k^{\Oh(k)}$). It would be interesting to investigate whether there is a data structure for the dynamic $k$-path problem that achieves (amortized)
% update time $2^{\Oh(k)}\cdot \log^c n$, for some constant $c$.

The $2^{\Oh(k^2)}$ update time offered by our dynamic data structure for \LPath is a direct consequence of the $2^{\Oh(d^2)}$ update time of the dynamic treedepth data structure. As we argued in Section~\ref{sec:introduction}, improving this factor is equivalent to improving the parametric factor in the complexity of the static fpt algorithm of Reidl et al.~\cite{ReidlRVS14} for computing the treedepth of a graph. However, note that for an application to the \LPath problem, we do not necessarily need to maintain an elimination forest of optimum height; a constant-factor approximation would perfectly suffice. Unfortunately, approximation algorithms for treedepth remain largely unexplored even in the static setting, which brings us to an old question (see e.g.~\cite{CzerwinskiNP19}): 
is there a constant-factor approximation algorithm for treedepth with running time $2^{o(d^2)}\cdot n^{\Oh(1)}$, where $d$ is the value of the treedepth?

Finally, we hope that our work might give some insight into the problem of maintaining an approximate tree decomposition of a dynamic graph of bounded treewidth. 
Here, even achieving polylogarithmic-time updates would be very interesting. This direction was also mentioned both by Alman et al.~\cite{AlmanMW17} and by Dvo\v{r}\'ak et al.~\cite{DvorakKT14}.

%%% Local Variables:
%%% mode: latex
%%% TeX-master: "main"
%%% End:

\bibliographystyle{abbrv}
\bibliography{references}

%\newpage

\appendix

\section{Additional material for the introduction}\label{app:intro}

\paragraph*{Other works on parameterized dynamic data structures.}
To give a broader background, we review some results on dynamic data structures for parameterized problems that are not directly relevant to the motivation of our work.

Dvo\v{r}\'ak and T\r{u}ma~\cite{DvorakT13} investigated the problem of counting occurrences (as induced subgraphs) of a fixed pattern graph $H$ in a dynamic graph $G$ that is assumed to be sparse 
(formally, always belongs to a fixed class of bounded expansion $\Cc$). They gave a data structure that maintains such a count with amortized update time $\Oh(\log^{c} n)$, where the constant $c$ 
depends both on $H$ and on the class $\Cc$. As classes of bounded treedepth have bounded expansion (see~\cite{NesetrilM12}), by taking $H$ to be a path on $k$ vertices we obtain
a data structure for the dynamic $k$-path problem in graphs of treedepth smaller than $k$ with amortized polylogarithmic update time, where the degree of the polylogarithm depends on $k$.
Note that this result is significantly weaker than the one provided by us and by Alman et al.~\cite{AlmanMW17}, though it is obtained using very different tools.
The result of Dvo\v{r}\'ak and T\r{u}ma~\cite{DvorakT13} is based on a data structure of Brodal and Fagerberg~\cite{BrodalF99} for maintaining a bounded-outdegree orientation of a graph of bounded degeneracy, which also can be considered a dynamic parameterized data structure.

The dynamic setting for parameterized vertex cover and other vertex-deletion problems was first considered by Iwata and Oka~\cite{IwataO14}. As explained in Section~\ref{sec:introduction}, this work was continued and  significantly extended by Alman et al.~\cite{AlmanMW17}.
More recent advances include dynamic kernels for hitting and packing problems in set systems with very low update times~\cite{BannachHRT19},
and work on monitoring timed automata in data streams~\cite{GrezMPPR20}. Also, Schmidt et al.~\cite{SchmidtSVZK20} investigated a combination of parameterization and the concept of $\mathsf{DynFO}$.
This setting is, however, somewhat different, as the main focus is on performing updates that can be described using simple logical formulas, and not necessarily executable efficiently in the classic sense.

\paragraph*{Comparison with Dvo\v{r}\'ak et al.~\cite{DvorakKT14}.}
We present a quick overview of the approach that Dvo\v{r}\'ak et al.~\cite{DvorakKT14} used in their dynamic treedepth data structure. While doing this, we explain how our techniques differ and improves upon this approach. In this overview, we assume familiarity with $\MSO_2$ and basic understanding of $\MSO_2$-types.

We assume the following setting. We have a dynamic graph $G$, a fixed $\MSO_2$ sentence $\varphi$, and a parameter $d$ so that the treedepth of $G$ is promised to always be upper bounded by $d$ at all times. Our goal is to maintain a fully dynamic data structure $\Dt_\varphi[F,G]$ that maintains a recursively optimal elimination forest $F$ of $G$ and an answer to the query whether $G\models \varphi$. Let $q$ be the maximum of $d$ and the quantifier rank (i.e. maximum number of nested quantifiers) in $\varphi$. Note that thus, $q$ is {\em{at least}} $d$.

With each vertex $u\in V(G)$, let us associate a graph $G_u$ defined as in Section~\ref{secov:data} or Section~\ref{sec:data-structure}: $G_u$ has vertex set $\SReach_{F,G}(u)\cup \desc_F(u)$ and edge set comprising of all the edges of $G$ with at least one endpoint in $\desc_F(u)$.
The idea is to associate with each vertex $u$ the {\em{type}} of $G_u$, which is a piece of information that concisely describes all the properties of $G_u$ needed both for the task of computing the treedepth, 
and for verifying satisfaction of $\varphi$.
In the work of Dvo\v{r}\'ak et al.~\cite{DvorakKT14}, this type is the $\MSO_2$ type of $G_u$ of quantifier rank $q$. The number of such types is bounded by a function of $q$ only, but this function is non-elementary: it is a tower of height $q$, which is not smaller than $d$.

We note that in~\cite{DvorakKT14}, 
this is presented somewhat differently. Namely, the type of $G_u$ is maintained implicitly by storing a bounded-size {\em{$S$-code}}: a representantive subgraph of $G_u$ having the same $\MSO_2$-type of quantifier rank $q$ as $G_u$, obtained by trimming superfluous subtrees in $F$.

Now, basic compositionality and idempotence properties of $\MSO_2$ imply that in order to compute the type of $G_u$, it suffices to know the multiset of types of graphs $G_v$ for all children $v$ of $u$ in $F$.
Moreover, there is a threshold $\tau$ depending on $d$ and $\varphi$ such that within this multiset, each type appearing more than $\tau$ times can be treated as if it appeared exactly $\tau$ times.
This can be related to the discussion of compositionality and idempotence in Sections~\ref{secov:data} and~\ref{sec:data-structure}.
Thus, for every vertex $u$ one can compute the type of $G_u$ from the types associated with the children in constant time, assuming we know the multiplicity of each type among the children up to the threshold $\tau$. This applies even though the number of children is unbounded. Intuitively, this allows efficient recomputation of the types upon modifications of the forest $F$, through a bucketing approach similar to that presented in Sections~\ref{secov:data} and~\ref{sec:data-structure}.

Thus, the design of the data structure itself in~\cite{DvorakKT14} is similar to ours: every vertex remembers all its children, but these children are partitioned into buckets (represented in~\cite{DvorakKT14} using {\em{merge nodes}}) according to their types. Note that thus, the number of buckets per vertex is bounded by a non-elementary function of $q\geq d$. The basic idea of achieving update time independent of $n$ is the same: during modifications of the elimination forest, we operate on whole buckets. Thus, re-attaching a whole bucket at a different place of the elimination forest can be done using a single operation in constant time.

The implementation of updates in~\cite{DvorakKT14} is, however, very different to ours and works roughly as follows. Suppose $F$ is a tree for simplicity. First, one finds a {\em{candidate new root}}: a vertex that may be the new root of an optimal elimination tree after the update. 
Now comes the main trick: being a candidate root can be expressed by a (quite complicated) $\MSO_2$ formula of quantifier rank $d$, 
hence we can use the types of rank $q\geq d$, stored in the data structure before the update, we can locate a candidate root.
Once the new root is located, we iteratively move the new root up the tree. During this process we need to fix a bounded number of subtrees, which can done by recursion, because each of the trees that needs to be fixed is of height at least one smaller. All in all, this update thus uses a fairly convoluted recursion scheme, where the number of recursive calls heavily depends on the number of types that the data structure keeps track of.

Thus, the aspect that contributes the most to the time complexity is the number of types on which the data structure relies, which directly corresponds to the number of buckets stored per vertex.
As we explained above, Dvo\v{r}\'ak et al.~\cite{DvorakKT14} use $\MSO_2$ types of a quantifier rank $q\geq d$, whose number is a tower of exponentials of height $q$. Note that the assumption that $q\geq d$ is necessary to be able to efficiently evaluate the $\MSO_2$ query locating a new root, which is needed in the update operation.

In our data structure, we are much more frugal when defining types of vertices for the purpose of maintaining a recursively optimal elimination forest. More precisely, we show that it is enough to classify each vertex $u$ according to (1) the treedepth of the subgraph induced by the descendants of $u$, including $u$; 
and (2) the set of ancestors of $u$ that are adjacent to $u$ or any of its descendants. This gives only $d\cdot 2^d$ different types. 
Moreover, we perform the update in a different way than Dvo\v{r}\'ak et al.: to construct an elimination forest $\wh{F}$ of the updated graph, 
we extract a {\em{core}} $K\subseteq V(G)$ of size $d^{\Oh(d)}$ that contains both endpoints of the updated edge,
compute an optimum elimination forest $F^K$ of $G[K]$ using the static algorithm of Reidl et al.~\cite{ReidlRVS14} in time $2^{\Oh(d^2)}$, and construct $\wh{F}$ by re-attaching parts of $F$ lying outside of $K$ to $F^K$.
While this method is conceptually simpler than the approach used in~\cite{DvorakKT14}, justifying the correctness requires a quite deep and technical dive into the combinatorics of treedepth and of elimination forests. This analysis is explained in Sections~\ref{secov:cores} and~\ref{sec:cores}.
We remark that the concept behind the construction of the cores is analogous to that behind the construction of the $S$-codes in~\cite{DvorakKT14}, but we execute it so that the size of the core is much smaller. We also use the cores in a quite different way.

As far as augmentation of the data structure with a dynamic programming procedure is concerned, this is automatic in the approach of Dvo\v{r}\'ak et al.~\cite{DvorakKT14} for $\MSO_2$-expressible problems. Namely, the data structure anyway stores all the information about $\MSO_2$-types up to quantifier rank $q$, so the answers to all boolean $\MSO_2$ queries up to this quantifier rank are explicitly maintained. This, of course, comes at the cost of a huge explosion of complexity due to maintaining a partition into types that is potentially much finer than needed for the problem we are interested in. The language of configuration schemes and its implementation in the dynamic treedepth data structure via {\em{mugs}}, which we present in Section~\ref{sec:dp}, is designed to remedy this complexity explosion. Namely, it allows one to design a dynamic programming procedure and automatically combine it with the dynamic treedepth data structure so that the overhead in the update time paid for augmentation is polynomial in the number of states.

%%% Local Variables:
%%% mode: latex
%%% TeX-master: "main"
%%% End:

%\input{pseudocodes}

\end{document}